\documentclass[a4paper,11pt,reqno]{amsart}
\usepackage[utf8]{inputenc}
\usepackage{mathrsfs}
\usepackage{dsfont}
\usepackage{hyperref}
\usepackage{amsmath}
\usepackage{amssymb}
\usepackage{amsthm}
\usepackage{amsfonts}
\usepackage{amstext}
\usepackage{amsopn}
\usepackage{amsxtra}
\usepackage{mathrsfs}
\usepackage{dsfont}
\usepackage{color}
\usepackage{tikz}
\usepackage{mathtools}
\usepackage{xparse}
\usepackage[capitalize]{cleveref}

\newtheorem{theorem}{Theorem}
\newtheorem{lemma}[theorem]{Lemma}
\newtheorem{proposition}[theorem]{Proposition}
\newtheorem{corollary}[theorem]{Corollary}
\newtheorem{remark}[theorem]{Remark}

\newtheorem{assumption}[theorem]{Assumption}

\newcommand{\R}{\mathbb{R}}

\newcommand{\N}{\mathbb{N}}
\newcommand{\bP}{\mathbb{P}}
\newcommand{\bS}{\mathbb{S}}

\newcommand{\dps}{\displaystyle}
\newcommand{\ii}{\infty}
\newcommand\1{{\ensuremath {\mathds 1} }}
\renewcommand\phi{\varphi}

\newcommand{\gS}{\mathfrak{S}}

\newcommand{\cS}{\mathcal{S}}

\newcommand{\cC}{\mathcal{C}}
\newcommand{\cU}{\mathcal{U}}

\newcommand{\cF}{\mathcal{F}}

\newcommand{\norm}[1]{ \left\| #1 \right\|}

\renewcommand{\geq}{\geqslant}
\renewcommand{\leq}{\leqslant}

\renewcommand{\tilde}{\widetilde}
\newcommand{\eps}{\varepsilon}

\newcommand{\nn}{\nonumber}
\newcommand{\rd}{\mathrm{d}}
\newcommand{\dx}{\rd x}
\newcommand{\dy}{\rd y}

\crefname{figure}{Figure}{Figures}

\DeclareMathOperator{\supp}{supp}

\DeclarePairedDelimiter\myp()
\DeclarePairedDelimiter\myt\{\}
\DeclarePairedDelimiter\myb[]
\DeclarePairedDelimiter\abs\lvert\rvert

\newcommand{\bQ}{\mathbb{Q}}
\newcommand{\id}{\, \mathrm{d}}

\providecommand\given{}
\newcommand\SetSymbol[1][]{%
	\nonscript\:#1\vert
	\allowbreak
	\nonscript\:
	\mathopen{}}
\DeclarePairedDelimiterX\Set[1]\{\}{%
	\renewcommand\given{\SetSymbol[\delimsize]}
	#1
}

\NewDocumentCommand\normt{ s o m m }{
	\IfBooleanTF{#1}{
		\norm*{#3}
	}{
		\IfNoValueTF{#2}{
			\norm{#3}
		}{
			\norm[#2]{#3}
		}
	}
	_{#4}
}

\title[Classical DFT: Universal Bounds]{Classical Density Functional Theory: Representability and Universal Bounds}

\author[M. Jex]{Michal Jex}
\address{CNRS \& CEREMADE, Universit\'e Paris-Dauphine, PSL University, 75016 Paris, France and Department of Physics, Faculty of Nuclear Sciences and Physical Engineering, Czech Technical University in Prague, B\v rehov\'a 7, 11519 Prague, Czech Republic }
\email{michal.jex@fjfi.cvut.cz}

\author[M. Lewin]{Mathieu Lewin}
\address{CNRS \& CEREMADE, Universit\'e Paris-Dauphine, PSL University, 75016 Paris, France}
\email{mathieu.lewin@math.cnrs.fr}

\author[P. S. Madsen]{Peter S. Madsen}
\address{CNRS \& CEREMADE, Universit\'e Paris-Dauphine, PSL University, 75016 Paris, France}
\email{madsen@ceremade.dauphine.fr}

\date{\today.}

\begin{document}

\begin{abstract}
We provide upper and lower bounds on the lowest free energy of a classical system at given one-particle density $\rho(x)$. We study both the canonical and grand-canonical cases, assuming the particles interact with a pair potential which decays fast enough at infinity.
\end{abstract}

\maketitle

\tableofcontents

\section{Introduction}

Density functional theory (DFT) is a powerful tool used in quantum physics  and chemistry to model quantum electrons in atoms, molecules and solids~\cite{Lieb-83b,DreGro-90,ParYan-94,EngDre-11,LewLieSei-19_ppt,CanLinLiuFri-22_ppt}. However, DFT is based on a rather general mathematical scheme and it can be applied to many other situations. This work is devoted to the rigorous study of \emph{classical DFT}, which is used for finite or infinite systems of interacting classical particles.

Classical DFT is widely employed in materials science, biophysics, chemical engineering and civil engineering~\cite{Wu-06}. It has a much lower computational cost than the more precise molecular dynamics simulations, which are limited to small systems and short times~\cite{GelGubRadSli-99,Roth-10,LanGorNei-13}. Classical DFT is typically used at interfaces between liquid-gas, liquid-liquid (in fluid mixtures), crystal-liquid and crystal-gas phases at bulk coexistence. The density is then non constant in space and varies in the interfacial region between the two phases.

The physical theory of inhomogenous fluids goes essentially back to the 60s~\cite{MorHir-61,DeDominicis-62,DomMar-64,StiBuf-62,LebPer-63}. Functional methods and their applications to the theory of the structure of bulk fluids were described in \cite{Percus-64,Stell-64}. The realization that methods developed in the quantum case by Hohenberg-Kohn-Sham~\cite{HohKoh-64,KohSha-65} could be transferred to classical fluids arose in the middle of the 70s, in particular in the works of Ebner-Saam~\cite{EbnSaaStr-76,EbnPun-76,SaaEbn-77} and Yang \textit{et al}~\cite{YanFleGib-76}. Several authors then developed approximate free-energy functionals to calculate the density profile and surface tension of the liquid-gas interface. The square-gradient approximation could later be derived rather systematically, following the important works of Hohenberg-Kohn-Sham on the gradient expansion of the uniform (quantum) electron gas. Deriving efficient functionals for the solid-liquid transition was harder and took longer~\cite{RamYus-77,YamYus-79,HayOxt-81}. Well-known references on classical DFT are the two reviews by Evans~\cite{Evans-79,Evans-92}. Other important physical references on the subject include~\cite{HanMcDon-90,PliBer-06,Baus-87,BauLut-91,Singh-91,Lowen-02,EvaOetRotKah-16,HanMcD-13}.

Rigorous works on classical DFT are rather scarse. Most of the mathematical works are about proving that one can find an external potential $V$ whose interacting equilibrium Gibbs measure has any desired given density $\rho$. This is called the \emph{inverse} or \emph{dual problem} and justifies the use of density functional methods. In quantum DFT, $V$ is called the Kohn-Sham potential and its existence is unclear in most situations. However, in the classical case, $V$ is usually well defined.

The grand-canonical 1D hard-core gas was solved exactly in a celebrated work by Percus~\cite{Percus-76}, who provided an exact expression of the external potential $V$. This was used and extended in later works~\cite{VanDavPer-89,Percus-82,Percus-82b,Percus-97}. In two famous works~\cite{ChaChaLie-84,ChaCha-84}, Chayes, Chayes and Lieb proved in a quite general setting (in particular any space dimension $d$) the existence and uniqueness of the dual potential $V$ at any positive temperature $T>0$. At $T=0$, the canonical model can be reformulated as a multi-marginal optimal transport problem~\cite{CotFriKlu-13,CotFriPas-15,Pass-15,MarGerNen-17,SeiMarGerNenGieGor-17}, where $V$ is usually called the \emph{Kantorovich potential}. Its existence and properties are known in many cases~\cite{Kellerer-84,Pascale-15,ButChaPas-18} but uniqueness usually does not hold. The grand-canonical case was studied in the recent article~\cite{MarLewNen-22_ppt}. Most of these works are based on compactness arguments and do not furnish any quantitative information on the shape of the potential $V$ in terms of the given density $\rho$. In the recent paper~\cite{JanKunTsa-22}, a novel Banach inversion theorem was used to provide an explicit formula for $V$ in terms of $\rho$ in the form of a convergent series, under the assumption that $\rho$ is small in $L^\ii(\R^d)$. This is the equivalent of the Virial expansion for uniform systems.\footnote{Recall that one can express the constant density $\rho$ of an infinite gas as a convergent series in terms of the activity $z=e^{\beta\mu}$, in the regime $z\ll1$~\cite{Ruelle}. This corresponds to placing the system in the constant external potential $V(x)=-\mu$. Since $\rho\sim_{z\to0} z$, the series is invertible and any small uniform density is therefore representable by such a uniform potential, with $\mu\sim_{\rho\to0}\beta^{-1}\log\rho$.}

In this work and the companion paper~\cite{JexLewMad-23b_ppt} we do not discuss the dual potential $V$ and instead focus on more quantitative properties of the model depending on the shape of the density $\rho$. The case of the three dimensional Coulomb interaction $w(x)=|x|^{-1}$ or, more generally long range Riesz interactions $|x|^{-s}$ with $s<d$ has been the object of several recent works~\cite{LewLieSei-19_ppt,Lewin-22}. Here we always assume that the interaction potential $w$ decays fast enough at infinity and do not discuss more complicated long range potentials such as Coulomb.

Our main goal in this paper is to show \emph{universal local bounds} on the free energy $F_T[\rho]$ at given density $\rho\in L^1(\R^d)$. By local we mean that we only use terms in the form
$$\int_{\R^d}\rho(x)^p\,\dx,\qquad \int_{\R^d}\rho(x)^q\log\rho(x)\,\dx.$$
The admissible values of $p$ and $q$ will depend on the temperature $T$ as well as on the singularity of the interaction potential $w$ at the origin, that is, how strong the particles repel each other when they get close. Such universal bounds are important in DFT. They can help finding the natural form of approximate functionals to be used for practical computations.\footnote{As an example, in the quantum case the Lieb-Oxford inequality~\cite{LieOxf-80} was used to calibrate some famous functionals such as PBE and SCAN~\cite{Perdew-91,LevPer-93,PerBurErn-96,TaoPerStaScu-03,SunPerRuz-15,SunRuzPer-15,Perdew_etal-16,PerSun-22}.} In addition, these bounds will be useful in our next work~\cite{JexLewMad-23b_ppt} where we study the local density approximation.

Deriving simple \emph{lower bounds} is usually easy, under reasonable stability assumptions on the interaction potential $w$. Obtaining \emph{upper bounds} can be much more difficult. They require constructing a good trial state, but the constraint that the density is given and must be exactly reproduced can generate important mathematical complications.

The simplest trial state is obtained by taking i.i.d.~particles, that is, a factorized $N$-particle probability $(\rho/N)^{\otimes N}$ where $N=\int_{\R^d}\rho\in\N$. Doing so provides an upper bound on the free energy in terms of mean-field theory, often called in this context the Kirkwood-Monroe functional~\cite{KirMon-41}:
\begin{equation}
\frac12\iint_{\R^d\times\R^d}w(x-y)\rho(x)\rho(y)\,\dx\,dy+T\int_{\R^d}\rho(x)\log\rho(x)\,\dx.
 \label{eq:KirkwoodMonroe}
\end{equation}
This only makes sense when the pair interaction potential $w$ is locally integrable. If $w$ is globally integrable, one can use Young's inequality and estimate the first integral by the local energy $(\int_{\R^d}w_+/2)\int_{\R^d}\rho(x)^2\,\dx$, where $w_+:=\max(w,0)$ denotes the positive part. The simplest models of classical DFT use~\eqref{eq:KirkwoodMonroe} as a basis.

In classical statistical mechanics, it is often convenient to consider potentials $w$ diverging fast enough at the origin, which helps to stabilize the system~\cite{Ruelle,Dobrushin-64,DobMin-67,Georgii-11,Rebenko-98,PetReb-07}. This divergence implies that the particles can never get too close to each other, and this requires that the trial state contains rather strong correlations. A factorized state is not appropriate and~\eqref{eq:KirkwoodMonroe} is infinite. The simplest singular interaction is of course the hard-core potential $w(x)=(+\ii)\1(|x|< r_0)$, which is simply infinite over a ball and vanishes outside.

In this paper we provide two different constructions of a correlated trial state, which give reasonable upper bounds on the classical free energy at given density, for singular interaction potentials at the origin. Our first method uses some ideas from harmonic analysis in the form of a Besicovitch-type covering lemma~\cite{Guzman}. We cover space with cubes whose size is adapted to the local value of the density, and put essentially one particle per cube, with the constraint that the cubes are far enough from each other. This method works very well in the grand-canonical setting where the number of particles is allowed to fluctuate. In order to handle the canonical ensemble, a different construction is needed. We instead use techniques from optimal transport theory developed in~\cite{ColMarStr-19}, which give a rather good bound at zero temperature, $T=0$. For $T>0$ we couple this to the Besicovitch-type covering lemma and obtain an upper bound which is not as good as the grand-canonical one.

In~\cite{JexLewMad-23b_ppt} we will study the behavior of $F_T[\rho]$ in some particular regimes and the upper universal bounds derived here will be useful. Namely we will consider the thermodynamic limit where $\rho$ is essentially constant over a large domain as well as the local density approximation when $\rho$ varies slowly over big regions. Such regimes have been recently considered for the three dimensional Coulomb potential $w(x)=|x|^{-1}$ in~\cite{LewLieSei-18,LewLieSei-20}, for more general Riesz potentials in~\cite{CotPet-19,CotPet-19b} and for a special class of positive-type interactions in~\cite{Mietzsch-20}. The methods used in these works all rely on the assumption that the potential is positive-definite, and new ideas are necessary in the general (short-range) case.

\bigskip

\noindent{\textbf{Acknowledgement.}} ML thanks Rupert L. Frank for providing him with a preliminary version of the book~\cite{FraLapWei-LT}. This project has received funding from the European Research Council (ERC) under the European Union's Horizon 2020 research and innovation programme (grant agreement MDFT No 725528 of ML). MJ also received financial support from the Ministry of Education, Youth and Sport of the Czech Republic under the Grant No.\ RVO 14000.

\section{Main results}

\subsection{Free energies at given density}
%
This subsection is mainly devoted to precisely introducing models and notation used in the paper. Our main results are stated in the next subsections.

\subsubsection{The interaction potential $w$}
For convenience, we work in $\R^d$ with a general dimension $d\geq1$. The physical cases of interest are of course $d\in\{1,2,3\}$ but the proofs are the same for all $d$, except sometimes for $d=1$. We consider systems of indistinguishable classical particles  interacting through a short-range pair potential $ w $.
Throughout the paper, we work with an interaction satisfying the following properties.

\begin{assumption}[on the short-range potential $w$]
\label{de:shortrangenew}
	Let $ w :\R^d\to\R\cup\{+\ii\} $ be an even lower semi-continuous function satisfying the following properties, for some constant $\kappa>0$:
	\begin{enumerate}
		\item $ w $ is \emph{stable}, that is,
		\begin{equation}
		\label{eq:stabilitynew}
			\sum\limits_{1 \leq j < k \leq N} w \myp{x_j - x_k} \geq - \kappa N
		\end{equation}
		for all $ N \in \N $ and $ x_1, \dotsc, x_N \in \R^d $;

		\smallskip

		\item $w$ is \emph{upper regular}, that is, there exist $ r_0 \geq 0 $, $ 0 \leq \alpha \leq \infty $ and $s>d$ such that
		\begin{equation}
			w(x) \leq \kappa\left(\1(|x|<r_0)\left(\frac{r_0}{|x|}\right)^\alpha+ \frac{1}{1 + \abs{x}^s}\right).
			\label{eq:assumption_w}
		\end{equation}
	\end{enumerate}
\end{assumption}

The lower semi-continuity of $w$ will be used later to ensure that the energy is lower semi-continuous as a function of the one-particle density (see Remark~\ref{rmk:lsc} below). In statistical mechanics, the stability condition~\eqref{eq:stabilitynew} is used to ensure the existence of the thermodynamic limit~\cite{Ruelle}. On the other hand, upper bounds of the form~\eqref{eq:assumption_w} are sometimes used to get more information on the equilibrium states~\cite{Ruelle-70}. At infinity, we assume that our potential $w$ is bounded above by $|x|^{-s}$, which is integrable since $s>d$. It could of course decay faster. On the other hand, the parameter $ \alpha $ determines the allowed repulsive strength of the interaction at the origin. If $ \alpha = 0 $, then $ w $ is everywhere bounded above, and if $ 0 < \alpha < d $, then $ w $ has at most an integrable singularity at the origin. In particular, the positive part $ w_+$ is integrable over the whole of $\R^d$ (since we are interested in upper bounds, the negative part $w_-$ will not play a role in this paper). In the case where $ \alpha \geq d $, $ w $ can have a non-integrable singularity at the origin. If $ \alpha = \infty $, then $ w $ can have a hard-core. Our convention is that $(r_0/|x|)^\alpha=(+\ii)\1(|x|<r_0)$ for $\alpha=+\ii$. When $\alpha<\ii$ we can always assume that $r_0=1$, possibly after increasing $\kappa$.

Most short range potentials of physical interest are covered by Assumption~\ref{de:shortrangenew}, including for instance the simple hard-core and the Lennard-Jones potential $w(x)=a|x|^{-12}-b|x|^{-6}$.

\subsubsection{The canonical free energy}

In this subsection we define the canonical free energy $F_T[\rho]$ at given density $\rho$.

Suppose that we have $ N $ particles in $ \R^d $, distributed according to some Borel probability measure $ \bP $ on $ \R^{dN} $. Since the particles are indistinguishable, we demand that the measure $ \bP $ is symmetric, that is,
\begin{equation*}
	\bP(A_{\sigma \myp{1}}\times\cdots\times A_{\sigma \myp{N}})
	= \bP(A_1\times\cdots\times A_N)
\end{equation*}
for any permutation $ \sigma $ of $ \Set{1, \dotsc, N} $, and any Borel sets $A_1,...,A_N\subset\R^d$. The one-body density of such a symmetric probability $ \bP $ equals $N$ times the first marginal of $\bP$, that is,
\begin{equation*}
	\rho_{\bP} = N\int_{\R^{d \myp{N-1}}} \id \bP \myp{\cdot,x_2, \dotsc, x_N},
\end{equation*}
where the integration is over $ x_2,\dotsc,x_N $. Equivalently, $\rho_\bP(A)=N\bP(A\times(\R^d)^{N-1})$ for every Borel set $A$. Note the normalization convention $ \rho_{\bP}(\R^d) = N $. For a non-symmetric probability $\bP$ we define $\rho_\bP$ as the sum of the $N$ marginals.

Notice that any positive measure $\rho$ on $\R^d$ with $\rho(\R^d)=N\in\N$ arises from at least one $N$--particle probability measure $\bP$. One can take for instance $\bP=(\rho/N)^{\otimes N}$ for independent and identically distributed particles.

The pairwise average interaction energy of the particles is given by
\begin{equation*}
	\cU_{N} \myp{\bP}
	= \int_{\R^{dN}} \sum\limits_{1 \leq j < k \leq N} w \myp{x_j - x_k} \id \bP \myp{x_1, \dotsc, x_N}.
\end{equation*}
It could in principle be equal to $+\ii$, but it always satisfies $\cU_{N} \myp{\bP}\geq -\kappa N$ due to the stability condition on $w$ in Assumption~\ref{de:shortrangenew}.
When considering systems at positive temperature $T>0$, it is necessary to also take the entropy of the system into account.
It is given by
\begin{equation}
\label{eq:canentropy}
	\cS_N \myp{\bP} := - \int_{\R^{dN}} \bP \myp{x} \log \big(N! \, \bP \myp{x}\big) \id x.
\end{equation}
If $ \bP $ is not absolutely continuous with respect to the Lebesgue measure on $ \R^{dN} $, we use the convention that $\cS_N \myp{\bP}=-\ii$. The factor $N!$ appears because the particles are indistinguishable. In fact, we should think that $N!\,\bP$ defines a probability measure over $(\R^d)^N/\gS_N$ where $\gS_N$ is the permutation group. We need to make sure that $\cS_N(\bP)<+\ii$, which follows if we assume for instance that $\rho_\bP$ is absolutely continuous with $\int_{\R^d}\rho_\bP|\log\rho_\bP|<\ii$. This is due to the well-known inequality (see, e.g.,~\cite[Lemma~6.1]{MarLewNen-22_ppt})
\begin{equation}
 \cS_N \myp{\bP}\leq -\int_{\R^d}\rho_\bP(x)\log\rho_\bP(x)\,\dx+N.
 \label{eq:upper_bd_S_canonical}
\end{equation}
The latter follows immediately from writing the relative entropy of $\bP$ with respect to $(\rho/N)^{\otimes N}$, which is non-negative, and using $(N/e)^N\leq N!$.

The total free energy of the system in the state $ \bP $ at temperature $ T \geq 0 $ equals
\begin{equation}
\label{eq:FFcan}
	\mathcal{F}_T \myp{\bP}
	:={} \cU_{N} \myp{\bP} - T \cS_N \myp{\bP}
	= \int_{\R^{dN}} \sum\limits_{j < k} w \myp{x_j - x_k} \id \bP \myp{x} + T \int_{\R^{dN}} \bP \log \myp{N! \, \bP }.
\end{equation}
It can be equal to $+\ii$ but never to $-\ii$ due to the stability of $w$ and thanks to the inequality~\eqref{eq:upper_bd_S_canonical} if $T>0$ and $\int_{\R^d}\rho_\bP|\log\rho_\bP|<\ii$.

Throughout the paper, we will only consider systems with a given one-body density $\rho$, which is absolutely continuous with respect to the Lebesgue measure. At $T>0$ we also assume that $\int_{\R^d}\rho|\log\rho|<\ii$. This allows us to consider the minimal energy of $ N $-particle classical systems with density $ \rho $, given by
\begin{equation}
\label{eq:Fcan}
\boxed{
	F_T [\rho]
	:= \inf_{\rho_{\bP} = \rho} \mathcal{F}_T \myp{\bP}
}
\end{equation}
where the infimum is taken over $ N $-particle states $ \bP $ on $ \R^{dN} $ with one-particle density $ \rho_{\bP} $ equal to $ \rho $. At $T=0$, the entropy term disappears and we obtain
\begin{equation}
\label{eq:Fcan_T0}
	F_0 [\rho]
	:= \inf_{\rho_{\bP} = \rho}\int \sum\limits_{1\leq j < k\leq N} w \myp{x_j - x_k} \id \bP \myp{x}.
\end{equation}
This is a multi-marginal optimal transport problem with symmetric cost $\sum_{j<k}w(x_j-x_k)$ and with all the marginals of $\bP$ equal to $\rho/N$~\cite{CotFriKlu-13,CotFriPas-15,Pass-15,MarGerNen-17,SeiMarGerNenGieGor-17}.
From the stability assumption on $w$ and~\eqref{eq:upper_bd_S_canonical}, we have
\begin{equation}
 F_T [\rho]\geq -(\kappa +T)N+T\int_{\R^d}\rho(x)\log\rho(x)\,\dx.
 \label{eq:lower_bd_F_T_simple}
\end{equation}
One of our goals will be to find simple conditions ensuring that $F_T[\rho]<\ii$. Before we turn to this question, we first introduce the grand-canonical problem.

\begin{remark}[Symmetry]\label{rmk:symmetry}
In the definition~\eqref{eq:Fcan} we can freely remove the constraint that $\bP$ is symmetric. Since the interaction is a symmetric function and the entropy $\cS_N$ is concave, the minimum is the same as for symmetric $\bP$'s. Recall that for a non-symmetric $\bP$, $\rho_\bP$ is by definition the sum of the $N$ marginals.
\end{remark}

\begin{remark}[Lower semi-continuity]\label{rmk:lsc}
The function $\rho\mapsto F_T[\rho]$ is lower semi-continuous for the strong topology. That is, we have
\begin{equation}
F_T[\rho]\leq \liminf_{n\to\ii} F_T[\rho_n]\quad\text{if $\int|\rho_n-\rho|\to0$ and $T\int\rho_n|\log\rho_n|\leq C$}
\end{equation}
At $T>0$ this is valid under the sole condition that $w$ is measurable (since the limiting probability $\bP$ is necessary absolutely continuous) but at $T=0$, this uses the lower semi-continuity of $w$. The details of the argument are provided later in the proof of Theorem~\ref{thm:non_compact}, for the convenience of the reader.
\end{remark}

\begin{remark}[Convexity and duality]\label{rmk:duality}
Using the concavity of the entropy $\cS_N$, one can verify that $\rho\mapsto F_T[\rho]$ is convex. This can be used to derive the \emph{dual formulation} of $F_T[\rho]$ in terms of external potentials
\begin{align}
F_T[\rho]={}& T\int_{\R^d}\rho\log\rho+\sup_{\tilde V}\bigg\{-\int_{\R^d}\rho(x) \tilde V(x)\,\rd x \nn \\
&-T\log\int_{\R^{dN}}\exp \myp[\bigg]{ -\frac{1}{T}\sum_{1\leq j<k\leq N}w(x_j-x_k)-\frac{1}{T}\sum_{j=1}^N\tilde V(x_j) } \rd\rho^{\otimes N}
\bigg\},
\label{eq:duality}
\end{align}
see~\cite{ChaChaLie-84}. Our notation $\tilde V$ is because the final physical dual potential is, rather, $V:=\tilde V-T\log\rho$. The existence of a maximizer $\tilde V$ realizing the above supremum is proved in~\cite{ChaChaLie-84}. It
is the unique potential (up to an additive constant) so that the corresponding Gibbs state has density $\rho$, that is,
$$\rho_\bP=\rho,\quad \bP=\frac{1}{Z}\exp\bigg(-\frac{1}{T}\sum_{1\leq j<k\leq N}w(x_j-x_k)-\frac{1}{T}\sum_{j=1}^N\tilde V(x_j)\bigg)\rho^{\otimes N}$$
with $Z$ a normalization constant.
At $T=0$, we have the similar formula
$$F_0[\rho]=\sup_{\tilde V}\bigg\{E_N[V]-\int_{\R^d}\rho(x) V(x)\,\rd x\bigg\}$$
where
$$E_N[V]=\inf_{x_1,...,x_N\in\R^d}\left\{\sum_{1\leq j<k\leq N}w(x_j-x_k)+\sum_{j=1}^NV(x_j)\right\},$$
is the ground state energy in the potential $V$~\cite{Kellerer-84}. Although there usually exist dual potentials at $T=0$, those are often not unique.
\end{remark}

\subsubsection{The grand-canonical free energy}
In the grand-canonical picture, where the exact particle number of the system is not fixed, a state $ \bP $ is a family of symmetric $ n $-particle positive measures $ \bP_n $ on $(\R^d)^n$, so that
\begin{equation*}
	 \sum_{n \geq 0} \bP_n\big((\R^d)^n\big)=1.
\end{equation*}
Here $\bP_0$ is just a number, interpreted as the probability that there is no particle at all in the system. After replacing $\bP_n$ by $\bP_n/\bP_n(\R^{dn})$, we can equivalently think that $\bP$ is a convex combination of canonical states.
The entropy of $ \bP $ is defined by
\begin{equation}
\label{eq:gcentropy}
	\cS \myp{\bP}
	:= \sum\limits_{n \geq 0} \cS_n \myp{\bP_n}
	= -\bP_0\log(\bP_0)- \sum\limits_{n \geq 1} \int_{\R^{dn}} \bP_n \log \myp{n! \, \bP_n },
\end{equation}
and the single particle density of the state $ \bP $ is
\begin{equation*}
	\rho_{\bP} = \sum\limits_{n \geq 1} \rho_{\bP_n}=\sum_{n\geq1}n\int_{(\R^d)^n}\rd\bP_n(\cdot,x_2,\dotsc,x_n).
\end{equation*}
The grand-canonical free energy of the state $ \bP $ at temperature $ T \geq 0 $ is
\begin{equation}
\label{eq:FFgc}
	\mathcal{G}_T \myp{\bP}
	:= \cU \myp{\bP} - T \cS \myp{\bP},
\end{equation}
where $ \cU \myp{\bP} $ denotes the interaction energy in the state $ \bP $,
\begin{equation}
\label{eq:gcinteraction}
	\cU \myp{\bP}
	:= \sum\limits_{n \geq 2} \cU_{n} \myp{\bP_n}
	= \sum\limits_{n \geq 2} \int_{\R^{dn}} \sum\limits_{j<k}^n w \myp{x_j - x_k} \id \bP_n \myp{x_1,...,x_N}.
\end{equation}
From the stability of $w$ we have
$$\cU_{n} \myp{\bP_n}\geq -\kappa n\,\bP_n(\R^{dn})$$
so that, after summing over $n$,
$$	\cU \myp{\bP}\geq -\kappa \int_{\R^d}\rho_\bP(x)\,\dx.$$
By~\cite[Lemma 6.1]{MarLewNen-22_ppt} we have the universal entropy bound
	\begin{equation}
	\label{eq:estim_entropy_1body}
		\cS \myp{\bP}\leq -\int_{\R^d}\rho_\bP\big(\log\rho_\bP-1).
	\end{equation}
This is because the entropy at fixed density $ \rho $ is maximized by the grand-canonical Poisson state
	\begin{equation}
	\label{eq:poisson_state}
		\bQ
		:= \left(\frac{e^{- \int_{\R^d} \rho}}{n!} \rho^{\otimes n}\right)_{n \geq 0}
	\end{equation}
whose entropy is the right side of~\eqref{eq:estim_entropy_1body}.

When keeping the one-particle density $ \rho = \rho_{\bP} \in L^1 \myp{\mathbb{R}^d} $ fixed, we denote the minimal grand-canonical free energy by
\begin{equation}
\label{eq:gcenergydensity}
\boxed{
	G_T [\rho]
	:= \inf_{\rho_{\bP} = \rho} \mathcal{G}_T \myp{\bP}.
}
\end{equation}
Using~\eqref{eq:estim_entropy_1body}, we obtain
	\begin{equation}
		G_T [\rho]
		\geq - \myp[\big]{\kappa  + T} \int_{\R^d} \rho + T \int_{\R^d} \rho \log \rho,
		\label{eq:lower_bd_G_T_simple}
	\end{equation}
where $ \kappa $ is the stability constant of $ w $ in Assumption~\ref{de:shortrangenew}.

\begin{remark}[Comparing $F_T$ and $G_T$]
	\label{rem:gc_can_connection}
	Since a canonical trial state is automatically also admissible for the grand-canonical minimisation problem \eqref{eq:gcenergydensity}, we have the bound
	\begin{equation*}
		G_T [\rho]
		\leq F_T [\rho]
	\end{equation*}
	for any density $ 0 \leq \rho \in L^1 \myp{\R^d} $ with integer mass. Hence, any universal lower energy bound for the grand-canonical ensemble is also a lower bound for the canonical ensemble. A natural question to ask is under which condition we have $F_T[\rho] =G_T[\rho]$ for a density $\rho$ of integer mass. In general this is a difficult problem. See~\cite{MarLewNen-22_ppt} for results and comments in this direction at $T=0$.

	If $\int_{\R^d}\rho=N+t$ with $t\in(0,1)$ and $N\in\N$, we can write $\rho=(1-t)\frac{N}{N+t}\rho+t\frac{N+1}{N+t}\rho$ and obtain after using the concavity of the entropy
	\begin{equation}
	  G_T [\rho]\leq (1-t)\,F_T \left[\frac{N}{N+t}\rho\right]+t\,F_T \left[\frac{N+1}{N+t}\rho\right].
	  \label{eq:upper_G_T_F_T}
	\end{equation}
	This can be used to deduce an upper bound on $G_T[\rho]$, once an upper bound has been established in the canonical case. We will see, however, that it is usually much easier to directly prove upper bounds on $G_T[\rho]$ than on $F_T[\rho]$.
\end{remark}

\begin{remark}[Weak lower semi-continuity]\label{rmk:lsc_GC}
The functional $\rho\mapsto G_T[\rho]$ is \emph{weakly} lower semi-continuous and, in fact, a kind of lower continuous envelope of $F_T[\rho]$ (see~\cite{LewLieSei-19_ppt,MarLewNen-22_ppt}). At $T=0$ this uses the lower semi-continuity of $w$.
\end{remark}

\begin{remark}[Duality II]\label{rmk:duality_GC}
Like in the canonical case, we have the dual formulation
\begin{multline}
G_T[\rho]=T\int_{\R^d}\rho\log\rho+\sup_{\tilde V}\bigg\{-\int_{\R^d}\rho(x) \tilde V(x)\,\rd x\\
-T\log\bigg[\sum_{n\geq0}\int_{\R^{dn}}\exp\bigg(-\frac{1}{T}\sum_{1\leq j<k\leq n}w(x_j-x_k)-\frac{1}{T}\sum_{j=1}^n\tilde V(x_j)\bigg)\rd\rho^{\otimes N}\bigg]
\bigg\},
\label{eq:duality_GC}
\end{multline}
see~\cite{ChaChaLie-84,ChaCha-84} and the more recent work~\cite[Sec.~4 \& 6]{MarLewNen-22_ppt}.
\end{remark}

\subsection{Representability}

Next we turn to the problem of representability. Namely, we are asking what kind of densities $\rho$ can arise from $N$--particle probabilities with finite free energy. This depends on the shape of the interaction potential $w$. We only address this question for $\rho\in L^1(\R^d)$ and do not look at general measures. The main result is that all densities are representable at zero temperature in the non-hard-core case ($\alpha<\ii$). At positive temperature, it is sufficient to assume in addition that $\int_{\R^d}\rho|\log\rho|<\ii$.

\begin{theorem}[Representability in the canonical case]\label{thm:representability}
Let $\rho\in L^1(\R^d)$ with $\int_{\R^d}\rho(x) \id x \in \N$. There exists a symmetric probability measure $\bP$ on $(\R^d)^N$ of density $\rho$ so that $|x_j-x_k|\geq \delta>0$ $\bP$--almost everywhere, for some $\delta>0$.

If $w$ satisfies Assumption~\ref{de:shortrangenew} without hard-core ($\alpha<\ii$), we obtain $F_0[\rho]<\ii$. If furthermore $\int_{\R^d}\rho|\log\rho|<\ii$, then $\bP$ can be assumed to have finite entropy and $F_T[\rho]<\ii$ for any $T>0$.
\end{theorem}

The theorem follows from results in optimal transport theory and we quickly outline the proof here for the convenience of the reader. In this paper we will prove much more. We will in fact need some of these tools and more details will thus be provided later in the paper.

\begin{proof}
If $\int_{\R^d}\rho=1$, we must take $\bP=\rho$ and end up with $F_T[\rho]=T\int\rho\log\rho$. In the rest of the proof we assume that $\int_{\R^d}\rho\geq2$.

For $\rho\in L^1(\R^d)$, the existence of $\bP$ is proved in~\cite[Theorem~4.3]{ColMarStr-19}. The number $\delta$ must be so that $\int_{B(x,\delta)}\rho<1$ for any $x\in\R^d$, where $B(x,R)$ denotes the ball centered at $x$ and of radius $R$. Such a $\delta>0$ always exists when $\rho\in L^1(\R^d)$. See Section~\ref{sec:OT} below for more details on the results from~\cite{ColMarStr-19}.

Next we prove that $\cF_0\myp{\bP}<\ii$. Since $\alpha<\ii$ (no hard-core), we can assume $r_0=1$. We then have $w(x)\leq C_\delta |x|^{-s}$ for all $|x|\geq\delta$, with the constant $C_\delta=\kappa(1+\delta^{s-\alpha})$, due to Assumption~\ref{de:shortrangenew}. Hence, on the support of $\bP$ we have
\begin{equation*}
\sum_{1\leq j<k\leq N}w(x_j-x_k)=\frac12\sum_{j=1}^N\sum_{k\neq j}w(x_j-x_k)
\leq \frac{C_\delta}{2}N\max_{\substack{|y_j|\geq\delta\\ |y_j-y_k|\geq \delta}}\sum_{j=1}^{N-1}\frac{1}{|y_j|^s}.
\end{equation*}
The maximum is bounded by $C\delta^{-s}$ independently of $N$ due to~\cite[Lemma~9]{Lewin-22}. Integrating with respect to $\bP$ we have proved that $\cF_0\myp{\bP}\leq C_\delta\delta^{-s}N$. This bound is not very explicit but it only depends on $\delta$ and $N$. Of course, $\delta$ itself depends on $\rho$ in a rather indirect way.

The probability measure $\bP$ obtained by the optimal transport method of~\cite{ColMarStr-19} is probably a singular measure, hence with an infinite entropy. In~\cite{CarDuvPeySch-17}, it is explained how to regularize any given $\bP$ using a method called the \emph{Block approximation}. This method works well for a compactly supported density, for which it easily implies $F_T[\rho]<\ii$. We quickly describe the method here and refer to Section~\ref{sec:Block} below for details. In short, we split the space into small cubes $\{\mathcal C_j\}$ of size proportional to $\delta$ and introduce the trial probability measure
$$\tilde\bP=\sum_{j_1,...,j_N}\bP(\mathcal C_{j_1}\times\cdots\times\mathcal C_{j_N})\frac{\rho\1_{\mathcal C_{j_1}}\otimes\cdots\otimes \rho\1_{\mathcal C_{j_N}}}{\int_{\mathcal C_{j_1}}\rho\cdots\int_{\mathcal C_{j_N}}\rho}.$$
That is, we take a convex combination of independent particles over small cubes with probability $\bP(\mathcal C_{j_1}\times\cdots\times\mathcal C_{j_N})$. Choosing the cubes small enough, we can ensure that $|x_j-x_k|\geq\delta/2$ on the support of $\tilde\bP$ and $\int_{\mathcal C_j}\rho<1$. A computation gives $\rho_{\tilde\bP}=\rho_{\bP}=\rho$. The entropy can be estimated by
$$\int_{\R^{dN}}\tilde\bP\log(N! \, \tilde\bP)\leq \int_{\R^d} \rho \log \rho  -\sum\limits_{j} \myp[\bigg]{\int_{\mathcal C_j} \rho } \log \myp[\bigg]{\int_{\mathcal C_j} \rho}$$
(see Lemma~\ref{lem:block_entropy} below). Estimating the last sum is not an easy task for a general density. For a compactly supported density we can simply bound it by $1/e$ times the numbers of cubes intersecting the support of $\rho$. Since the energy of $\tilde\bP$ is finite by the previous argument, we deduce that $F_T[\rho]<\ii$ for any $\rho$ of compact support.

It thus remains to explain how to prove that $F_T[\rho]$ is finite for a density $\rho$ of unbounded support. The idea is of course to truncate it. We choose two radii $R_1<R_2$ so that
$$\int_{\R^d\setminus B_{R_2}}\rho=\int_{B_{R_2}\setminus B_{R_1}}\rho=\frac12$$
(using here $\int\rho\geq2$) and we define for shortness $\rho_1:=\rho\1_{B_{R_1}}$, $\rho_2:=\rho\1_{B_{R_2}\setminus B_{R_1}}$ and $\rho_3:=\rho\1_{\R^d\setminus B_{R_2}}$. We can write
$$\rho=\frac{\rho_1+2\rho_2}2+\frac{\rho_1+2\rho_3}2$$
where $\int_{\R^d}(\rho_1+2\rho_2)=\int_{\R^d}(\rho_1+2\rho_3)=N$.
From the convexity of $F_T$ we obtain
$$F_T[\rho]\leq \frac12F_T[\rho_1+2\rho_2]+\frac12F_T[\rho_1+2\rho_3].$$
The first density $\rho_1+2\rho_2$ has compact support hence has a finite energy, as explained above. For the second density $\rho_1+2\rho_3$ we use an uncorrelated trial state in the form $\bP_1\otimes_s (2\rho_3)$ where $\bP_1$ is also constructed as before, but with $\rho$ replaced by $\rho_1$ which has mass $N-1$. Here $\otimes_s$ means the symmetric tensor product.
A calculation shows that
\begin{align*}
F_T[\rho_1+2\rho_3]\leq{}& \cF_T\big(\bP_1\otimes_s(2\rho_3)\big)\\
={}&\cF_T(\bP_1)+2\iint_{\R^{2d}}w(x-y)\rho_1(x)\rho_3(y)\id x \id y \\
&+2T\int \rho_3\log(2\rho_3) \\
\leq{}&  \cF_T(\bP_1)+(N-1)\sup_{|x|\geq R_2-R_1}|w(x)|+2T\int_{\R^d\setminus B_{R_2}} \rho\log(2\rho).
\end{align*}
Thus the finiteness for densities of compact support implies the same for all densities. In fact, after optimizing over $\bP_1$ we have proved the bound
\begin{align*}
F_T[\rho]\leq{}& \frac{F_T[\rho_1+2\rho_2]+F_T[\rho_1]}2\\
&+\frac{N-1}{2}\sup_{|x|\geq R_2-R_1}|w(x)|+T\int_{\R^d\setminus B_{R_2}} \rho\log(2\rho).
\end{align*}
This concludes the proof of Theorem~\ref{thm:representability}.
\end{proof}

We have not considered here the hard-core potential, to which we will come back later in Section~\ref{sec:hard-core}. Representability is much more delicate in this case. From the inequality~\eqref{eq:upper_G_T_F_T}, we immediately obtain the following.

\begin{corollary}[Representability in the grand-canonical case]
Let $\rho\in L^1(\R^d)$. Then we have $G_0[\rho]<\ii$ if $w$ has no hard-core ($\alpha<\ii$). If furthermore $\int_{\R^d}\rho|\log\rho|<\ii$, then $G_T[\rho]<\ii$ for all $T>0$.
\end{corollary}

\subsection{Local upper bounds}

Recall that we already have rather simple lower bounds in~\eqref{eq:lower_bd_F_T_simple} and~\eqref{eq:lower_bd_G_T_simple}. The proof of Theorem~\ref{thm:representability} furnishes an upper bound on $F_T[\rho]$ but it depends on the smallest distance $\delta$ between the particles in the system, which is itself a highly nonlinear and nonlocal function of $\rho$. For non compactly-supported densities, the proof also involves the two radii $R_1,R_2$ which depend on $\rho$ as well.

Our goal here is to provide simple \emph{local upper bounds} involving only integrals of the given density $\rho$. We start in the next subsection by recalling the simple integrable case at the origin $\alpha<d$, for which we can just choose i.i.d.~particles. The case $\alpha\geq d$ is much more complicated since particles cannot be allowed to get too close.

\subsubsection{Upper bound in the weakly repulsive case $\alpha<d$}
In the case where $ w_+ $ is integrable at the origin, it is easy to provide a simple upper bound.

\begin{theorem}[Weakly repulsive case $\alpha<d$]
\label{thm:weak_w_bound}
	Let $ w $ satisfy Assumption~\ref{de:shortrangenew} with $ \alpha < d $.
	Let $ 0 \leq \rho \in L^1 \myp{\R^d} \cap L^2 \myp{\R^d} $ with integer mass $ \int \rho \in \N $. Let also $T\geq0$ and assume that $\int_{\R^d}\rho|\log\rho|<\ii$ if $T>0$. Then we have
	\begin{align}
		F_T [\rho]&\leq \frac12\iint_{\R^d\times\R^d}w(x-y)\rho(x)\rho(y)\,\dx\,dy+T\int_{\R^d}\rho\log\rho\nn\\
		&\leq \frac{\normt{w_+}{L^1}}{2} \int_{\R^d} \rho^2 + T \int_{\R^d} \rho \log \rho.\label{eq:can_weak_bound}
	\end{align}
	In the grand-canonical case we have the exact same bound on $G_T[\rho]$, this time without any constraint on $\int_{\R^d} \rho $ and with $\rho\log\rho$ replaced by $\rho(\log\rho-1)$ in the last integral.
\end{theorem}

As we have mentioned in the introduction, the functional appearing on the right side of the first line of~\eqref{eq:can_weak_bound} is the so-called \emph{Kirkwood-Monroe free energy}~\cite{KirMon-41}, which is the simplest approximation of $F_T[\rho]$. It only makes sense for a locally integrable potential $w$. In addition to being an exact upper bound, the Kirkwood-Monroe free energy also provides the exact behavior of $F_T[\rho]$ in some regimes. This was studied in many works, including for instance~\cite{LebPen-66,GatPen-69,Gates-72,GreKle-76,ButLeb-05} for the infinite gas at high density and~\cite{BraHep-77,Spohn-81,MesSpo-82,Kiessling-89,Kiessling-93,CagLioMarPul-92,CagLioMarPul-95,KiePer-95,Rougerie-LMU} for trapped systems in the mean-field limit.

\begin{proof}
	We denote $ N = \int_{\R^d} \rho $ and simply take as a trial state the pure tensor product $ \bP := \myp{\rho/N }^{\otimes N} $.
	The interaction energy satisfies
	\begin{align}
		\cU_{N} \myp{\bP}
		={}& \frac{N \myp{N-1}}{2} \int_{\R^{dN}} w \myp{x_1 - x_2} \myp[\Big]{\frac{\rho}{N} }^{\otimes N} \myp{x} \id x_1 \cdots \id x_N \nn\\
		={}& \frac{1-1/N}{2} \iint_{\R^d\times\R^d} w \myp{x_1 - x_2} \rho \myp{x_1} \rho \myp{x_2} \id x_1 \id x_2\label{eq:mean-field2}\\
		\leq{}& \frac{\normt{w_+}{L^1}}{2} \int_{\R^d} \rho^2.\nn
	\end{align}
From the stability condition on $w$, we know that for any $\eta\geq0$ with $\int\eta=1$,
		$$\cU_{K} \myp{\eta^{\otimes K}}=\frac{K(K-1)}{2}\iint_{\R^d\times\R^d} w \myp{x -y} \eta(x)\eta(y)\,\dx\,\dy\geq -\kappa K.$$
Letting $K\to\ii$, we find
$$\iint_{\R^d\times\R^d} w \myp{x -y} \eta(x)\eta(y)\,\dx\,\dy\geq0,\qquad \forall \eta\geq0.$$
This is how the stability is expressed in mean-field theory~\cite{LewNamRou-16c}. Since the double integral in~\eqref{eq:mean-field2} is non-negative, we can remove the $1/N$ for an upper bound. The entropy can itself be estimated by
	\begin{align*}
		-\cS_N \myp{\bP}
		={}& \int_{\R^{dN}} \myp[\Big]{\frac{\rho}{N} }^{\otimes N} \log \myp[\Big]{N! \myp[\Big]{\frac{\rho}{N} }^{\otimes N}} \\
		={}& \log \myp[\Big]{\frac{N!}{N^N}} + \int_{\R^d} \rho \log \rho
		\leq{} \int_{\R^d} \rho \log\rho,
	\end{align*}
	showing that \eqref{eq:can_weak_bound} holds.
	In the grand-canonical case we use instead the Poisson state in~\eqref{eq:poisson_state} and exactly obtain the mean-field energy on the right side of~\eqref{eq:can_weak_bound} with $\rho\log\rho$ replaced by $\rho(\log\rho-1)$ in the last integral.
\end{proof}

\subsubsection{Upper bounds in the strongly repulsive case $\alpha\geq d$}

When $\alpha\geq d$ the right side of~\eqref{eq:can_weak_bound} is infinite due to the non-integrability of $w$ at the origin. We cannot use a simple uncorrelated probability $\bP$ as a trial state and it is necessary to correlate the particles in such a way that they never get too close to each other. The difficulty is to do this at fixed density, with a reasonable energy cost. Also, we expect the typical distance between the particles to depend on the local value of $\rho$. If we imagine that there are $\rho(x)$ particles per unit volume in a neighborhood of a point $x$, then the distance should essentially be proportional to $\rho(x)^{-1/d}$. We thus expect a bound in terms of $\rho(x)^{1+\alpha/d}$ for large densities. We can only fully solve this question in the grand-canonical case. In the canonical case we can only treat $T=0$ in full. The following is our first main result.

\begin{theorem}[Strongly repulsive case $\alpha \geq d$]
\label{thm:GC_bound}
Suppose that the interaction $ w $ satisfies Assumption~\ref{de:shortrangenew} with $ d \leq \alpha < \infty $. Let $T\geq 0$ and assume that for $ T > 0 $, we have $\int_{\R^d}\rho|\log\rho|<\ii$.

\smallskip

\noindent$\bullet$ In the grand-canonical ensemble, we have for any $ 0 \leq \rho \in L^1 \myp{\R^d} $,
		\begin{align}
		G_T [\rho]&\leq  C\kappa \int_{\R^d}\rho^2+CT \int_{\R^d} \rho + T \int_{\R^d} \rho \log \rho\nn\\
			&\qquad +\begin{cases}
			\dps C\kappa r_0^\alpha \int_{\R^d} \rho^{1+\frac{\alpha}{d}}&\text{for $\alpha>d$,}\\[0.4cm]
			\dps C\kappa r_0^d \left(\int_{\R^d} \rho^2+\int_{\R^d}\rho^2\big(\log r_0^d\rho\big)_+\right)&\text{for $\alpha=d$.}
			\end{cases}\label{eq:GC_bound}
		\end{align}
Here the constant $ C $ only depends on the dimension $ d $ and the powers $\alpha,s$ from Assumption~\ref{de:shortrangenew}.

\smallskip

\noindent$\bullet$ In the canonical ensemble we have the same estimate on $F_T[\rho]$ for all $T\geq0$ in dimension $d=1$ and on $F_0[\rho]$ at $T=0$ for $d\geq2$, provided of course that $\rho$ has an integer mass.
\end{theorem}

In the proof we provide an explicit value for the constant $C$ in~\eqref{eq:GC_bound} but we do not display it here since it is by no means optimal and depends on the cases. The parameters $\kappa$ and $r_0$ can be used to track the origin of the different terms in our bound~\eqref{eq:GC_bound}. The integrable part of the potential gives the $\rho^2$ term as it did in Theorem~\ref{thm:weak_w_bound}. The terms involving $r_0^\alpha$ on the second line are solely due to the divergence of $w$ at the origin. It is important that we get here the expected and optimal $\rho^{1+\alpha/d}$ due to the singularity. Finally, we have an additional term involving $T\rho$ which is an error in the entropy due to our construction. We otherwise get the optimal $T\rho\log\rho$.

In dimension $d=1$, the proof of Theorem~\ref{thm:GC_bound} is relatively easy, both in the canonical and grand-canonical cases. It is detailed for convenience in Section~\ref{sec:proof_1D}. The idea is to split the density $\rho$ into successive intervals of mass $1/2$ and then write $\rho=(2\rho_\text{odd}+2\rho_\text{even})/2$ where $\rho_\text{odd}$ is the density restricted to the odd intervals and $\rho_\text{even}$ to the even ones. We then take a trial state of the form $(\bP_\text{odd}+\bP_\text{even})/2$, where $\bP_\text{odd}$ corresponds to placing exactly one particle per odd interval at density $2\rho$ and $\bP_\text{even}$ is defined similarly. This way we have inserted some distance between the particles. It depends on the form of $\rho$ in the opposite set of intervals. The interaction between the particles can then be easily controlled in terms of $\rho^{1+\alpha/d}$, as we explain in Section~\ref{sec:proof_1D}.

In higher dimensions, there seems to be no general way of splitting $\R^d$ into disjoints sets containing a fixed mass of $\rho$, so that each set has finitely many neighbors at a given distance (except perhaps for very special densities~\cite{GigLeo-17}). We can however carry over a similar argument as in the 1D case if we allow a covering with intersections. The Besicovitch covering lemma~\cite{Guzman} allows us to work with cubes $Q_j$ intersecting with finitely many other cubes, such that $\int_{Q_j}\rho$ is any given number. We can also distribute the $Q_j$ into a finite (universal) number of subcollections so that the cubes in each family are disjoint and not too close to each other. For each collection of disjoint cubes we then use a simple tensor product similar to the 1D case. The interaction is estimated using that the length of the cubes is related to $\int_{Q_j}\rho^{1+\alpha/d}$, leading to a bound involving only $\int_{\R^d}\rho^{1+\alpha/d}$. This proof was inspired by the presentation in the recent book~\cite{FraLapWei-LT} of a proof of the Lieb-Thirring and Cwikel-Lieb-Rozenblum inequalities from~\cite{Rozenbljum-71,Rozenbljum-72,Weidl-96}, thus in a completely different context. The difficulty here is that we have no information on the number of particles in each subcollection, due to the overlaps. This is the reason why the proof works well in the grand-canonical setting, but not in the canonical case. The details are given in Section~\ref{sec:proof_GC}.

To prove the result in the canonical case at $T=0$ for $d\geq2$, we use a completely different method based on optimal transport tools from~\cite{ColMarStr-19}. As we will explain in Section~\ref{sec:proof_C}, the latter work can be used to construct a trial state $\bP$ with $\rho_\bP=\rho$ so that the distance between any two given particles on the support can be related to some average local value of the density around the particles. This is how we can obtain the bound~\eqref{eq:GC_bound} at $T=0$ in the canonical case.

The next natural step is to smear this trial measure $\bP$ and use it at $T>0$ but we could unfortunately not give an optimal bound on the entropy of the smearing. Our bound relies on the local radius $ R \myp{x} $ of a density $ \rho $, which is thoroughly studied in Section~\ref{sec:OT} and is defined as follows. Let $ 0 \leq \rho \in L^1 \myp{\mathbb{R}^d} $ with $ \int_{\R^d} \rho \myp{y} \id y > 1 $.
For each $ x \in \R^d $, we define the local radius $ R(x)$  to be the largest number satisfying
\begin{equation}
\label{eq:Rdef}
\int_{B \myp{x, R \myp{x}}} \rho \myp{y} \id y = 1.
\end{equation}
This number is always bounded below for a given $\rho\in L^1(\R^d)$ but behaves like $|x|$ at infinity. If $\rho$ has compact support, then $R(x)$ is bounded on the support of $\rho$.

\begin{theorem}[Strongly repulsive case $\alpha\geq d$ II]\label{thm:OT_block}
Suppose that the interaction $ w $ satisfies Assumption~\ref{de:shortrangenew} with $ 2\leq d \leq \alpha < \infty $. Let $T>0$ and $ 0 \leq \rho \in L^1 \myp{\R^d}$ of integer mass with $\int_{\R^d}\rho|\log\rho|<\ii$. Then we have
		\begin{align}
		F_T [\rho]&\leq  C(\kappa+T) \int_{\R^d}\rho^2+CT \int_{\R^d} \rho + T \int_{\R^d} \rho \log \rho+T \int_{\R^d} \rho \log R^d\nn\\
			&\qquad +\begin{cases}
			\dps C\kappa r_0^\alpha \int_{\R^d} \rho^{1+\frac{\alpha}{d}}&\text{for $\alpha>d$,}\\[0.4cm]
			\dps C\kappa r_0^d \left(\int_{\R^d} \rho^2+\int_{\R^d}\rho^2\big(\log r_0^d\rho\big)_+\right)&\text{for $\alpha=d$,}
			\end{cases}\label{eq:can_bound2}
		\end{align}
where the constant $ C $ only depends on the dimension $ d $ and the powers $\alpha,s$ from Assumption~\ref{de:shortrangenew}.
\end{theorem}

The main difference compared to~\eqref{eq:GC_bound} is the additional term $T\int\rho\log R^d$, which we conjecture should not be present. It is only affecting the bound in places where $R$ is large on the support of $\rho$, that is, where one cannot find a sufficient amount of mass at a finite distance of $x$. Another small difference is the additional term $CT\int\rho^2$ due to our way of estimating the entropy. The proof is detailed in Section~\ref{sec:proof_OT_block} below.

The upper bounds in Theorems~\ref{thm:GC_bound} and~\ref{thm:OT_block} will be very useful for our next work~\cite{JexLewMad-23b_ppt} where we study $F_T[\rho]$ and $G_T[\rho]$ for extended systems. The sub-optimal upper bound~\eqref{eq:can_bound2} in the canonical case will be sufficient in this context.

\begin{remark}[Lower bounds]\rm
Even when $w$ really behaves like $|x|^{-\alpha}$ at the origin (for instance satisfies $w(x)\geq c|x|^{-\alpha}$ for some $c>0$), a lower bound in the form~\eqref{eq:GC_bound} cannot hold in general. This is because the density can be large in regions where there is only one particle at a time, which does not create any divergence in the interaction. As an example, consider $N$ points $X_1,...,X_N\in\R^d$ and place around each point one particle in the state $\chi_r:=|B_r|^{-1}\1_{B_r}$, with $r$ small enough. The corresponding state is the (symmetrization of the) tensor product $\bP_r=\bigotimes_{j=1}^N\chi_r(\cdot-X_j)$. Assuming that $w$ is continuous, its interaction energy behaves as
$$\lim_{r\to0}\cU_{N}(\bP_r)=\sum_{1\leq j<k\leq N}w(X_j-X_k)$$
hence stays finite, whereas the entropy equals
$$\cS_{N}(\bP_r)=-N\int\chi_r\log\chi_r=N\log(|B_1|r^d)\underset{r\to0}\longrightarrow-\ii.$$
On the other hand, the right side of~\eqref{eq:GC_bound} diverges much faster, like $Nr^{-\alpha}$. This proves that a lower bound of the form~\eqref{eq:GC_bound} cannot hold for all possible densities.

Nevertheless, it is expected that the term $\int\rho^{1+\alpha/d}$ should appear when there are many particles in a small domain and is thus optimal in such situations. For instance, assuming $w\geq c|x|^{-\alpha}$ for $|x|\leq r_0$ and taking $\rho=N|B_{r_0/2}|^{-1}\1_{B_{r_0/2}}$ ($N$ particles at uniform density in the small ball), we see that
$$F_T[\rho]\geq \min_{x_1,...,x_N\in B_{r_0/2}}\left( \sum_{1\leq j<k\leq N}\frac{c}{|x_j-x_k|^\alpha}\right)+T\log(N/|B_{r_0/2}|)-TN.$$
The first minimum is known to behave like $N^{1+\alpha/d}r_0^{-\alpha}$ in the limit $N\to\ii$~\cite[Lemma~1]{Lewin-22}, which is exactly proportional to $\int\rho^{1+\alpha/d}$. Thus in this case, the lower bound holds and the power $1+\alpha/d$ is optimal.
\end{remark}

\subsection{The hard-core case}\label{sec:hard-core}
We conclude this section with a discussion of the hard-core case, which is notoriously more difficult~\cite[Sec.~9]{ChaChaLie-84}. We start with the question of representability of a given density and then turn to some upper bounds on the free energy.

\subsubsection{Representability}
Let $r_0>0$ be a positive number and consider the hard-core potential $w_{r_0}(x)=(+\ii)\1(|x|<r_0)$. Then we have for any $N$-particle probability measure $\bP$
$$\cU_{N}(\bP)=\begin{cases}
0&\text{if $|x_j-x_k|\geq r_0$ $\forall j\neq k$, $\bP$--almost surely,}\\
+\ii&\text{otherwise.}
\end{cases}$$
The set of $\bP$'s such that $\cU_{N}(\bP)=0$ is convex and its extreme points are the symmetric tensor products of Dirac deltas located at distance $\geq r_0$ from each other. It follows that the convex set of $w_{r_0}$--representable densities is the convex hull of the densities in the form
\begin{equation}
\rho=\sum_{j=1}^N \delta_{x_j},\qquad \min_{j\neq k}|x_j-x_k|\geq r_0.
\label{eq:extreme_hard_core}
\end{equation}
There is a similar result in the grand-canonical case. In spite of this simple characterization, it seems very hard, in general, to determine whether a given density belongs to this convex set or not.

In dimension $d=1$, the problem can be solved exactly. Any extreme point~\eqref{eq:extreme_hard_core} satisfies
\begin{equation}
\rho \big([x,x+r_0)\big)\leq 1,\qquad\forall x\in\R,
 \label{eq:cond_1D_hard_core}
\end{equation}
since there is always at most one Dirac delta in any interval of length $r_0$. This property pertains on the whole convex hull of $w_{r_0}$--representable densities. Conversely, any positive measure $\rho$ with $\rho(\R)=N$ satisfying~\eqref{eq:cond_1D_hard_core} can be written as a convex combination of Dirac deltas at distance $\geq r_0$. To see this, assume for simplicity $\rho\in L^1(\R)$ and define as in~\cite{ColPasMar-15} the non-decreasing function $t\mapsto x(t)$ on $(0,N)$ so that
$$\int_{-\ii}^{x(t)}\rho(s)\,\rd s=t,\qquad \forall t\in(0,N).$$
To avoid any ambiguity when the support of $\rho$ is not connected, we can choose $x(t)$ to be the largest possible real number satisfying the above condition. The function $t\mapsto x(t)$ is differentiable, except possibly on a countable set, with $x'(t)=\rho(x(t))^{-1}$. When $\rho>0$ almost surely, we have $\lim_{t\to0^+}x(t)=-\ii$ and $\lim_{t\to N^-}x(t)=+\ii$. From the definition of $x(t)$ we have
\begin{equation}
\rho=\int_0^N\delta_{x(t)}\,\rd t.
 \label{eq:decomp_rho_1D}
\end{equation}
Indeed, if we integrate the right side against some continuous function $f$ we find
$\int_0^N f(x(t))\,\rd t=\int_\R f(s)\rho(s)\,\rd s$ after changing variable $s=x(t)$. Now we can also rewrite~\eqref{eq:decomp_rho_1D} as
\begin{equation}
\rho=\int_0^1\sum_{k=0}^{N-1}\delta_{x(t+k)}\,\rd t.
 \label{eq:decomp_rho_1D_bis}
\end{equation}
By definition of $x(t)$ we have
$$\int_{x(t+k)}^{x(t+k+1)}\rho(s)\,\rd s=1,\qquad \forall k=0,...,N-2,\quad \forall t\in(0,1),$$
and therefore $|x(t+k+1)-x(t+k)|\geq r_0$ when the condition~\eqref{eq:cond_1D_hard_core} is satisfied. Hence~\eqref{eq:decomp_rho_1D_bis} is the sought-after convex combination of delta's located at distance $\geq r_0$. The corresponding $N$-particle probability is
\begin{equation}
\bP=\Pi_s\int_0^1\delta_{x(t)}\otimes\delta_{x(t+1)}\otimes\cdots\otimes\delta_{x(t+N-1)}\,\rd t
\label{eq:form_Monge_1D}
\end{equation}
where
\begin{equation}
	\Pi_s \myp{f_1 \otimes\cdots\otimes f_N}
	= \frac{1}{N!} \sum\limits_{\sigma \in \gS_N} f_{\sigma\myp{1}} \otimes\cdots\otimes f_{\sigma\myp{N}},
	\label{eq:def_Pi_s}
\end{equation}
is the symmetrization operator. At positive temperature, the previous state can be regularized using the block approximation described in the proof of Theorem~\ref{thm:representability}, provided that $\int_\R\rho|\log\rho|<\ii$ and~\eqref{eq:cond_1D_hard_core} holds with a strict inequality.

In dimensions $d\geq2$, the situation is much less clear. The condition~\eqref{eq:cond_1D_hard_core} can be re-expressed in the form
\begin{equation}
\boxed{R_\rho:=\min_{x\in\R^d}R(x)\geq \frac{r_0}{2}}
\label{eq:cond_R_rho_hard_core}
\end{equation}
where $R(x)$ is the radius previously defined in~\eqref{eq:Rdef}. This can also be written in the form
$$\int_{B(x,r_0/2)}\rho\leq1,\qquad \forall x\in\R^d.$$
This is definitely a \emph{necessary condition} for a density to be $w_{r_0}$--representable, in dimension $d\geq1$. Otherwise we would be able to find an $x\in\R^d$ and an $R<r_0/2$ such that $\int_{B(x,R)}\rho>1$. But then the probability that there are at least two particles in the ball $B(x,R)$ cannot vanish for any $\bP$ of density $\rho$ and those are at distance $<r_0$. This was already mentioned in~\cite[Sec.~9]{ChaChaLie-84}.

For $d\geq2$ the condition~\eqref{eq:cond_R_rho_hard_core} is definitely \emph{not sufficient} for a density to be representable. A counter example arises naturally within the \emph{sphere packing} problem. Recall that the $d$-dimensional \emph{sphere packing density}
\begin{equation}
\rho_c(d):=\lim_{\ell\to\ii}\frac{\max\{ N\ :\ \exists x_1,...,x_N\in \Omega_\ell,\ |x_j-x_k|\geq1\}}{|\Omega_\ell|}
\label{eq:sphere_packing}
\end{equation}
gives the maximal number of points per unit volume one can put while ensuring that they are at distance $\geq1$ to each other. Here $\Omega$ is any fixed smooth domain and $\Omega_\ell=\ell\Omega$. The packing density equals $\rho_c(1)=1$ in dimension $d=1$ and is otherwise only known in dimensions $ d\in\{ 2,3,8,24\} $, for which it is given by some special lattices~\cite{Cohn-17,Viazovska-21}. The \emph{sphere packing fraction} is defined by
$$v_c(d):=\rho_c (d)|B_{1/2}|=2^{-d}\rho_c (d)|B_{1}|$$
and represents the fraction of the volume occupied by the balls. This is simply $v_c(1)=1$ in dimension $d=1$ but is strictly less than 1 for $d\geq2$. Some volume has to be left unoccupied due to the impossibility to fill space with disjoint balls of fixed radius.
It has been shown that $v_c(d)$ tends to 0 exponentially fast in the limit $d\to\ii$ but its exact behavior is still unknown~\cite{TorSti-06}. Let us now consider a constant density $\rho(x)=\rho_0\1_{\Omega_\ell}(x)$ over a large domain $\Omega_\ell=\ell\Omega$ (for instance a ball). Then we have $R(x)=(\rho_0|B_1|)^{-1/d}$ well inside $\Omega_\ell$, whereas $R(x)\geq (\rho_0|B_1|)^{-1/d}$ close to the boundary. This shows that for this density
$$R_\rho=\min_{x\in\R^d}R(x)=(\rho_0|B_1|)^{-\frac1d}=\frac{r_0}2 \myp[\bigg]{ \frac{r_0^{-d}\rho_c(d)}{\rho_0v_c(d)} }^{\frac1d}.$$
In particular, the condition~\eqref{eq:cond_R_rho_hard_core} is satisfied whenever $\rho_0\leq r_0^{-d}\rho_c(d)/v_c(d)$.
On the other hand, it is clear from the packing problem (rescaled by $r_0$) that when $\rho_0>r_0^{-d}\rho_c(d)$ the density cannot be representable for $\ell$ large enough. Otherwise we would be able to place $N=\rho_0|\Omega_\ell|>r_0^{-d}\rho_c(d)|\Omega_\ell|$ points in $\Omega_\ell$ at distance $r_0$, which contradicts the definition of $\rho_c(d)$. In conclusion, we have found that, in dimensions $d\geq2$, constant densities $\rho_0\1_{\Omega_\ell}$ with
$$r_0^{-d}\rho_c(d)<\rho_0\leq \frac{r_0^{-d}\rho_c(d)}{v_c(d)}$$
satisfy~\eqref{eq:cond_R_rho_hard_core} but cannot be $w_{r_0}$--representable for $\ell\gg1$.

As a side remark, we mention that there are representable densities satisfying~\eqref{eq:cond_R_rho_hard_core}, with $R_\rho$ as close as we want to $r_0/2$. We can just take the sum of two Dirac deltas placed at distance $R\geq r_0$ or a smooth approximation of it. This proves that there cannot exist a simple necessary and sufficient condition of hard core representability involving $R_\rho$ only, in dimensions $d\geq2$. This is in stark contrast with the one-dimensional case.

There exists, however, a simple \emph{sufficient condition} in a form that was conjectured in~\cite[p.~116]{ChaChaLie-84}. In~\cite[Theorem 4.1]{ColMarStr-19} (see also Theorem~\ref{thm:clstate} below), it is proved that any density satisfying
$$\boxed{R_\rho\geq r_0}$$
is $w_{r_0}$--representable. The same holds when $T>0$ if one puts a strict inequality. It would be interesting to know if such a result is valid for $R_\rho\geq c_d r_0$ with $c_d<1$, depending on the dimension.

The conclusion of our discussion is that there seems to exist no simple characterization of hard core representability in dimensions $d\geq2$, involving averages of $\rho$ over balls. There are necessary or sufficient conditions but they do not match.

\subsubsection{Upper bounds}
Next we discuss upper bounds in the hard core case.

Even if we do not completely understand when a density is hard-core representable, the energy is very easy to bound when it is the case. Let us assume that $w$ satisfies Assumption~\ref{de:shortrangenew} with $\alpha=+\ii$ and that $\rho\in L^1(\R^d)$ is $w$--representable. For simplicity we also assume that $w=+\ii$ on $B_{r_0}$. Then, for any optimizer $\bP$, we have $|x_j-x_k|\geq r_0$ for $j\neq k$, $\bP$--almost surely. This implies
\begin{equation}
F_0[\rho]= \cU_N(\bP)\leq \int_{(\R^d)^N}\sum_{1\leq j<k\leq N} \frac{\kappa\1(|x_j-x_k|\geq r_0)}{|x_j-x_k|^s}\,\rd\bP\leq C\kappa Nr_0^{-s}
 \label{eq:simple_hard_core_T0}
\end{equation}
by~\cite[Lemma~9]{Lewin-22}. The constant $C$ only depends on $s$ and $d$. Upper bounds are easy once we know that the particles cannot get too close.

Constructing trial states with a good entropy is more difficult. Our proofs of Theorems~\ref{thm:GC_bound} and~\ref{thm:OT_block} work in the hard-core case, but they require additional conditions, of the form
$$R_\rho>r_0\qquad\text{or}\qquad\int_{B(x,r_0/2)}\rho\leq \eps$$
for a sufficiently small $\eps$. We do not state the corresponding results here and rather refer the reader to Remarks~\ref{rmk:Hard-core-1D},~\ref{rmk:Hard-core-GC}, and~\ref{rmk:Hard-core-CT} below. In the rest of this section we quickly discuss the grand-canonical 1D case which has been studied in a famous paper of Percus~\cite{Percus-76} and the situation where $\rho$ is bounded uniformly.

\subsubsection*{The 1D grand-canonical Percus formula}
The grand-canonical inverse problem was completely solved by Percus in dimension $d=1$ in~\cite{Percus-76} (see also~\cite{RobVar-81}). Under the optimal assumption that $R_\rho>r_0/2$, he proved that the grand-canonical Gibbs state with external potential
$$V(x)=-\log\rho(x)+\log\left(1-\int_{x-r_0}^x\rho\right)-\int_x^{x+r_0}\frac{\rho(s)}{1-\int_{s-r_0}^s\rho}\,\rd s$$
and hard-core $w_{r_0}$ has the density $\rho$. Since the potential $\tilde V=V+\log\rho$ solves the supremum in the dual formula~\eqref{eq:duality_GC}, we obtain
\begin{equation}
\boxed{G_T[\rho]= T\int_\R\rho(x)\big(\log\rho(x)-1\big)\,\rd x-T\int_\R\rho(x)\log\left(1-\int_{x-r_0}^x\rho\right)\rd x}
\label{eq:estim_Percus}
\end{equation}
for the hard core potential $w_{r_0}$. This explicit expression shows us that, in one dimension, the nonlocality is solely due to the second logarithmic term, which involves the local average $\int_{x-r_0}^x\rho$ over a window of length $r_0$. This is further discussed in~\cite{RobVar-81}.

For a general potential $w$ satisfying Assumption~\ref{de:shortrangenew}, we only obtain an upper bound and need to add $C\kappa r_0^{-s}\int_\R\rho$ by~\eqref{eq:simple_hard_core_T0}. We can estimate the logarithm by assuming, for instance, that $\int_{x-r_0}^x\rho\leq 1-\eps$ for all $x\in\R$.

To our knowledge the canonical problem was never solved in the manner of Percus. It would be interesting to derive an upper bound on $F_T[\rho]$ of the same form as the right side of~\eqref{eq:estim_Percus}.

\subsubsection*{Bound for densities uniformly bounded by the packing density}
In dimensions $d\geq2$ we have no simple criterion of representability, as we have seen. One simpler situation is when $\rho$ is everywhere bounded above by the sphere packing density,  which we have defined in~\eqref{eq:sphere_packing}. Then we can prove it is representable and furnish an explicit upper bound on its grand-canonical free energy.

\begin{theorem}[Hard-core case with packing density bound]\label{thm:hard-core2}
Assume that $w$ satisfies Assumption~\ref{de:shortrangenew} with $\alpha=+\ii$. Let $\rho_c(d)$ be the sphere packing density in~\eqref{eq:sphere_packing} and $v_c(d)=2^{-d}\rho_c(d)|B_1|$ be the volume fraction. Let $\rho \in L^1 \myp{\R^d,\R_+}$ be such that
$$\rho(x)\leq (1-\eps)^dr_0^{-d}\rho_c(d)$$
for some $\eps\in(0,1)$. We also assume that $\int_{\R^d}\rho|\log\rho|<\infty$ if $ T > 0 $. Then
\begin{equation*}
G_T[\rho]\leq\frac{C\kappa}{r_0^s}\int_{\mathbb R^d}\rho+T\int_{\R^d}\rho\log\rho+T\log\left(\frac{2^d}{\eps^d v_c(d)}\right)\int_{\R^d}\rho,
\end{equation*}
with a constant $C$ depending only on the dimension $d$ and the power $s$ from Assumption~\ref{de:shortrangenew}.
\end{theorem}

The idea of the proof is to first construct a trial state for a constant density $\rho_0\approx(1-\eps)^dr_0^{-d}\rho_c(d)$ by using a periodic sphere packing with a large period, uniformly averaged over translations (often called a ``floating crystal''~\cite{LewLieSei-19b}). We then ``geometrically localize''~\cite{Lewin-11} this state to make it have density $\rho$. The proof is detailed later in Section~\ref{sec:proof_hard-core}.

\section{Proof of Theorem~\ref{thm:GC_bound} in dimension \texorpdfstring{$d=1$}{d=1}}\label{sec:proof_1D}

We start with the one-dimensional canonical case, for which the argument is relatively easy. We detail the proof for the convenience of the reader and because this will pave the way for the more complicated covering methods in higher dimensions. We only consider here the canonical case. The grand-canonical bound~\eqref{eq:GC_bound} follows using~\eqref{eq:upper_G_T_F_T}, but in the next section we will provide a direct proof in the grand-canonical case which also works in dimension $d=1$.

\begin{theorem}[$ d = 1 $]
	Suppose the interaction $ w $ satisfies Assumption~\ref{de:shortrangenew} with $ 1 \leq \alpha < \infty $.
	Let $ T \geq 0 $ and assume that $ \int_{\R} \rho \abs{\log \rho} < \infty $ for $ T > 0 $.
	Then for any density $ 0 \leq \rho \in L^1 \myp{\R} $ with $ \int_{\R} \rho \in \N $, we have
	\begin{align}
		F_T \myb{\rho}
		\leq{}& \frac{4 \kappa s}{s-1} \int_{\R} \rho^2 + \log \myp{2} T \int_{\R} \rho + T \int_{\R} \rho \log \rho \nn \\
		&+	\begin{cases}
				\dps \frac{2^{3+2\alpha}}{\alpha -1} \kappa r_0^\alpha \int_{\R} \rho^{1+\alpha}	&\text{for $\alpha>1$,}\\[0.4cm]
				\dps 2^5 \kappa r_0 \myp[\bigg]{2 \log \myp{2} \int_{\R} \rho^2+\int_{\R}\rho^2 \myp[\big]{\log r_0 \rho}_+ }	&\text{for $\alpha=1$.}
			\end{cases}
	\label{eq:onedim}
	\end{align}
\end{theorem}
\begin{proof}
	Denoting $ N = \int_{\R} \rho $, we can split the real numbers $ \R $ into two families $ \myp{L_j}_{j=1}^N $, $ \myp{L_j^{\ast}}_{j=1}^N $ of disjoint intervals in such a way that the mass of $ \rho $ in each of these intervals is exactly $ \int_{L_j} \rho = 1/2 $, and such that each $ L_j $ has neighboring intervals only among the $ L_j^{\ast} $, and vice versa (see \cref{fig:intervals}).

\begin{figure}
\begin{tikzpicture}
    \node[anchor=south west,inner sep=0] (image) at (0,0) {\includegraphics[width=0.8\textwidth]{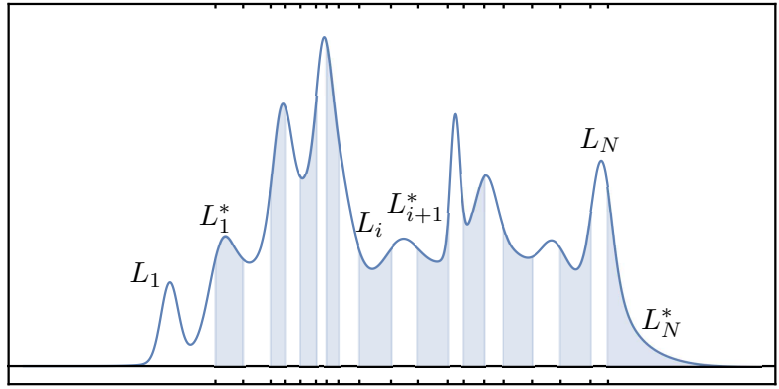}};
    \begin{scope}[x={(image.south east)},y={(image.north west)}]
        \draw (0.18,0.35) node[below] {$L_1$};
        \draw (0.27,0.5) node[below] {$L^*_1$};
                        \draw (0.47,0.48) node[below] {$L_i$};
        \draw (0.53,0.53) node[below] {$L^*_{i+1}$};
                \draw (0.77,0.7) node[below] {$L_N$};
        \draw (0.85,0.23) node[below] {$L^*_N$};
    \end{scope}
\end{tikzpicture}
  \caption{Sketch of intervals}
  \label{fig:intervals}
  \end{figure}

	This allows us to write
	\begin{equation*}
		\rho = \frac{1}{2} \myp[\Big]{ \sum\limits_j 2 \rho \mathds{1}_{L_j} + \sum\limits_j 2 \rho \mathds{1}_{L_j^{\ast}}},
	\end{equation*}
	a convex combination of two measures with mass equal to $ N $.
	As trial states for each of these, we take the symmetric tensor products
	\begin{equation*}
		\bQ = \Pi_s \myp[\Big]{\bigotimes\limits_j \myp{2 \rho \mathds{1}_{L_j}} },
		\qquad
		\bQ^{\ast} = \Pi_s \myp[\Big]{\bigotimes\limits_j \myp{2 \rho \mathds{1}_{L_j^{\ast}}} },
	\end{equation*}
	where $ \Pi_s $ denotes the symmetrization operator in~\eqref{eq:def_Pi_s}.
	Then the state
	\begin{equation*}
		\bP := \frac{1}{2} \myp{\bQ + \bQ^{\ast}}
	\end{equation*}
	has one-body density equal to $ \rho_{\bP} = \rho $.
	Using that the intervals $ \myp{L_j} $ are all disjoint, we have for instance
	\begin{align*}
		- \mathcal{S}_N \myp{\bQ}
		={}& \int_{\R^N} \frac{1}{N!} \sum\limits_{\sigma \in S_N} \bigotimes_j \myp{2 \rho \mathds{1}_{L_{\sigma \myp{j}}}} \log \myp[\Big]{\bigotimes_j \myp{2 \rho \mathds{1}_{L_{\sigma \myp{j}}}}} \\
		={}& \sum\limits_{j=1}^N \int_{\R} 2 \rho \mathds{1}_{L_j} \log \myp{2 \rho \mathds{1}_{L_j}}
		={} \int_{\bigcup_j L_j} 2 \rho \log \myp{2 \rho},
	\end{align*}
	and similarly for $ \bQ^{\ast} $.
	By concavity of the entropy, we conclude that 
	\begin{equation*}
		- \mathcal{S}_N \myp{\bP}
		\leq{} - \frac{1}{2} \mathcal{S}_N \myp{\bQ} - \frac{1}{2} \mathcal{S}_N \myp{\bQ^{\ast}}
		={} \log 2 \int_{\R} \rho + \int_{\R} \rho \log \rho.
	\end{equation*}

	To estimate the interaction energy in the state $ \bP $, it suffices to provide an estimate for both $ \bQ $ and $ \bQ^{\ast} $.
	We write here the argument only for $ \bQ $, since the argument for $ \bQ^{\ast} $ is exactly the same.
	By Assumption~\ref{de:shortrangenew} and the construction of $ \bQ $, we immediately have
	\begin{align*}
		\mathcal{U}_N \myp{\bQ}
		={}& \iint_{\R^2} w \myp{x-y} \rho_{\bQ}^{\myp{2}} \myp{x,y} \id x \id y \\
		\leq{}& 4 \kappa \sum\limits_{i < j} \iint_{\R^2} \myp[\Big]{\frac{r_0^{\alpha} \mathds{1} \myp{\abs{x - y} < r_0}}{\abs{x - y}^{\alpha}} + w_2 \myp{x - y} } \times \nn \\
		&\qquad\qquad\qquad\qquad \times \rho \myp{x} \mathds{1}_{L_i} \myp{x} \rho \myp{y} \mathds{1}_{L_j} \myp{y} \id x \id y
	\end{align*}
	where $ w_2 \myp{x} = \myp{1 + \abs{x}^s}^{-1} $ and $\rho_{\bQ}^{\myp{2}}$ is the two-particle correlation function.
	For the contribution from the tail of the interaction, we have by Young's inequality
	\begin{align*}
		\MoveEqLeft[6] \sum\limits_{i < j} \iint_{\R^2} w_2 \myp{x - y} \rho \myp{x} \mathds{1}_{L_i} \myp{x} \rho \myp{y} \mathds{1}_{L_j} \myp{y} \id x \id y \\
		\leq{}& \frac{1}{2} \iint_{\R^2} w_2 \myp{x-y} \rho \myp{x} \rho \myp{y} \id x \id y \\
		\leq{}& \frac{\normt{w_2}{L_1}}{2} \int_{\R} \rho^2
		\leq{} \frac{s}{s-1} \int_{\R} \rho^2.
	\end{align*}
	From the core of $ w $ we get
	\begin{align*}
		4 \sum\limits_{i < j} \iint_{\R^2} \frac{\mathds{1} \myp{\abs{x - y} < r_0}}{\abs{x - y}^{\alpha}} \rho \myp{x} \mathds{1}_{L_i} \myp{x} \rho \myp{y} \mathds{1}_{L_j} \myp{y} \id x \id y
		\leq \sum\limits_{i < j} \frac{\mathds{1}_{\rd \myp{L_i,L_j} < r_0}}{\rd \myp{L_i,L_j}^{\alpha}}.
	\end{align*}
	The idea now is to use the intervals $ \myp{L_j^{\ast} } $ to estimate the sum above.
	For each $ i $ we denote by $ \eta_i $ the minimal length of neighboring intervals,
	\begin{equation*}
		\eta_i = \min \Set{\ell_j^{\ast} \given \rd \myp{L_i, L_j^{\ast}}= 0},
	\end{equation*}
	where $ \ell_j^{\ast} := \abs{L_j^{\ast}} $ is the interval length, and we re-order the collection $ \myp{L_i} $ such that $ \eta_1 \leq \cdots \leq \eta_N $.
	Fixing the index $ i $, we now clearly have for $ j > i $,
	\begin{equation*}
		\rd \myp{L_i, L_j} \geq \eta_j \geq \eta_i,
	\end{equation*}
	in particular, $ \eta_i $ is smaller than the side length of any interval neighboring $ L_j $.
	Pick $ x_j \in \overline{L_i} $ and $ y_j \in \overline{L_j} $ such that $ \rd \myp{L_i, L_j} = \abs{x_j-y_j} $, and let $ L_k^{\ast} $ be the neighboring interval of $ L_j $ facing $ y_j $, that is, $ \rd \myp{y_j, L_k^{\ast}} = 0 $.
	Defining
	\begin{equation*}
		\widetilde{L}_j := \myp{y_j-\eta_i/2, y_j+\eta_i /2} \cap L_k^{\ast},
	\end{equation*}
	then $ \abs{\widetilde{L}_j} = \eta_i / 2 $, and $ \eta_i /2 \leq \abs{x_j-y} \leq \abs{x_j-y_j} $ for all $ y \in \widetilde{L}_j $, so we can estimate
	\begin{equation*}
		\frac{\mathds{1}_{\rd \myp{L_i,L_j} < r_0}}{\rd \myp{L_i,L_j}^{\alpha}}
		\leq{} \frac{2}{\eta_i} \int_{\widetilde{L}_j} \frac{\mathds{1} \myp{\abs{x_j - y} < r_0 }}{\abs{x_j - y}^{\alpha}} \id y
		= \frac{2}{\eta_i} \int_{\widetilde{L}_j-x_j} \frac{\mathds{1} \myp{\abs{y} < r_0 }}{\abs{y}^{\alpha}} \id y.
	\end{equation*}
	Now summing over $ j $ gives
	\begin{align}
		\sum\limits_{j = i+1}^N \frac{\mathds{1}_{\rd \myp{L_i,L_j} < r_0}}{\rd \myp{L_i,L_j}^{\alpha}}
		\leq{}& \frac{2}{\eta_i} \int_{\R} \frac{\mathds{1} \myp{\eta_i /2 \leq \abs{y} < r_0 }}{\abs{y}^{\alpha}} \id y
		={} \frac{4}{\eta_i} \myp[\bigg]{ \int_{\eta_i / 2}^{r_0} \frac{1}{\abs{y}^{\alpha}} \id y }_+ \nn \\
		\leq{}&	\begin{cases}
					\dps \frac{2^{1+\alpha}}{\alpha -1} \frac{1}{\eta_i^{\alpha}}	&\text{for $\alpha>1$,}\\[0.4cm]
					\dps \frac{4}{\eta_i} \myp[\Big]{\log \myp[\Big]{\frac{2 r_0}{\eta_i}} }_+	&\text{for $\alpha=1$.}
				\end{cases}
	\label{eq:onedim_interact}
	\end{align}
	By Hölder's inequality, we have by construction of the intervals $ L_j^{\ast} $
	\begin{equation*}
		\frac{1}{\myp{\ell_j^{\ast}}^{\alpha}} \leq 2^{1+\alpha} \int_{L_j^{\ast}} \rho^{1+\alpha},
	\end{equation*}
	for any $ \alpha > 1 $, so in this case we conclude that
	\begin{equation*}
		\sum\limits_{i < j} \frac{\mathds{1}_{\rd \myp{L_i,L_j} < r_0}}{\rd \myp{L_i,L_j}^{\alpha}}
		\leq{} \sum\limits_{i=1}^N \frac{2^{1+\alpha}}{\alpha -1} \frac{1}{\eta_i^{\alpha}}
		\leq{} \frac{2^{2+\alpha}}{\alpha -1} \sum\limits_{i=1}^N \frac{1}{\myp{\ell_i^{\ast}}^{\alpha}}
		\leq{} \frac{2^{3+2\alpha}}{\alpha -1} \sum\limits_{i=1}^N \int_{L_i^{\ast}} \rho^{1+\alpha}.
	\end{equation*}
	The same bound holds for the interaction energy of $ \bQ^{\ast} $, but with the intervals $ L_i^{\ast} $ replaced by $ L_i $ at the end.
	This finishes the proof of the $ \alpha > 1 $ case in \eqref{eq:onedim}.

	To finish the $ \alpha = 1 $ case, we note that applying Jensen's inequality on the function $ t \mapsto t^2 \myp{ \log \myp{2\lambda t}}_+ $ for $ \lambda > 0 $ yields
	\begin{equation*}
		\frac{1}{\ell_j^{\ast}} \myp[\Big]{\log \myp[\Big]{ \frac{\lambda}{\ell_j^{\ast}} }}_+
		={} 4 \ell_j^{\ast} \myp[\Big]{\frac{1}{\ell_j^{\ast}} \int_{L_j^{\ast}} \rho}^2 \myp[\Big]{ \log \myp[\Big]{ \frac{2 \lambda}{\ell_j^{\ast}} \int_{L_j^{\ast}} \rho }}_+
		\leq{} 4 \int_{L_j^{\ast}} \rho^2 \myp{ \log \myp{2 \lambda \rho} }_+.
	\end{equation*}
	Hence, continuing from \eqref{eq:onedim_interact}, we get
	\begin{align*}
		\sum\limits_{i < j} \frac{\mathds{1}_{\rd \myp{L_i,L_j} < r_0}}{\rd \myp{L_i,L_j}^{\alpha}}
		\leq{} \sum\limits_{i=1}^N \frac{4}{\eta_i} \myp[\Big]{\log \myp[\Big]{\frac{2 r_0}{\eta_i}} }_+
		\leq{} &8 \sum\limits_{i=1}^N \frac{1}{\ell_i^{\ast}} \myp[\Big]{\log \myp[\Big]{\frac{2 r_0}{\ell_i^{\ast}}} }_+ \\
		\leq{}& 2^5 \sum\limits_{i=1}^N \int_{L_i^{\ast}} \rho^2 \myp{ \log \myp{4 r_0 \rho} }_+.
	\end{align*}
	Since the corresponding bound also holds for $ \bQ^{\ast} $, this concludes the proof.
\end{proof}

\begin{remark}[Hard-core case]\label{rmk:Hard-core-1D}
	In the case where $ w $ has a hard-core with range $ r_0 > 0 $, it follows from the proof above that
	\begin{equation}
	\label{eq:onedim_hardcore}
		F_T \myb{\rho}
		\leq \myp[\Big]{\frac{4 \kappa s}{\myp{s-1} r_0} +\log \myp{2} T} \int_{\R} \rho + T \int_{\R} \rho \log \rho
	\end{equation}
	for any density $ \rho \in L^1 \myp{\R} $ satisfying the (sub-optimal) condition $\int_{x}^{x+r_0} \rho \leq \frac{1}{2}$ for all $ x \in \R$.
\end{remark}

\section{Proof of Theorem~\ref{thm:GC_bound} in the grand-canonical case}\label{sec:proof_GC}

In the course of our proof we need to cover the support of our density using disjoint cubes separated by a distance depending on the local value of the density, in order to have a reasonable control of the interaction. We obtain such a covering by a variant of the Besicovitch lemma~\cite{Guzman}, which we first describe in this subsection. It is different from the standard formulation.

For simplicity we work with a compactly supported density $\rho$ with $\int_{\R^d}\rho>1$. For every $x\in\R^d$, we define $\ell(x)$ to be the \emph{largest} number such that
\begin{equation}
\int_{x+\ell(x)\mathcal C}\rho(x)\,\dx=\frac1{3^d(4^d+1)},
\label{eq:ell}
\end{equation}
where $\mathcal C=(-1/2,1/2)^d$ is the unit cube centered at the origin. It is convenient to work with cubes instead of balls. It is important that the chosen value of the integral in~\eqref{eq:ell} is universal and only depends on the space dimension $d$. This value is motivated by the estimates which will follow, it could be any fixed number $<1$ at this point.  The number $\ell(x)$ always exists since the full integral is larger than 1. The function $x\mapsto \ell(x)$ is upper semi-continuous. To simplify our notation we denote by $\mathcal C(x):=x+\ell(x)\mathcal C$ the cube centered at $x$ of side length $\ell(x)$. By H\"older's inequality we get
$$\frac1{3^d(4^d+1)}=\int_{\mathcal C(x)}\rho\leq \ell(x)^{\frac{\alpha d}{\alpha+d}} \myp[\bigg]{ \int_{\mathcal C(x)}\rho^{1+\frac{\alpha}d} }^{\frac{d}{d+\alpha}}$$
and thus obtain the estimate
\begin{equation}
\frac1{\ell(x)^\alpha}\leq 3^{\alpha+d}(4^d+1)^{1+\frac{\alpha}{d}}\int_{\mathcal C(x)}\rho^{1+\frac\alpha{d}},\qquad \forall x\in\R^d
\label{eq:estim_ell}
\end{equation}
on the local length $\ell(x)$. The standard Besicovitch covering lemma (as stated for instance in~\cite{FraLapWei-LT,Guzman}) implies for compactly supported densities that there exists a set of points $x_j'^{(k)}$ with $1\leq k\leq K'\leq 4^d+1$ and $1\leq j\leq J_k$ such that
\begin{itemize}
 \item the cubes $\big(\mathcal C(x_j'^{(k)})\big)_{\substack{1\leq k\leq K'\\ 1\leq j\leq J_k}}$ cover the support of $\rho$ and each $x\in\R^d$ is in at most $2^d$ such cubes,
 \item for every $k$, the cubes $\big(\mathcal C(x'^{(k)}_j)\big)_{1\leq j\leq J_k}$ are all disjoint.
\end{itemize}
We need to obtain different families which satisfy additional properties, namely we require the cubes to have a safety distance to all the larger cubes within the same family, this distance being comparable to the side length of the cube in question. The precise statement is the following.

\begin{lemma}[Besicovitch with minimal distance]
\label{lem:Besicovitch_lemma}
Let $\rho$ be a compactly supported density with $\int_{\R^d}\rho>1$. Then there exists a set of points $x_j^{(k)}$ with $1\leq k\leq K\leq 3^d(4^d+1)$ and $1\leq j\leq J_k<\ii$ such that
\begin{itemize}
 \item the cubes $\big(\mathcal C(x_j^{(k)})\big)_{\substack{1\leq k\leq K\\ 1\leq j\leq J_k}}$ cover the support of $\rho$ and each $x\in\R^d$ is in at most $2^d$ such cubes,
 \item for every $k$, the cubes $\big(\mathcal C(x^{(k)}_j)\big)_{1\leq j\leq J_k}$ in the $k$th collection satisfy
 $$\rd\left(\mathcal C(x^{(k)}_j),\mathcal C(x^{(k)}_\ell)\right)\geq\frac 1 2\min\Big\{\ell(x^{(k)}_j),\ell(x^{(k)}_\ell)\Big\}.$$
\end{itemize}
\end{lemma}

\begin{proof}
We start the proof by applying the standard Besicovitch covering lemma recalled above. We obtain $K$ collections of disjoint cubes. To impose the minimal distance we separate each family into $3^d$ subfamilies. Specifically, we use that the maximal number of disjoint cubes of side length $\geq\ell$ intersecting a cube of side length $2\ell$ is at most $3^d$. Thus if we look at a given cube of side length $\ell$, only $3^d-1$ other bigger cubes can be at distance $\leq\ell/2$. By induction we can thus always distribute all our cubes into $3^d$ subfamilies, while ensuring the distance property for all the bigger cubes.
\end{proof}

Using Lemma \ref{lem:Besicovitch_lemma} we obtain the following partition of unity
\begin{equation}
\1_{\supp \rho}=\sum_{k=1}^{K}\sum_{j=1}^{J_k}\frac{\1_{\mathcal C(x^{(k)}_j) \cap \supp \rho}}{\eta},
\qquad \1_{\supp \rho} \leq \eta:=\sum_{k=1}^{K}\sum_{j=1}^{J_k}\1_{\mathcal C(x^{(k)}_j)}\leq 2^d
\label{eq:partition}
\end{equation}
which we are going to use to construct our trial state for the upper bound on $G_T[\rho]$. We split the proof into several steps. We start with the case $\alpha>d$ and treat the special case $\alpha=d$ at the very end.

\medskip

\noindent \textbf{Step 1. Less than one particle.} If $\int_{\R^d}\rho\leq1$, we consider the probability $\bP=(\bP_n)$ given by
\begin{equation}
 \bP_0=1-\int_{\R^d} \rho,\qquad\bP_1=\rho,\qquad \bP_n=0\text{ for $n\geq2$,}
 \label{eq:simple_less1}
\end{equation}
which has density $\rho$ and no interaction energy. Its free energy is thus just equal to the entropy term
$$-T \cS \myp{\bP} = T\left(1-\int_{\R^d}\rho\right)\log\left(1-\int_{\R^d}\rho\right)+T\int_{\R^d}\rho\log\rho.$$
The first term is negative and thus we obtain the desired inequality
\begin{equation}
G_T[\rho]\leq T\int_{\R^d}\rho\log\rho\qquad \text{for}\quad \int_{\R^d}\rho\leq1.
\label{eq:esimple_upper_one_part}
\end{equation}

\medskip

\noindent \textbf{Step 2. Compactly supported densities ($\alpha>d$).} Next we consider the case of a compactly supported density $\rho$ with $\int_{\R^d}\rho>1$.
Using the partition~\eqref{eq:partition} we write
$$\rho=\frac1K\sum_{k=1}^{K} \myp[\bigg]{ \sum_{j}\rho_j^{(k)}} ,\qquad\rho_j^{(k)}:=\frac{K\rho\1_{Q_j^{(k)}}}{\eta},$$
where we abbreviated $Q_j^{(k)}=\mathcal C(x^{(k)}_j)$ for simplicity. This is a (uniform) convex combination of the $K$ densities $\rho^{(k)}=\sum_{j}\rho_j^{(k)}$.
For fixed $k$, the $\rho_j^{(k)}$ have disjoint supports with distance greater or equal to $\min\{\ell(x_j^{(k)}),\ell(x_{j'}^{(k)})\}/2$. In addition, we have
$$\int\rho_j^{(k)}=K\int_{Q_j^{(k)}} \frac{\rho}{\eta}\leq 3^d(4^d+1)\int_{Q_j^{(k)}}\rho\leq 1.$$
This is the reason for our choice of the constant in~\eqref{eq:ell}.
Our trial state is given by
$$\bP:=\frac1K\sum_{k=1}^K\bP^{(k)}$$
where $$\bP^{(k)}=\bigotimes_{j=1}^{J_k}\left(\left(1-\int_{\R^d}\rho_j^{(k)}\right)\oplus\rho_j^{(k)}\oplus0\oplus\ldots\right)$$
is the symmetrized tensor product of the states in~\eqref{eq:simple_less1}, which has density $\rho^{(k)}$. Using the concavity of the entropy, our upper bound is, thus,
\begin{align*}
G_T[\rho]&\leq \frac1K\sum_k \mathcal{G}_T(\bP^{(k)})\\
&\leq \frac1{K}\sum_{k=1}^{K}\bigg(\sum_{1\leq i<j\leq J_k}\iint_{\R^d\times\R^d}\rho_i^{(k)}(x)\rho_j^{(k)}(y)w(x-y)\,\dx\,\dy\\
&\qquad +T\sum_{j=1}^{J_k}\int_{Q_j^{(k)}}\rho_j^{(k)}\log\rho_j^{(k)}\bigg).
\end{align*}
We have
\begin{align*}
\frac{1}{K}\sum_{k=1}^{K}\sum_{j=1}^{J_k}\int_{Q_j^{(k)}}\rho_j^{(k)}\log\rho_j^{(k)}
&=\frac{1}{K}\sum_{k=1}^{K}\sum_{j=1}^{J_k}\int_{Q_j^{(k)}}\rho_j^{(k)}\log\frac{K\rho}\eta\\
&=\int_{\R^d}\rho\log\frac{K\rho}\eta
\leq\int_{\R^d}\rho\log\rho+3d\int_{\R^d}\rho
\end{align*}
since $K\leq 15^d\leq e^{3d}$ and $\eta\geq1$. Thus we obtain
\begin{multline*}
G_T[\rho]\leq  \frac1K\sum_{k=1}^{K}\sum_{1\leq i<j\leq J_k}\iint_{\R^d\times\R^d}\rho_i^{(k)}(x)\rho_j^{(k)}(y)w(x-y)\,\dx\,\dy\\
+T\int_{\R^d}\rho\log\rho+3Td\int_{\R^d}\rho.
\end{multline*}
Our next task is to estimate the interaction, for every fixed $k$. By Assumption~\ref{de:shortrangenew} we have $w\leq w_1+w_2$ with $w_1(x)=\kappa (r_0/|x|)^{\alpha}\1(|x|<r_0)$ and $w_2(x)=\kappa(1+|x|^s)^{-1}$. We first estimate the term involving the integrable potential $w_2$ using Young's inequality as
\begin{align*}
&\sum_{1\leq i<j\leq J_k}\iint_{\R^d\times\R^d}\rho_i^{(k)}(x)\rho_j^{(k)}(y)w_2(x-y)\,\dx\,\dy\\
&\qquad \leq \frac12\iint_{\R^d\times\R^d}\rho^{(k)}(x)\rho^{(k)}(y)w_2(x-y)\,\dx\,\dy\\
&\qquad\leq \frac{\norm{w_2}_{L^1}}2\int_{\R^d}(\rho^{(k)})^2= \frac{\norm{w_2}_{L^1}}2 K^2\int_{\cup_i Q_i^{(k)}}\frac{\rho^2}{\eta^2}.
\end{align*}
After summing over $k$ this gives
\begin{equation*}
\frac1K\sum_{k=1}^K\sum_{1\leq i<j\leq J_k}\iint_{\R^d\times\R^d}\rho_i^{(k)}(x)\rho_j^{(k)}(y)w_2(x-y)\,\dx\,\dy
\leq \frac{\norm{w_2}_{L^1}}2 K\int_{\R^d}\frac{\rho^2}{\eta}.
\end{equation*}
Using for instance
$$\int_{\R^d}w_2=\kappa|\bS^{d-1}|\int_0^\ii \frac{r^{d-1}}{1+r^s}\,\rd r\leq \kappa|\bS^{d-1}|\frac{s}{d \myp{s-d}},$$
and recalling that $\eta\geq1$ and $K\leq 3^d(4^d+1)$, we obtain
\begin{multline*}
\frac1K \sum_{k=1}^{K}\sum_{1\leq i<j\leq J_k}\iint_{\R^d\times\R^d}\rho_i^{(k)}(x)\rho_j^{(k)}(y)w_2(x-y)\,\dx\,\dy\\
\leq \kappa \frac{s|\bS^{d-1}|}{2d(s-d)}3^{d}(4^d+1)\int_{\R^d}\rho^2.
\end{multline*}

Next we consider the more complicated term involving the singular part $w_1=\kappa\1(|x|<r_0) (r_0/|x|)^{\alpha}$. To simplify our notation, we remove the superscript $(k)$ and thus consider the collection $(\rho_j)_{j=1}^J$ of functions supported in the disjoint cubes $Q_j$ with the safety distance. For every $i\neq j$, using $\int_{\R^d}\rho_j\leq 1$, we can estimate
$$\iint\rho_i(x)\rho_j(y)w_1(x-y)\,\dx\,\dy\leq \frac{\kappa r_0^\alpha}{\rd (Q_i,Q_j)^\alpha}.$$
Recall that when $|Q_i|\leq |Q_j|$, the distance $\rd (Q_i,Q_j)$ is at least equal to $\ell_i/2$. We can order our $J$ cubes so that the volume is increasing: $|Q_1|\leq |Q_2|\leq\cdots\leq |Q_J|$. We need to estimate
$$\sum_{i=1}^{J-1}\sum_{j=i+1}^J\frac{1}{\rd (Q_i,Q_j)^\alpha}=\sum_{i=1}^{J-1}\frac1{\ell_i^\alpha}\sum_{j=i+1}^J\frac{1}{\rd (\mathcal C,Q'_{i,j})^\alpha}$$
where $\mathcal C=(-1/2,1/2)^d$ and for every $i$, we have denoted by $Q'_{i,j}$  the cube centered at $(x_j-x_i)/\ell_i$, of volume $|Q_j|/|Q_i|\geq1$. To estimate the sum in $j$, we use the following lemma, which is based on the integrability at infinity of $|x|^{-\alpha}$ and is similar to~\cite[Lemma~9]{Lewin-22}.

\begin{lemma}\label{lem:cubes}
Let $\mathcal C=(-1/2,1/2)^d$ be the unit cube and consider any collection of non-intersecting cubes $Q_j$ with the property that $|Q_j|\geq1$ and $\rd (\mathcal C,Q_j)\geq\frac12$. Then we have
\begin{equation}
\sum_j \frac{1}{\rd (\mathcal C,Q_j)^\alpha}\leq \frac{3^\alpha2^{7d}d^2}{|\bS^{d-1}|(\alpha-d)}\,.
\label{eq:estim_sum}
\end{equation}
\end{lemma}

The constant on the right of~\eqref{eq:estim_sum} is not at all optimal and is only displayed for concreteness.

\begin{proof}[Proof of Lemma~\ref{lem:cubes}]
Let $X_j\in\mathcal C$ and $Y_j\in Q_j$ be so that $\rd(\mathcal C,Q_j)=|X_j-Y_j|\geq\frac12$.
For any $x\in B(X_j,1/8)$ and $y\in B(Y_j,1/8)$ we have
$$\frac{|X_j-Y_j|}2\leq |X_j-Y_j|-\frac14\leq |x-y|\leq |X_j-Y_j|+\frac14\leq \frac32|X_j-Y_j|.$$
Integrating over $x'\in \mathcal C\cap B(X_j,1/8)$ and $y'\in Q_j\cap B(Y_j,1/8)$ we obtain
\begin{align*}
\frac1{\rd (\mathcal C,Q_j)^\alpha}&=\frac1{|X_j-Y_j|^\alpha}\\
&\leq \frac{(3/2)^\alpha}{|\mathcal C\cap B(X_j,1/8)|\;| Q_j\cap B(Y_j,1/8)|}\int_\mathcal C\int_{Q_j}\frac{\dx\,\dy}{|x-y|^\alpha}.
\end{align*}
The volume of the intersection of a ball of radius $1/8$ centered at $X_j$ in a cube of volume $\geq1$ and that of the other cube is bounded away from 0. It is in fact minimal when $X_j, Y_j$ are located at a corner, yielding
$$|\mathcal C\cap B(X_j,1/8)|\geq \frac{|\bS^{d-1}|}{2^{4d}d},\qquad |Q_j\cap B(Y_j,1/8)|\geq \frac{|\bS^{d-1}|}{2^{4d}d}.$$
Thus, we obtain
$$\frac1{\rd (\mathcal C,Q_j)^\alpha}\leq \frac{3^\alpha 2^{8d-\alpha}d^2}{|\bS^{d-1}|^2}\int_\mathcal C\int_{Q_j}\frac{\dx\,\dy}{|x-y|^\alpha}.$$
Summing over $j$ using that the cubes are disjoint we obtain
$$\sum_j \frac1{\rd (\mathcal C,Q_j)^\alpha}\leq \frac{3^\alpha 2^{8d-\alpha }d^2}{|\bS^{d-1}|^2}\int_\mathcal C\int_{\R^d}\frac{\1(|x-y|\geq\frac12) \,\dx\,\dy}{|x-y|^\alpha}= \frac{3^\alpha2^{7d}d^2}{|\bS^{d-1}|(\alpha-d)}$$
as was claimed.
\end{proof}

From the estimates~\eqref{eq:estim_ell} and \eqref{eq:estim_sum}, we deduce that
\begin{multline*}
\sum_{1\leq i<j\leq J_k}\iint_{|x-y|\leq r_0}\rho_i(x)\rho_j(y)w_1(x-y)\,\dx\,\dy\\
\leq \kappa r_0^\alpha \frac{d^23^{d+2\alpha} 2^{7d}(4^d+1)^{1+\frac{\alpha}{d}}}{|\bS^{d-1}|(\alpha-d)}\int_{\cup_iQ_i}\rho^{1+\frac\alpha{d}}.
\end{multline*}
Using that $\int_{\cup_iQ_i}\rho^{1+\frac\alpha{d}} \leq \int_{\R^d} \rho^{1+ \frac{\alpha}{d}}$ and summing over $ K $, we obtain our final estimate
\begin{multline}
G_T[\rho]\leq  \kappa \frac{s|\bS^{d-1}|}{2d(s-d)}3^{d}(4^d+1)\int_{\R^d}\rho^2 +\kappa r_0^\alpha \frac{d^23^{d+2\alpha} 2^{7d}(4^d+1)^{1+\frac{\alpha}{d}}}{|\bS^{d-1}|(\alpha-d)}\int_{\R^d}\rho^{1+\frac\alpha{d}}\\
+T\int_{\R^d}\rho\log\rho+3Td\int_{\R^d}\rho.
\label{eq:estim_final_proof}
\end{multline}
This is our final upper bound, with non-optimal constants only displayed for concreteness.

\medskip

\noindent \textbf{Step 3. General densities  ($\alpha>d$).}
In order to be able to use the Besicovitch lemma, we restricted ourselves to compactly supported densities. We prove here that the exact same estimate holds for general densities. Let $\rho\in (L^1\cap L^{1+\alpha/d})(\R^d,\R_+)$, $\eps\in(0,1)$ and write
$$\rho=(1-\eps)\frac{\rho\1_{\mathcal C_L}}{1-\eps}+\eps\frac{\rho\1_{\R^d\setminus \mathcal C_L}}{\eps}$$
with $\cC_L=(-L/2,L/2)^d$.
Using the concavity of the entropy, we obtain
\begin{equation}
G_T[\rho]\leq (1-\eps)G_T\left[\frac{\rho\1_{\mathcal C_L}}{1-\eps}\right]+\eps\, G_T\left[\frac{\rho\1_{\R^d\setminus \mathcal C_L}}{\eps}\right].
\label{eq:estim_support_convexity}
\end{equation}
We choose $L$ so large that
$$\int_{\R^d\setminus \mathcal C_L}\rho\leq\eps,$$
which allows us to use~\eqref{eq:esimple_upper_one_part} for the second term on the right of~\eqref{eq:estim_support_convexity}. For the first term we just use Step 2. We find
\begin{multline*}
G_T[\rho]\leq \frac{C\kappa r_0^\alpha}{(1-\eps)^{\frac\alpha{d}}}\int_{\mathcal C_L}\rho^{1+\frac\alpha{d}}+\frac{C\kappa}{1-\eps}\int_{\mathcal C_L}\rho^{2}+CT\int_{\mathcal C_L}\rho+T\int_{\R^d}\rho\log\rho\\
+T\log\eps^{-1}\int_{\R^d\setminus\mathcal  C_L}\rho+T\log(1-\eps)^{-1}\int_{\mathcal C_L}\rho.
\end{multline*}
By passing first to the limit $L\to\ii$ and then $\eps\to0$, we conclude that $\rho$ satisfies the same estimate~\eqref{eq:estim_final_proof} as for compactly supported densities.

\medskip

\noindent \textbf{Step 4. Case $\alpha=d$.}
The case when the core of the interaction behaves as $w_1(x)=\kappa r_0^d |x|^{-d} \1 \myp{|x|\leq r_0}$ is similar to the previous situation with some small changes. The function is not integrable around the origin which requires to have a safety distance between particles in our trial state. However this interaction is also non-integrable without cutoff at infinity so we need to use that the core of our interaction is compactly supported on the ball of radius $r_0$. The following alternative to Lemma \ref{lem:cubes} is going to be useful.

\begin{lemma}\label{lem:cubes_d=a}
Let $\mathcal C_0=(-\ell_0/2,\ell_0/2)^d$ and consider any collection of non-intersecting cubes $Q_j$ with the property that $|Q_j|\geq\ell_0^d$ and $\rd (\mathcal C_0,Q_j)\geq\frac{\ell_0}{2}$. Then we have
\begin{equation}
\sum_j \frac{\1_{\rd (\mathcal C_0,Q_j)\leq r_0}}{\rd(\mathcal C_0,Q_j)^d}\leq C\ell_0^{-d}\left(\log\left(\frac{2r_0}{\ell_0}\right)\right)_+\,.
\label{eq:estim_sum_d=a}
\end{equation}
\end{lemma}

\begin{proof}
We assume $\ell_0\leq 2r_0$ otherwise there is nothing to prove.
Let $X_j\in\mathcal C_0$ and $Y_j\in Q_j$ be such that  $\rd(\mathcal C_0,Q_j)=|X_j-Y_j|\geq\frac{\ell_0}2$.
For any $x\in B(X_j,\ell_{0}/8)$ and $y\in B(Y_j,\ell_{0}/8)$ we have
$$\frac{|X_j-Y_j|}2\leq |X_j-Y_j|-\frac{\ell_0}4\leq |x-y|\leq |X_j-Y_j|+\frac{\ell_0}4\leq \frac32|X_j-Y_j|.$$
Integrating over $x'\in \mathcal C_0\cap B(X_j,\ell_{0}/8)$ and $y'\in Q_j\cap B(Y_j,\ell_{0}/8)$ we obtain
\begin{align*}
\MoveEqLeft[4] \frac{\1_{\rd (\mathcal C_0,Q_j)\leq r_0}}{\rd (\mathcal C_0,Q_j)^d}
={} \frac{\1_{\rd (\mathcal C_0,Q_j)\leq r_0}}{|X_j-Y_j|^d}\\
\leq{}& \frac{(3/2)^d}{|\mathcal C_0\cap B(X_j,\ell_0/8)|\;|Q_j\cap B(Y_j,\ell_0/8)|}\int_{\mathcal C_0}\int_{Q_j}\frac{\1_{\rd (|x-y|)\leq r_0}}{|x-y|^d}\,\dx\,\dy.
\end{align*}
Summing over all cubes we get
\begin{equation*}
\sum_{j}\frac{\1_{\rd (\mathcal C_0,Q_j)\leq r_0}}{\rd (\mathcal C_0,Q_j)^d}
\leq \frac{2^{7d}3^dd}{|\mathbb S^{d-1}|\ell_0^{d}}\int_{\frac{\ell_0}2}^{r_0}r^{-1} \,\rd r
=\frac{2^{7d}3^dd}{|\mathbb S^{d-1}|}\frac{\log(2r_0/\ell_0)}{\ell_0^{d}}.
\end{equation*}
\end{proof}

Next we explain how to relate the right side of~\eqref{eq:estim_sum_d=a} with the density $\rho$. Recall from \eqref{eq:ell} that
\begin{equation}
\int_{\cC(x)}\rho(y)\,\dy=\frac1{3^d(4^d+1)}
\label{eq:elld}
\end{equation}
where $\cC(x)$ is the cube of side length $\ell(x)$ centered at $x$. By Jensen's inequality, we have for every convex function $F$
\begin{equation}
\ell(x)^dF\left(\frac{1}{\ell(x)^d3^d(4^d+1)}\right)
= \ell\myp{x}^d F \myp[\bigg]{\frac{1}{\ell \myp{x}^d} \int_{\mathcal{C} \myp{x}} \rho}
\leq \int_{\cC(x)}F\big(\rho(y)\big)\,\dy.
 \label{eq:convex_F}
\end{equation}
Applying this to
$$F(t)=t^2\Big(\log \big(6^d(4^d+1)r_0^dt\big)\Big)_+,$$
we obtain
\begin{align}
&\ell(x)^{-d}\left(\log\left(\frac{2r_0}{\ell(x)}\right)\right)_+\nn\\
&\qquad\leq  \frac{3^{2d}(4^d+1)^2}{d}\int_{\cC(x)}\rho(y)^2\Big(\log\big(6^d(4^d+1)r_0^d\rho(y)\big)\Big)_+\,\dy\nn\\
&\qquad\leq  \frac{3^{2d}(4^d+1)^2}{d}\int_{\cC(x)}\rho(y)^2\Big(4d+\big(\log r_0^d\rho(y)\big)_+\Big)\,\dy.
 \label{eq:estim_log_rho}
\end{align}
The rest of the proof is similar to the case $\alpha>d$, using~\eqref{eq:estim_log_rho} and Lemma~\ref{lem:cubes_d=a}. We omit the details. This concludes the proof of Theorem~\ref{thm:GC_bound}.\qed

\begin{remark}[Hard-core case]\label{rmk:Hard-core-GC}
The previous proof can be used in the hard core case $\alpha=\infty$, under the (sub-optimal) condition that $\int_{x+r_0\mathcal C}\rho<\frac{1}{3^d(4^d+1)}$ for all $x$, where $\mathcal C=(-1/2,1/2)^d$. The interaction can be bounded by $\kappa C N$ as we have seen in~\eqref{eq:simple_hard_core_T0}, leading to the bound
		\begin{align*}
		G_T [\rho]&\leq  C\kappa\int_{\R^d}\rho+CT \int_{\R^d} \rho + T \int_{\R^d} \rho \log \rho\,.
		\end{align*}
\end{remark}

\section{Proofs in the canonical case}\label{sec:proof_C}

\subsection{The local radius \texorpdfstring{$R(x)$}{R(x)} in optimal transport}\label{sec:OT}
Here we explain how to construct canonical trial states using a result from optimal transport, in order to obtain bounds at zero temperature for a singular interaction ($d\leq\alpha\leq\ii$).

Consider any density $ \rho $ with $ \int \rho > 1 $, and recall the local radius $ R(x) $ from~\eqref{eq:Rdef}.
Note that $ R \myp{x} $ can never be zero because $ \rho $ as a measure does not have any point mass.
The function $ R $ is connected to the Hardy-Littlewood maximal function $ M_\rho $, defined by
\begin{equation}
	\label{eq:maximaldef}
	M_\rho \myp{x} := \sup_{r > 0} \frac{1}{\abs{B_r}} \int_{B \myp{x,r}} \rho \myp{y} \id y,
\end{equation}
where $ \abs{B_r} $ denotes the volume of a ball in $ \R^d $ of radius $ r $.
By definition of $ R $ it is clear that
\begin{equation*}
	\frac{1}{\abs{B_1} R \myp{x}^d}
	= \frac{1}{\abs{B_1} R \myp{x}^d} \int_{B \myp{x, R \myp{x}}} \rho \myp{y} \id y
	\leq M_\rho \myp{x},
\end{equation*}
so that we have the pointwise bound
\begin{equation}
\label{eq:Rmaximalbound}
	\frac{1}{R \myp{x}}
	\leq \myp{\abs{B_1} M_\rho \myp{x}}^{\frac{1}{d}}.
\end{equation}
Furthermore, using Hölder's inequality gives for any $ p > 0 $,
\begin{equation*}
	1 = \int_{B \myp{x, R \myp{x}}} \rho
	\leq \abs{B \myp{x, R \myp{x}}}^{\frac{p}{p+d}} \myp[\Big]{\int_{B \myp{x,R \myp{x}}} \rho^{1+\frac{p}{d}}}^{\frac{d}{p+d}},
\end{equation*}
implying for $ \rho \in L_{\mathrm{loc}}^{1+\frac{p}{d}} \myp{\R^d} $ the bound
\begin{equation}
\label{eq:R_bound}
	\frac{1}{R \myp{x}^p}
	\leq \abs{B_1}^{\frac{p}{d}} \int_{B \myp{x,R \myp{x}}} \rho^{1+\frac{p}{d}}.
\end{equation}
It is also apparent that $ R $ is 1-Lipschitz-continuous, see e.g. \cite[Theorem 4.1]{ColMarStr-19}. One might also remark that $ R $ always stays away from zero, i.e.
\begin{equation}
\label{eq:Rlowbound}
	R_{\rho}
	:= \min_{x \in \R^d} R \myp{x} > 0.
\end{equation}
This is an immediate consequence of the facts that $ R $ is continuous and that, necessarily, $ \lim_{\abs{x} \to \infty} R \myp{x} = \infty $, because $ \rho \in L^1 \myp{\R^d} $.

To obtain an upper bound on the canonical energy at a fixed density $ 0 \leq \rho \in L^1 \myp{\R^d} $, it is convenient to have existence of states $ \bP $ in which the distance between the particles is bounded from below in terms of the function $ R $ from \eqref{eq:Rdef}. The following is a consequence of a result from~\cite{ColMarStr-19}.

\begin{theorem}[Optimal transport state]
	\label{thm:clstate}
	Let $ 0 \leq \rho \in L^1 \myp{\R^d} $ with $ N = \int_{\R^d} \rho \in \N $.
	There exists an $ N $-particle state $ \bP $ with density $ \rho_{\bP} = \rho $ such that
	\begin{equation}
	\label{eq:plansupport}
		\abs{x_i - x_j} \geq \max \Big( R_{\rho}, \tfrac{R(x_i) + R(x_j)}{3} \Big)\quad\text{for $1\leq i\neq j\leq N$}
	\end{equation}
	$\bP$--almost everywhere, where $ R $ is the function defined by \eqref{eq:Rdef}, and $ R_{\rho} $ is its minimum in~\eqref{eq:Rlowbound}.
\end{theorem}

\begin{proof}
	The proof is a simple application of \cite[Theorem 4.3]{ColMarStr-19}.
	For any $ 0 < \eta < 1 $ (we will choose $ \eta = 1/3 $ in a moment) and any $ x \in \R^d $, define a set
	\begin{equation*}
		\widetilde{B} \myp{x} = \Set{ y \in \R^d \given \abs{x-y} < \eta \myp{R \myp{x} + R \myp{y}} }.
	\end{equation*}
	Then we have for any $ 0 < t < 1 -\eta $, using the Lipschitz continuity of $ R $,
	\begin{align*}
		\widetilde{B} \myp{x}
		\subseteq{}& \Set{y \in \R^d \given t \abs{x-y} < \eta R \myp{x} } \cup \Set{y \in \R^d \given \myp{1-t} \abs{x-y} < \eta R \myp{y} } \\
		\subseteq{}& B \myp{x, \tfrac{\eta}{t} R \myp{x}} \cup \Set{y \in \R^d \given \myp{1-t} \abs{x-y} < \eta \myp{ R \myp{x} + \abs{x-y}} } \\
		={}& B \myp{x, \tfrac{\eta}{t} R \myp{x}} \cup B \myp{x, \tfrac{\eta}{1-t-\eta} R \myp{x}}.
	\end{align*}
	We wish to choose $ t $ and $ \eta $ such that the measure of right hand side is equal to one (with respect to the measure $ \rho $).
	First, requiring the two balls to have the same radius leads to the choice $ t = \frac{1-\eta}{2} $.
	Next, we choose $ \eta $ such that $ \frac{\eta}{t} = \frac{\eta}{1-t-\eta} = \frac{2 \eta}{1-\eta} = 1 $, which implies $ \eta = 1/3 $.

	Now, defining an open and symmetric set $ D \subseteq \R^d \times \R^d $ by
	\begin{equation*}
		D = \Set[\Big]{ \myp{x,y} \in \R^d \times \R^d \given \abs{x-y} < \max \left( R_{\rho}, \tfrac{R \myp{x} + R \myp{y}}{3} \right) },
	\end{equation*}
	then $ B \myp{x} := \Set{y \in \R^d \given \myp{x,y} \in D} $ satisfies
	\begin{equation*}
		B \myp{x} = B \myp{x, R_{\rho}} \cup \widetilde{B} \myp{x}
		\subseteq B \myp{x, R \myp{x}}.
	\end{equation*}
	Thus, by definition of $ R $, we have $ \rho \myp{B \myp{x}} \leq \rho \myp{B \myp{x, R\myp{x}}} = 1 $, and since $ D $ is open and symmetric, \cite[Theorem 4.3]{ColMarStr-19} asserts the existence of a $ \bP $ with the claimed properties.
	Specifically, one can take $ \bP $ to be the optimizer for the multi-marginal optimal transport problem associated to the cost
	\begin{equation*}
		c \myp{x} = \min \Set{ \rd \myp{x, A}, 1},
	\end{equation*}
	where $ A $ denotes the set containing the $(x_1,...,x_N)$ satisfying~\eqref{eq:plansupport}.
\end{proof}

\subsection{Proof of Theorem~\ref{thm:GC_bound} in the canonical case at \texorpdfstring{$T=0$}{T=0}}
The existence of the state from \cref{thm:clstate} allows us to prove the last part of \cref{thm:GC_bound} about the canonical free energy at zero temperature. For convenience we state a proposition valid for any state $\bP$ for which the particles satisfy an inequality similar to~\eqref{eq:plansupport}.
Along with \cref{prop:clbound3} below (which covers the case $ \alpha = d $), this immediately implies Theorem~\ref{thm:GC_bound} in the canonical case.

\begin{proposition}[Zero temperature energy bound, $ d < \alpha < \infty $]
	\label{prop:clbound}
	Let $ w $ satisfy Assumption~\ref{de:shortrangenew} with $ d < \alpha < \infty $.
	Let $\bP$ be any $ N $-particle probability measure with one-body density $ \rho:=\rho_{\bP} $ satisfying
	\begin{equation}
	\label{eq:plansupport2}
		\abs{x_i - x_j} \geq \eta \big( R \myp{x_i} + R \myp{x_j}\big) \qquad \text{ for $1\leq  i \neq j\leq N$,}
	\end{equation}
	$\bP$--almost everywhere, for some $ 0 < \eta \leq 1 $. Then the interaction energy in the state $ \bP $ is bounded by
	\begin{equation}
		F_0 \myb{\rho}
		\leq{} \cU_{N} \myp{\bP}
		\leq{} \frac{C\kappa r_0^\alpha}{\eta^{\alpha}} \int_{\R^d} \rho \myp{x}^{1+\frac{\alpha}{d}} \id x + \frac{C\kappa}{\eta^d} \int_{\R^d} \rho \myp{x}^2 \id x
	\label{eq:clenergybound}
	\end{equation}
	with $C$ a constant depending only on $d,\alpha,s$.
\end{proposition}

\begin{proof}[Proof of \cref{prop:clbound}]
In this proof we will not keep track of the exact value of the constants, since we will need the (unknown) one from the Hardy-Littlewood inequality. Hence $C$ denotes here a generic constant depending only on $d,\alpha,s$.
By the assumptions on $ w $, we have
	\begin{equation}
		\cU_{N} \myp{\bP}
		\leq{} \kappa \int_{\R^{dN}} \sum_{1\leq i<j\leq N} \left(\frac{r_0^\alpha\1(|x_i-x_j|\leq r_0)}{\abs{x_i - x_j}^{\alpha}} + \frac{1}{1+ \abs{x_i-x_j}^s}\right) \rd \bP \myp{x}.
	\label{eq:interactionbound2}
	\end{equation}
Let $x= ( x_1, \dotsc, x_N) $ be in the support of $\bP$. After permutation we can assume that $ R \myp{x_1} \leq R \myp{x_2} \leq \cdots \leq R \myp{x_N} $.
	We fix the index $ i $ and consider the points $ x_i - x_j $ in $ \R^d $ for $ j = i+1, \dotsc, N $.
	Because of \eqref{eq:plansupport2}, these points are all at a distance at least $ \eta \myp{R\myp{x_i} + R \myp{x_j}} $ from the origin, and
	\begin{equation*}
		\abs{\myp{x_i - x_j} - \myp{x_i - x_k}} = \abs{x_j - x_k} \geq \eta \myp{R \myp{x_j} + R \myp{x_k}}.
	\end{equation*}
	Hence we can place $ N - i $ disjoint balls in $ \mathbb{R}^d $ with radii $ \eta R \myp{x_j} $, centered at the points $ x_i - x_j $, respectively.
	Inside each of these balls, we place a smaller ball of radius $ \frac{\eta}{2} R \myp{x_j} $, centered at
	\begin{equation*}
		z_j = \myp[\Big]{1- \frac{\eta R \myp{x_j} }{2 \abs{x_i-x_j}}} \myp{x_i - x_j}.
	\end{equation*}
	Then $ x_i - x_j $ is the point on the boundary of $ B \myp{z_j, \frac{\eta}{2} R \myp{x_j} } $ which is the farthest from the origin (see \cref{fig:balls}), so that
	\begin{equation*}
		\frac{1}{\abs{x_i - x_j}^{\alpha}} = \min_{y \in B \myp{z_j, \frac{\eta}{2} R \myp{x_j} } } \frac{1}{\abs{y}^{\alpha}}.
	\end{equation*}

  \begin{figure}
\begin{tikzpicture}
    \node[anchor=south west,inner sep=0] (image) at (0,0) {\includegraphics[width=0.6\textwidth]{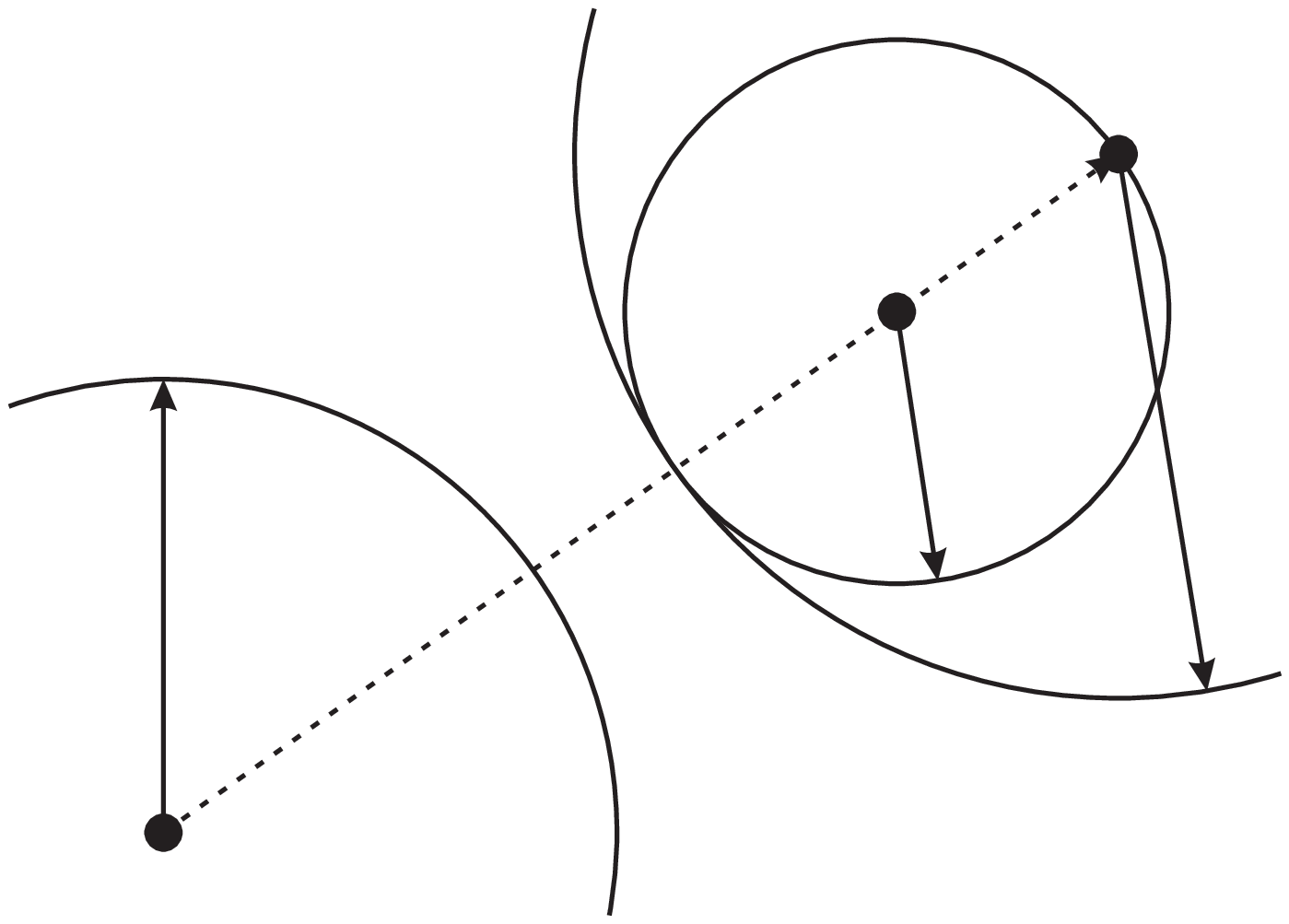}};
    \begin{scope}[x={(image.south east)},y={(image.north west)}]
        \draw (0.1,0.1) node[below] {$0$};
                \draw (.7,.78) node[below] {$z_j$};
                                \draw (.93,.85) node[above] {$x_i-x_j$};
 \draw (0.62,.0) node[above] {$B(0,\eta R(x_i))$};
                 \draw (.7,.96) node[above] {$B \myp{z_j, \frac{\eta}{2} R \myp{x_j} }$};
  \draw (0.22,.4) node[above] {$\eta R(x_i)$};
                                          \draw (0.8,.5) node[above] {$\frac\eta2 R(x_j)$};
                               \draw (1,.45) node[above] {$\eta R(x_j)$};
                                          \draw (0.85,.12) node[above] {$B(x_i-x_j,\eta R(x_j))$};
    \end{scope}
\end{tikzpicture}
  \caption{Sketch of construction}
  \label{fig:balls}
  \end{figure}

	Note that the distance from $ B \myp{z_j, \frac{\eta}{2} R \myp{x_j} } $ to the origin is bounded from below by
	\begin{equation*}
		\rd \left(0, B \myp{z_j, \frac{\eta}{2} R \myp{x_j}} \right)
		\geq \abs{z_j} - \frac{\eta}{2} R \myp{x_j}
		= \abs{x_i-x_j} - \eta R \myp{x_j}
		\geq \eta R \myp{x_i}.
	\end{equation*}
	Using this, along with the fact that all the balls are disjoint, we get the pointwise bound
	\begin{align}
		\sum\limits_{j=i+1}^N \frac{1}{\abs{x_i - x_j}^{\alpha}}
		\leq{}& \sum\limits_{j=i+1}^N \frac{1}{\abs{B \myp{z_j,\frac{\eta}{2} R \myp{x_j}}}} \int_{B \myp{z_j, \frac{\eta}{2} R \myp{x_j}}} \frac{1}{\abs{y}^{\alpha}} \id y \nn \\
		\leq{}& \frac{1}{\abs{B \myp{0,\frac{\eta}{2} R \myp{x_i}}}} \int_{B \myp{0, \eta R \myp{x_i}}^c} \frac{1}{\abs{y}^{\alpha}} \id y \label{eq:interactbound}\\
		={}& \frac{2^d}{\abs{B_1}} \frac{1}{\myp{\eta R \myp{x_i}}^{\alpha}} \int_{B \myp{0, 1}^c} \frac{1}{\abs{y}^{\alpha}} \id y \nn
	\end{align}
	for $ \bP $-a.e. $ x \in \mathbb{R}^{dN} $.
	We conclude that the contribution to the energy from the core of $ w $ can be bounded by
	\begin{align}
		\int_{\R^{dN}} \sum\limits_{i=1}^N \sum\limits_{j = i +1}^N \frac{1}{\abs{x_i - x_j}^{\alpha}} \id \bP \myp{x}
		\leq{}& \frac{C}{\eta^{\alpha}} \int_{\R^{dN}} \sum\limits_{i = 1}^N \frac{1}{R \myp{x_i}^{\alpha}} \id \bP \myp{x} \nn \\
		={}& \frac{C}{\eta^{\alpha}} \int_{\R^d} \frac{\rho \myp{x}}{R \myp{x}^{\alpha}} \id x.
	\label{eq:core_bound}
	\end{align}
	Similarly, we get for the contribution from the tail of $ w $,
	\begin{align*}
		\sum\limits_{j = i + 1} \frac{1}{1 + \abs{x_i - x_j}^s}
		\leq{}& \sum\limits_{j = i + 1} \frac{1}{ \abs{B \myp{z_j, \frac{\eta}{2} R \myp{x_j}}} } \int_{B \myp{z_j, \frac{\eta}{2} R \myp{x_j}}} \frac{1}{1 + \abs{y}^s} \id y \\
		\leq{}& \frac{2^d}{\abs{B_1}} \frac{1}{\myp{\eta R \myp{x_i}}^d} \int_{\R^d} \frac{1}{1 + \abs{y}^s} \id y,
	\end{align*}
	so
	\begin{equation}
	\label{eq:tail_bound}
		\int_{\R^{dN}} \sum\limits_{i=1}^N \sum\limits_{j = i +1}^N \frac{1}{1+\abs{x_i - x_j}^s} \id \bP \myp{x}
		\leq{} \frac{C}{\eta^d} \int_{\R^d} \frac{\rho \myp{x}}{R \myp{x}^d} \id x.
	\end{equation}

	Finally, recalling from \eqref{eq:Rmaximalbound} that $ R \myp{x} $ is bounded from below in terms of the maximal function of $ \rho $, we apply the Hölder and Hardy-Littlewood maximal inequalities to obtain for any power $ p > 0 $,	\begin{align*}
		\int_{\R^d} \frac{\rho \myp{x}}{R \myp{x}^p} \id x
		\leq{}& C \int_{\R^d} \rho \myp{x} \myp{M_\rho} \myp{x}^{\frac{p}{d}} \id x \\
		\leq{}& C \myp[\Big]{ \int_{\R^d} \rho \myp{x}^{1+ \frac{p}{d}} \id x }^{\frac{d}{d+p}} \myp[\Big]{ \int_{\R^d} \myp{M_\rho} \myp{x}^{1+\frac{p}{d}} \id x}^{\frac{p}{d+p}} \\
		\leq{}& C \int_{\R^d} \rho \myp{x}^{1+\frac{p}{d}}.
	\end{align*}
	Using this on \eqref{eq:core_bound} and \eqref{eq:tail_bound}, and combining with \eqref{eq:interactionbound2}, we obtain the claimed bound~\eqref{eq:clenergybound}.
\end{proof}

\begin{proposition}[Special case $ \alpha = d $]
	\label{prop:clbound3}
	Let $ w $ be an interaction satisfying Assumption \ref{de:shortrangenew} with $ \alpha = d $, and $ 0 \leq \rho \in L^1 \myp{\R^d} $ a density with $ \int \rho = N $.
	Then, for any $ N $-particle probability measure $ \bP $ with one-body density $ \rho_{\bP} = \rho $ satisfying \eqref{eq:plansupport2} for some $ 0 < \eta \leq 1 $, the interaction energy is bounded by
	\begin{align}
		F_0 \myb{\rho}
		\leq{}& \int_{\R^{dN}} \sum\limits_{1 \leq i < j \leq N} w \myp{x_i-x_j} \id \bP \myp{x_1, \dotsc, x_N} \nn \\
		\leq{}& \frac{\kappa r_0^d C}{\eta^{2d}} \int_{\R^d} \rho^2 \myp[\Big]{ \log \myp[\Big]{\frac{c r_0^d}{\eta^{2d}} \rho} }_+ + \frac{\kappa C}{\eta^{2d}} \int_{\R^d} \frac{1}{1 + \abs{y}^s} \id y \int_{\R^d} \rho^2,
	\label{eq:clenergybound3}
	\end{align}
	where the constants $ c $ and $ C $ depend only on the dimension $ d $.
\end{proposition}

\begin{proof}
	The proof goes along the same lines as the proof of \cref{prop:clbound}.
	However, complications arise due to the fact that $ 1/\abs{x}^d $ is not integrable at infinity, so we need to take into account the finite range $ r_0 $ of the core of $ w $.
	Incidentally, this also forces us to avoid using the Hardy-Littlewood maximal inequality later in the proof.
	Following the proof of \cref{prop:clbound} up to \eqref{eq:interactbound} and noting that $ B \myp{z_j, \frac{\eta}{2} R \myp{x_j}} \subseteq B \myp{0, \abs{x_i-x_j}} $, we have
	\begin{align*}
		\sum\limits_{j=i+1}^N \frac{\mathds{1} \myp{\abs{x_i - x_j} \leq r_0}}{\abs{x_i - x_j}^d}
		\leq{}& \sum\limits_{j=i+1}^N \frac{\mathds{1} \myp{\abs{x_i - x_j} \leq r_0}}{\abs{B \myp{z_j,\frac{\eta}{2} R \myp{x_j}}}} \int_{B \myp{z_j, \frac{\eta}{2} R \myp{x_j}}} \frac{1}{\abs{y}^d} \id y \\
		\leq{}& \frac{1}{\abs{B \myp{0,\frac{\eta}{2} R \myp{x_i}}}} \int_{\eta R \myp{x_i} \leq \abs{y} \leq r_0} \frac{1}{\abs{y}^d} \id y \\
		={}& \frac{\abs{\mathbb{S}^{d-1}}}{\abs{B \myp{0,\frac{\eta}{2} R \myp{x_i}}}} \myp[\Big]{ \int_{\eta R \myp{x_i}}^{r_0} \frac{1}{r} \id r }_+ \\
		={}& \frac{2^d}{\eta^d R \myp{x_i}^d} \myp[\Big]{ \log \myp[\Big]{ \frac{r_0^d}{\eta^d R \myp{x_i}^d} } }_+.
	\end{align*}
	This leads to the bound
	\begin{align}
		\MoveEqLeft[6] \int_{\R^{dN}} \sum\limits_{i=1}^N \sum\limits_{j = i +1}^N w \myp{x_i - x_j} \id \bP \myp{x} \nn \\
		\leq{}& \kappa r_0^d \frac{2^d}{\eta^d} \int_{\R^{dN}} \sum\limits_{i=1}^N \frac{1}{R \myp{x_i}^d} \myp[\Big]{ \log \myp[\Big]{ \frac{r_0^d}{\eta^d R \myp{x_i}^d} } }_+ \id \bP \myp{x} \nn \\
		&+ \kappa \frac{2^d}{\abs{B_1} \eta^d} \int_{\R^d} \frac{1}{1 + \abs{y}^s} \id y \int_{\R^{dN}} \sum\limits_{i=1}^N \frac{1}{ R \myp{x_i}^d} \id \bP\myp{x},
	\end{align}
	where, in this case, we cannot use the Hardy-Littlewood maximal inequality on the first term.
	However, this can be circumvented using the fact that $ \abs{x_i - x_j} \geq \eta \myp{R \myp{x_i} + R \myp{x_j}} $ on the support of $ \bP $, which is the content of \cref{lem:Rconfigbound} below.
	Using the lemma, we conclude
	\begin{align*}
		F_0 \myb{\rho}
		\leq{}& \frac{\kappa r_0^d C}{\eta^{2d}} \int_{\R^d} \rho^2 \myp[\Big]{ \log \myp[\Big]{\frac{2^d \abs{B_1} r_0^d}{\eta^{2d}} \rho} }_+ + \frac{\kappa C}{\eta^{2d}} \int_{\R^d} \frac{1}{1 + \abs{y}^s} \id y \int_{\R^d} \rho^2,
	\end{align*}
	where the constant $ C $ depends only on the dimension $ d $.
\end{proof}

\begin{lemma}
	\label{lem:Rconfigbound}
	Let $ 0 \leq \rho \in L^1 \myp{\R^d} $ be any density with $ \int \rho > 1 $, and take any configuration of points $ x_1, \dotsc, x_M \in \R^d $ satisfying $ \abs{x_i - x_j} \geq \eta \myp{R \myp{x_i} + R \myp{x_j}} $ for $ i \neq j $, for some $ 0 < \eta \leq 1 $.
	Then we have the bounds
	\begin{equation}
	\label{eq:Rconfigbound1}
		\sum\limits_{i=1}^M \frac{1}{R \myp{x_i}^p}
		\leq \frac{C_{d,p}}{\eta^p} \int_{\R^d} \rho^{1+\frac{p}{d}}
	\end{equation}
	for any $ p > 0 $, and for any $ \lambda > 0 $,
	\begin{equation}
	\label{eq:Rconfigbound2}
		\sum\limits_{i=1}^M \frac{1}{R \myp{x_i}^d} \myp[\Big]{ \log \myp[\Big]{\frac{\lambda}{R \myp{x_i}^d}} }_+
		\leq \frac{C_d}{\eta^d} \int_{\R^d} \rho^2 \myp[\Big]{ \log \myp[\Big]{ \frac{2^d \lambda}{\eta^d} \abs{B_1} \rho} }_+.
	\end{equation}
\end{lemma}

\begin{proof}
	We consider any configuration $ x_1, \dotsc, x_M $ as in the statement, and seek to provide a bound on the sum $ \sum_{i=1}^M \frac{1}{R \myp{x_i}^p} $.
	We order the points such that $ R \myp{x_1} \leq \cdots \leq R \myp{x_M} $, and assume first for simplicity that all the balls $ B \myp{x_j, R \myp{x_j}} $ intersect the smallest ball $ B \myp{x_1, R \myp{x_1}} $.
	The main idea of the following argument is to split the space $ \mathbb{R}^d $ into shells of exponentially increasing width, centered around $ x_1 $, and arguing that the number of points among $ x_2, \dotsc, x_M $ that can lie in each shell is universally bounded.
	To elaborate, take any $ \tau > 1 $ and consider for $ m \in \N_0 $ the spherical shell of points $ y \in \mathbb{R}^d $ satisfying
	\begin{equation}
	\label{eq:shell}
		\tau^m \eta R \myp{x_1}
		\leq \abs{x_1 - y}
		< \tau^{m+1} \eta R \myp{x_1}.
	\end{equation}
	Note that if $ x_j $ lies in this shell, then by Lipschitz continuity of $ R $,
	\begin{equation*}
		\frac{2\eta}{1+\eta} R \myp{x_j}
		\leq \abs{x_1-x_j}
		< \tau^{m+1} \eta R \myp{x_1},
	\end{equation*}
	immediately implying that
	\begin{equation}
	\label{eq:shell_upper}
		R \myp{x_j} < \frac{1+\eta}{2} \tau^{m+1} R \myp{x_1}.
	\end{equation}
	This means that the ball $ B \myp{x_j, \eta R \myp{x_j}} $ is contained in
	\begin{equation*}
		B \myp{x_j, \eta R \myp{x_j}}
		\subseteq B \myp{x_1, \abs{x_1 - x_j} + \eta R \myp{x_j}}
		\subseteq B \myp[\big]{x_1, \frac{3+\eta}{2} \tau^{m+1} \eta R \myp{x_1}}.
	\end{equation*}
	Furthermore, by the assumption that $ B \myp{x_1, R \myp{x_1}} \cap B \myp{x_j, R \myp{x_j}} \neq \emptyset $, we have that
	\begin{equation}
	\label{eq:shell_lower}
		\frac{\tau^m}{2} \eta R \myp{x_1}
		\leq \frac{1}{2} \abs{x_1 - x_j}
		< \frac{1}{2} \myp{R \myp{x_1} + R \myp{x_j}}
		\leq R \myp{x_j}.
	\end{equation}
	Since the balls $ B \myp{x_j, \eta R \myp{x_j}} $ are all disjoint, we conclude that the number of $ x_j $'s that can lie in the $ m $'th shell around $ x_1 $ is bounded by the ratio of the volumes
	\begin{align*}
		\# \Set{j \mid \tau^m \eta R \myp{x_1} \leq \abs{x_1 - x_j} < \tau^{m+1} \eta R \myp{x_1}}
		\leq{}& \frac{ \abs{B \myp[\big]{x_1, \frac{3+\eta}{2} \tau^{m+1} \eta R \myp{x_1}}} }{ \abs{ B \myp{0, \frac{\tau^m}{2} \eta R \myp{x_1}}} } \\
		\leq{}& 4^d \tau^d.
	\end{align*}
	Note also that no $ x_j $ can be placed inside the first shell (corresponding to $ \abs{x_1- x_j} < \eta R \myp{x_1} $), because we always have $ \abs{x_1 - x_j} \geq \eta \myp{R \myp{x_1} + R \myp{x_j}} $ by assumption.
	Now, for any power $ p > 0 $, this allows us to bound, using \eqref{eq:R_bound},
	\begin{align}
		\sum\limits_{j = 1}^M \frac{1}{R \myp{x_j}^p}
		\leq{}& 4^d \tau^d \sum\limits_{m=0}^{\infty} \frac{2^p}{\myp{\tau^m \eta R \myp{x_1} }^p}
		\leq{} \frac{ 2^{p+2d} \tau^{p+d} }{ \eta^p \myp{\tau^p - 1} } \frac{1}{R \myp{x_1}^p} \nn \\
		\leq{}& \frac{ 2^{p+2d} \tau^{p+d} }{ \eta^p \myp{\tau^p - 1} } \abs{B_1}^{\frac{p}{d}} \int_{B\myp{x_1, R \myp{x_1}}} \rho \myp{y}^{1+\frac{p}{d}} \id y.
	\label{eq:config_bound1}
	\end{align}
	To bound the sum involving the logarithm, we note first that for any $ \lambda > 0 $, applying Jensen's inequality to the function $ t \mapsto t^2 \myp{\log \lambda t}_+ $ yields
	\begin{align}
		\frac{1}{R \myp{x}^d} \myp[\Big]{\log \myp[\Big]{ \frac{\lambda}{R \myp{x}^d}}}_+
		={}& \frac{\abs{B_1}}{\abs{B \myp{x}}} \myp[\bigg]{\int_{B \myp{x}} \rho}^2 \myp[\Big]{\log \myp[\Big]{\frac{\lambda \abs{B_1}}{\abs{B \myp{x}}} \int_{B \myp{x}} \rho}}_+ \nn \\
		\leq{}& \abs{B_1} \int_{B \myp{x}} \rho^2 \myp{ \log \myp{\lambda \abs{B_1} \rho}}_+.
	\label{eq:Rlogbound}
	\end{align}
	Using this, we obtain by again summing over all the shells,
	\begin{align*}
		\sum\limits_{i=1}^M \frac{1}{R \myp{x_i}^d} \myp[\Big]{ \log \myp[\Big]{\frac{\lambda}{R \myp{x_i}^d}} }_+
		\leq{}& \sum\limits_{m=0}^{\infty} \frac{4^d \tau^d 2^d}{\myp{\tau^m \eta R \myp{x_1} }^d} \myp[\Big]{ \log \myp[\Big]{\frac{2^d \lambda}{\myp{\tau^m \eta R \myp{x_1} }^d}} }_+ \\
		\leq{}& \frac{2^{3d} \tau^d}{\eta^d} \sum\limits_{m=0}^{\infty} \frac{1}{\tau^{dm} R \myp{x_1}^d} \myp[\Big]{ \log \myp[\Big]{\frac{2^d \lambda}{\eta^d R \myp{x_1}^d}} }_+ \\
		\leq{}& \frac{2^{3d} \tau^{2d} \abs{B_1}}{\eta^d \myp{\tau^d -1}} \int_{B \myp{x_1, R \myp{x_1}}} \rho^2 \myp[\Big]{ \log \myp[\Big]{ \frac{2^d \lambda}{\eta^d} \abs{B_1} \rho} }_+.
	\end{align*}

	Finally, we generalize to the case where not all the balls $ B \myp{x_j, R \myp{x_j}} $ intersect the smallest ball $ B \myp{x_1, R \myp{x_1}} $.
	We split the configuration $ \myp{x_j}_{1 \leq j \leq M} $ into clusters $ \myp[\big]{x_j^{\myp{k}}}_{1 \leq j \leq n_k} $ with $ 1 \leq k \leq K $, such that:
	\begin{itemize}
		\item For any $ k $, $ R \myp{x_1^{\myp{k}}} \leq \cdots \leq R \myp{x_{n_k}^{\myp{k}}} $.

		\item $ B \myp{x_1^{\myp{k}}, R \myp{x_1^{\myp{k}}}} \cap B \myp{x_j^{\myp{k}}, R \myp{x_j^{\myp{k}}}} \neq \emptyset $ for any $ j,k $.

		\item The balls $ B \myp{x_1^{\myp{k}}, R \myp{x_1^{\myp{k}}}} $ are all pairwise disjoint for $ k = 1, \dotsc, K $.
	\end{itemize}
	Then, using \eqref{eq:config_bound1} on each cluster, we get for instance
	\begin{align*}
		\sum\limits_{j = 1}^M \frac{1}{R \myp{x_j}^p}
		\leq{}& \frac{ 2^{p+2d} \tau^{p+d} }{ \eta^p \myp{\tau^p - 1} } \abs{B_1}^{\frac{p}{d}} \sum\limits_{k=1}^K \int_{B \myp{x_1^{\myp{k}}, R \myp{x_1^{\myp{k}}}}} \rho \myp{y}^{1+\frac{p}{d}} \id y \\
		\leq{}& \frac{ 2^{p+2d} \tau^{p+d} }{ \eta^p \myp{\tau^p - 1} } \abs{B_1}^{\frac{p}{d}} \int_{\mathbb{R}^d} \rho \myp{y}^{1+\frac{p}{d}} \id y,
	\end{align*}
	which concludes the proof of \eqref{eq:Rconfigbound1}.
	\eqref{eq:Rconfigbound2} follows in the same way.
\end{proof}


\begin{remark}[Hard-core at zero temperature]
	\label{rmk:clbound2}
As we have mentioned in Section~\ref{sec:hard-core}, in the hard core case $\alpha=+\ii$, we know from~\eqref{eq:simple_hard_core_T0} that for any representable density $\rho$, we have
	\begin{equation}
	\label{eq:clenergybound2}
		F_0 [\rho]
		\leq \frac{\kappa C}{r_0^s} \int_{\R^d} \rho \myp{x} \id x,
	\end{equation}
where the constant $ C $ depends only on $ d $ and $ s $. The problem is to determine when $\rho$ is representable. Using Theorem~\ref{thm:clstate}, this is the case when for instance $R_\rho=\min_{x} R(x)\geq r_0$.
\end{remark}

\subsection{The block approximation}\label{sec:Block}
While the state from \cref{thm:clstate} is useful for obtaining energy bounds at zero temperature, it might be singular with respect to the Lebesgue measure on $ \R^{dN} $, leaving it unsuitable to use for the positive temperature case, because the entropy in this case will be infinite.
Here we describe a simple way of regularizing states, while keeping the one-body density fixed, which is a slight generalization to any partition of unity of the construction in~\cite{CarDuvPeySch-17}.
Essentially, it works by cutting $ \R^d $ into "blocks" and then locally replacing the state by a pure tensor product.

Let $ \sum \chi_j = \mathds{1}_{\R^d} $ be any partition of unity, and $ \bP $ any $ N $-particle state with density $ \rho $.
The corresponding \emph{block approximation} is defined by
\begin{equation}
	\widetilde\bP
	:= \sum\limits_{i_1,...,i_N} \bP \myp{\chi_{i_1} \otimes\cdots\otimes \chi_{i_N}} \frac{\myp{ \rho \chi_{i_1}} \otimes\cdots \otimes \myp{ \rho \chi_{i_N}} }{\prod_{k=1}^N\int_{\R^d} \rho \chi_{i_k}},
\label{eq:block}
\end{equation}
where we denote
\begin{equation*}
	\bP \myp{\chi_{i_1} \otimes\cdots\otimes \chi_{i_N}}
	:= \int_{\mathbb{R}^{dN}} \chi_{i_1} \otimes\cdots\otimes \chi_{i_N} \id \bP.
\end{equation*}
That is, $ \widetilde{\bP} $ is a convex combination of tensor products of the normalized $ \frac{\rho \chi_i}{\int \rho \chi_i} $.
One can easily show that $ \widetilde{\bP} $ has one-body density $ \rho_{\widetilde{\bP}} = \rho $.
Furthermore, it is clear that $ \widetilde{\bP} $ is a symmetric measure whenever $ \bP $ is, so we can also write
\begin{equation*}
	\widetilde{\bP}
	= \sum\limits_{i_1,...,i_N} \bP \myp{\chi_{i_1} \otimes\cdots\otimes \chi_{i_N}} \Pi_s \myp[\Big]{ \frac{\rho \chi_{i_1}}{\int \rho \chi_{i_1}} \otimes\cdots\otimes \frac{\rho \chi_{i_N}}{\int \rho \chi_{i_N}} },
\end{equation*}
where $ \Pi_s $ denotes the symmetrization operator  in~\eqref{eq:def_Pi_s}. In~\cite{CarDuvPeySch-17} the chosen partition of unity is just a tiling made of cubes, but in fact any partition works. Applying Jensen's inequality yields the following.

\begin{lemma}[Entropy of the block approximation]
	\label{lem:block_entropy}
	Suppose that the state $ \bP $ and the partition of unity $ \myp{\chi_j} $ are such that $ \chi_{i_1}, \dotsc, \chi_{i_N} $ all have disjoint supports whenever $ \bP \myp{\chi_{i_1} \otimes \cdots \otimes \chi_{i_N}} \neq 0 $.
	Then we have
	\begin{multline}
	\label{eq:block_entropy}
		\int_{\R^{dN}} \widetilde{\bP} \log \myp{ N! \, \widetilde\bP}
		\leq{} \int_{\R^d} \rho \log \rho + \int_{\R^d} \rho \sum\limits_i \chi_i \log \chi_i\\ -\sum\limits_{i} \myp[\Big]{\int_{\R^d} \rho \chi_{i}} \log \myp[\Big]{\int_{\R^d} \rho \chi_{i}}.
	\end{multline}
\end{lemma}

\begin{remark}\label{rmk:block_compact}
	Since $ \myp{\chi_i} $ is a partition of unity, the term above involving $ \chi_i \log \chi_i $ can always be estimated from above by zero.
	On the other hand, it is not clear that the sum in last term above is even finite for an arbitrary partition $ \myp{\chi_i} $.
	However, it turns out to behave nicely in many situations. For instance, if $\int\rho\chi_j\leq1$ for all $j$, we can estimate it by $ 1/e $ times the number of terms in the partition of unity, which is typically finite when $\rho$ has compact support.
\end{remark}

\begin{proof}
	The entropy of the block approximation can be estimated using Jensen's inequality by
	\begin{align*}
		& \int \widetilde{\bP} \log \myp{ N! \, \widetilde\bP} \\
		&\ \leq \sum\limits_{i_1,...,i_N} \bP \myp{\chi_{i_1}\otimes\cdots\otimes \chi_{i_N}} \int \Pi_s \left( \bigotimes_k \frac{\rho \chi_{i_k}}{\int \rho \chi_{i_k}} \right) \log \left( N! \, \Pi_s  \bigotimes_k \frac{\rho \chi_{i_k}}{\int \rho \chi_{i_k}} \right)\\
		&\  = \sum\limits_{i_1,...,i_N} \bP \myp{\chi_{i_1}\otimes\cdots\otimes \chi_{i_N}} \int \bigotimes_k \frac{\rho \chi_{i_k}}{\int \rho \chi_{i_k}}  \log \left(\sum_{\sigma \in \gS_N}\bigotimes_k \frac{\rho \chi_{i_{\sigma(k)}}}{\int \rho \chi_{i_{\sigma(k)}}} \right).
	\end{align*}
	We have here used the symmetry of $\bP$ to remove the first $\Pi_s$. It is important that the $N!$ has disappeared in the logarithm. For any non-zero term, the supports of the $ \chi_{i_k} $ are all disjoint, hence only the case $\sigma=\text{Id}$ remains in the sum. Using that
	\begin{equation*}
		\int \bigotimes_k \frac{\rho \chi_{i_k}}{\int \rho \chi_{i_k}}  \log \left(\bigotimes_k \frac{\rho \chi_{i_k}}{\int \rho \chi_{i_k}} \right)
		= \sum\limits_{k=1}^N \int \frac{\rho \chi_{i_k}}{\int \rho \chi_{i_k}} \log \frac{ \rho \chi_{i_k}}{\int \rho \chi_{i_k}}
	\end{equation*}
and plugging this into the previous expression, we conclude that \eqref{eq:block_entropy} holds.
\end{proof}

\subsection{Proof of Theorem~\ref{thm:OT_block}  in the canonical case at \texorpdfstring{$T>0$}{T>0}}\label{sec:proof_OT_block}
	We assume first that the density $ \rho $ is compactly supported, and then remove this assumption at the end.
	Applying the Besicovitch covering lemma~\cite{FraLapWei-LT,Guzman} on the cover $ \Set{B \myp{x, \eps R \myp{x}} \given x \in \supp \rho} $ gives the existence of a (finite) set of points $ \myp{y_j} \subseteq \supp \rho $ satisfying that $ \myp{B_j} := \myp{B \myp{y_j, \eps R \myp{y_j}}} $ covers the support of $ \rho $, and the multiplicity of the cover is universally bounded, i.e.,
	\begin{equation*}
		1 \leq \varphi \myp{x}
		:= \sum\limits_j \1_{B_j} \myp{x}
		\leq C_d,
		\qquad x \in \supp \rho,
	\end{equation*}
	where the constant $ C_d $ depends only on the dimension $ d $, and thus not on $ \eps $ or $ \rho $.
	This gives us a partition of unity $ \myp{\chi_j} $ defined by $ \chi_j := \frac{\1_{B_j}}{\varphi} $.
	One way of constructing the Besicovitch cover is to inductively maximize $ \eps R \myp{y_j} $ over the remaining volume $ y_j \in \supp \rho \setminus \bigcup_{k=1}^{j-1} B_k $, supposing that $ y_1, \dotsc, y_{j-1} $ have already been chosen.
	This construction implies the bound on the distances
	\begin{equation}
	\label{eq:Bes_distance}
		\abs{y_j - y_k}
		\geq \max\myp{\eps R \myp{y_j}, \eps R \myp{y_k}}
		\geq \frac{\eps}{2} \myp{R \myp{y_j} + R \myp{y_k}}
	\end{equation}
	for all $ j \neq k $.

	We now take the optimal transport state $ \bP $ obtained from \cref{thm:clstate}, and denote by $ m_j := \int \rho \chi_j = \int_{B_j} \frac{\rho}{\varphi} $ the local mass of $ \rho $ with respect to the partition of unity $ \myp{\chi_j} $.
	As a trial state for the free energy, we take the block approximation \eqref{eq:block} of $ \bP $ using the $ \chi_j $, i.e.,
	\begin{equation*}
		\bP_{\varepsilon}
		:={} \sum\limits_{j_1, \dotsc, j_N} \mathbb{P} \myp{\chi_{j_1} \times \cdots \times \chi_{j_N}} \myp[\Big]{\frac{\rho \chi_{j_1}}{m_{j_1}} } \otimes \cdots \otimes \myp[\Big]{\frac{\rho \chi_{j_N}}{m_{j_N}} }.
	\end{equation*}
	We show that the support of $ \bP_{\varepsilon} $ satisfies the condition \eqref{eq:plansupport2} for some $ \eta $.
	For any point $ \myp{x_1, \dotsc, x_N} \in \supp \bP_{\varepsilon} $, there must be a term in the sum above such that $ \bP \myp{\chi_{j_1} \times \cdots \times \chi_{j_N}} \neq 0 $, and $ x_k \in B_{j_k} = B\myp{y_{j_k}, \varepsilon R \myp{y_{j_k}}} $ for all $ k $.
	In particular, since the support of $ \bP $ satisfies \eqref{eq:plansupport}, there exist $ z_1, \dotsc, z_N $ with $ z_k \in B_{j_k} $ and $ \abs{z_k - z_{\ell}} \geq \frac{1}{3} \myp{R \myp{z_k} + R \myp{z_{\ell}}} $ for any $ k \neq \ell $.
	By the Lipschitz continuity of $ R $, $ x_k \in B_{j_k} $ implies that $ R \myp{y_{j_k}} \leq \frac{1}{1-\varepsilon} R \myp{x_k} $, so
	\begin{equation*}
		\abs{x_k - z_k}
		\leq 2 \varepsilon R \myp{y_{j_k}}
		\leq \frac{2\varepsilon}{1-\varepsilon} R \myp{x_k}.
	\end{equation*}
	Finally, this gives us the bound
	\begin{align}
		\abs{x_k - x_{\ell}}
		\geq{}& \abs{z_k - z_{\ell}} - \abs{x_k - z_k} - \abs{x_{\ell} - z_{\ell}} \nn \\
		\geq{}& \frac{1}{3} \myp{R \myp{z_k} + R \myp{z_{\ell}}} - \abs{x_k - z_k} - \abs{x_{\ell} - z_{\ell}} \nn \\
		\geq{}& \frac{1}{3} \myp{R \myp{x_k} + R \myp{x_{\ell}}} - \frac{4}{3} \myp{\abs{x_k - z_k} + \abs{x_{\ell} - z_{\ell}}} \nn \\
		\geq{}& \frac{1}{3} \myp[\Big]{1 - \frac{8 \varepsilon}{1-\varepsilon}} \myp{R \myp{x_k} + R \myp{x_{\ell}}}.
	\label{eq:OT_block_supp}
	\end{align}
	This argument also shows that if $ \bP \myp{\chi_{j_1} \times \cdots \times \chi_{j_N}} \neq 0 $, then the sets $ B_{j_k} $ are disjoint for $ k = 1, \dotsc, N $, provided that $ \varepsilon < \frac{1}{9} $.

	Now, since $ \bP_{\eps} $ satisfies \eqref{eq:OT_block_supp}, it follows from \cref{prop:clbound} that the interaction energy (in case $ \alpha > d $) is bounded by
	\begin{equation*}
		\cU_{N} \myp{\bP_{\eps}}
		\leq{} C\kappa r_0^\alpha \int_{\R^d} \rho^{1+\frac{\alpha}{d}} + C\kappa \int_{\R^d} \rho^2,
	\end{equation*}
	and similarly for $ \alpha = d $, using \cref{prop:clbound3}.
	Thus, to show \eqref{eq:can_bound2}, it only remains to provide a bound on the entropy of the state $ \bP_{\varepsilon} $.
	First, applying \cref{lem:block_entropy} immediately gives
	\begin{align*}
		\int_{\R^{dN}} \bP_{\varepsilon} \log \myp{N! \, \bP_{\varepsilon}}
		\leq{} \int_{\R^d} \rho \log \rho - \sum\limits_{j} m_j \log m_j.
	\end{align*}
	Then, for any numbers $ s,t \geq 0 $, we can use the elementary bound
	\begin{equation*}
		- s \log \myp{t s} \leq \frac{1}{et}
	\end{equation*}
	to conclude that
	\begin{align*}
		- \sum\limits_{j} m_j \log m_j
		={}& \sum\limits_{j} m_j \log \myp{R\myp{y_j}^d} - m_j \log \myp{R \myp{y_j}^d m_j} \\
		\leq{}& \sum\limits_{j} d \int \rho \myp{x} \chi_{j} \myp{x} \log \myp{ \myp{1+\eps} R \myp{x} } \id x + \frac{1}{e R \myp{y_j}^d} \\
		\leq{}& d \log \myp{1+\eps} \int_{\R^d} \rho + d \int_{\R^d} \rho \log R + \frac{C}{\eps^d} \int_{\R^d} \rho^2,
	\end{align*}
	where the last inequality uses \eqref{eq:Bes_distance} and \cref{lem:Rconfigbound}.
	This proves \cref{thm:OT_block} for compactly supported densities.
\qed

\begin{remark}[Hard-core case]\label{rmk:Hard-core-CT}
In the hard core case $\alpha=+\ii$, the above proof provides the bound
		\begin{multline}
		F_T [\rho]\leq  C\frac{\kappa}{r_0^d}\int_{\R^d}\rho+CT \int_{\R^d} \rho + T \int_{\R^d} \rho \log \rho+\frac{CTr_0^d}{(R_\rho-r_0)^d}\int_{\R^d}\rho^2\\+T \int_{\R^d} \rho \log R^d\label{eq:estim_hard_core_T>0}
			\end{multline}
under the assumption that $R_\rho=\min_x R(x)>r_0$, where $C$ only depends on $d$ and $s$. The main difference is the estimate on the distance between the particles in \eqref{eq:OT_block_supp}. We need to keep the maximum and use
		\begin{align*}
		\abs{x_k - x_{\ell}}
		\geq{}&  \max\left\{R_\rho,\frac{1}{3} \myp{R \myp{x_k} + R \myp{x_{\ell}}}\right\}  - \frac{8 \varepsilon}{3(1-\varepsilon)} \myp{R \myp{x_k} + R \myp{x_{\ell}}}\\
		\geq{}& \left(1-\frac{8\eps}{1-\eps}\right)R_\rho.
			\end{align*}
			Taking $\eps=\min(R_\rho/r_0-1,1)/100$ provides~\eqref{eq:estim_hard_core_T>0}.
\end{remark}

\subsection{Removal of the compactness condition} To finish this section we describe how to extend a result holding for compactly supported densities to general integrable ones, using this time a compactness argument.

\begin{theorem}
\label{thm:non_compact}
	Assume that $w$ satisfies Assumption~\ref{de:shortrangenew}.
	If we have for some $1\leq p\leq q<\ii$ with $q\geq2$ and some constants $C_j\geq0$
	\begin{multline}
	\label{eq:estimate}
		F_T[\rho]\leq C_0\int\rho+C_1\int\rho^{p}+C_2\int\rho^q +T\int\rho\log\rho\\
		+C_3\int\rho^2 (\log\rho)_++C_4\int\rho \log R
	\end{multline}
	for all $\rho\in L^1\cap L^q$ of compact support, then the same holds with the same constants for all $\rho\in L^1\cap L^q$. If $T>0$ we assume in both cases that $\int_{\R^d}\rho|\log\rho|<\ii$.
\end{theorem}

\begin{proof}
Let us first assume $C_4=0$ for simplicity. Our proof uses that the energy $\rho\mapsto F_T[\rho]$ is lower semi-continuous for the strong topology of $L^1$, 	as previously mentioned in Remark~\ref{rmk:lsc}, that is,
	\begin{multline}
	F_T[\rho]\leq\liminf_{n\to\ii}F_T[\rho_n]\\\text{if $\rho_n\to\rho$ strongly in $L^1(\R^d)$ with $\int\rho_n^{q}+T\int\rho_n|\log\rho_n|\leq C$.}
	\label{eq:lsc}
	\end{multline}
The theorem then follows immediately by letting
	$$\rho_n:=\frac{N}{\int_{B_n}\rho}\;\rho\1_{B_n}$$
	the truncation of $\rho$ over the ball of radius $n$. Note that $\rho_n\leq (1+o(1))\rho$. The sequence $\rho_n$ clearly satisfies the convergence properties of~\eqref{eq:lsc} and therefore the lower semi-continuity provides
	\begin{align*}
		F_T[\rho]
		\leq{}& \! \liminf_{n\to\ii} F_T[\rho_n] \\
		\leq{}&\liminf_{n\to\ii} \Bigg\{C_0N+ C_1 \myp[\bigg]{ \frac{N}{\int_{B_n}\rho} }^p\int_{B_n}\rho^p+C_2 \myp[\bigg]{ \frac{N}{\int_{B_n}\rho}}^q\int_{B_n}\rho^q \\
&\quad+T\frac{N}{\int_{B_n}\rho}\int_{B_n}\rho\log\rho+T\frac{N}{\int_{B_n}\rho} \log \myp[\bigg]{ \frac{N}{\int_{B_n}\rho}} \int_{B_n}\rho\\
&\quad +C_3 \myp[\bigg]{ \frac{N}{\int_{B_n}\rho} }^2\int_{B'_n}\rho^2\log\rho+2 \myp[\bigg]{\frac{N}{\int_{B_n}\rho}}^2 \log \myp[\bigg]{\frac{N}{\int_{B_n}\rho} } \int_{B'_n}\rho^2\Bigg\}\\
		={}& C_0N+C_1\int\rho^{p}+C_2\int\rho^q +T\int\rho\log\rho+C_3\int\rho^2(\log\rho)_+,
	\end{align*}
where $B'_n:=B_n\cap\big\{\rho\geq N^{-1}\int_{B_n}\rho\big\}$.

When $C_4>0$ the proof is similar. We need to use that $(1+|x|)/C\leq R(x),R_n(x)\leq C(1+|x|)$ for some $C>0$ (depending on $\rho$), where
$R_n(x)$ is the local radius of the truncated density $\rho_n$, which converges locally to $R$. The uniform bounds on $R$ and $R_n$ imply that we must work under the assumptions that $\int\rho(\log|x|)_+$ is finite (otherwise there is nothing to show). The limit follows from dominated convergence.

For the convenience of the reader, we conclude by quickly recalling the proof of the lower semi-continuity~\eqref{eq:lsc}. We consider an arbitrary sequence $\rho_n$ converging to $\rho$ strongly in $L^1$ and satisfying the bounds in~\eqref{eq:lsc}. It is known that there exists an optimal $\bP_n$ for $F_T[\rho_n]$ (but we could as well use a quasi-minimizer). From the upper bound we have $F_T[\rho_n]\leq C$ for some constant $C$ and therefore
	\begin{align*}
		C\geq{}& F_T[\rho_n]
		={} \mathcal{F}_T(\bP_n)\\
		={}& \int_{(\R^d)^N} \sum_{1\leq j<k\leq N}w(x_j-x_k) \, \bP_n+T\int\bP_n\log(N! \, \bP_n)\\
		={}& \int_{(\R^d)^N} \myp[\bigg]{ \sum_{1\leq j<k\leq N}w(x_j-x_k)+\kappa N} \bP_n+T\int\bP_n\log\left(\frac{\bP_n}{(\rho_n/N)^{\otimes N}}\right)\\
		&-\kappa N+T\int\rho_n\log\rho_n+T\log\frac{N!}{N^N}.
	\end{align*}
	The first term is non-negative from the stability property of $w$ and the second is a relative entropy, hence is also non-negative. We have thus proved that
	$$T\int\bP_n\log\bP_n\leq C(\rho,N,T)$$
	where the constant can depend on $\rho,N,T$ but not on $n$.
	On the other hand, we know that the sequence $(\bP_n)$ is tight, that is,
	$$\int_{\max|x_j|\geq R}\,\rd\bP_n\leq \int \sum_{j=1}^N\1(|x_j|\geq R)\,\rd\bP_n=\int_{|x|\geq R}\rho_n$$
	where the right side is small due to the strong convergence in $L^1$. After extraction of a subsequence, this implies $\int F\,\rd\bP_n\to\int F\,\rd\bP$ for every $F\in C^0_b$. Taking $F(x_1,...,x_N)=\sum_{j=1}^Nf(x_j)$ with $f\in C^0_b$, we find that $\int f\rho_{\bP_n}\to \int f\rho_{\bP}$, that is, $\rho_{\bP}=\rho$. In addition, we have (by convexity)
	\begin{equation}
		T\int\bP\log\bP\leq T\liminf_{n\to\ii}\int\bP_n\log\bP_n\leq C.
	\label{eq:entropy_wlsc}
	\end{equation}
	Hence $\bP$ is admissible for $F_T[\rho]$, and absolutely continuous with respect to the Lebesgue measure if $T>0$. We thus have
	$$\liminf_{n\to\ii}\int F\,\rd\bP_n\geq\int F\,\rd\bP$$
	for every measurable function $F\geq0$ if $T>0$ (using the absolute continuity of $\bP$) and for every lower semi-continuous function $F\geq0$ if $T=0$. This is satisfied for our interaction $w$ by Assumption~\ref{de:shortrangenew} and therefore we obtain as we wanted
	\begin{align*}
		\MoveEqLeft[8] \liminf_{n\to\ii}\int_{(\R^d)^N} \myp[\bigg]{ \sum_{1\leq j<k\leq N}w(x_j-x_k)+\kappa N }\rd\bP_n \\
		\geq{}& \int_{(\R^d)^N} \myp[\bigg]{ \sum_{1\leq j<k\leq N}w(x_j-x_k)+\kappa N } \rd\bP.
	\end{align*}
	Together with the entropy bound~\eqref{eq:entropy_wlsc} when $T>0$, this proves that
	$$\liminf_{n\to\ii}\left(F_T[\rho_n]+\kappa N\right)\geq F_T[\rho]+\kappa N$$
	which is the claimed lower semi-continuity~\eqref{eq:lsc}.
	\end{proof}

\section{Proof of Theorem~\ref{thm:hard-core2} in the hard-core case}\label{sec:proof_hard-core}

In this section we prove Theorem \ref{thm:hard-core2} concerning densities which are uniformly bounded in terms of the tight packing density $\rho_c(d)$. We start by constructing a trial state with constant density by averaging a periodic tight packing over translations. Such a uniform average of a periodic lattice is often called a ``floating crystal''~\cite{LewLieSei-19b,LewLieSei-19_ppt} in Physics and Chemistry. Finally, we estimate the entropic cost of ``geometrically localizing''~\cite{Lewin-11} this state to enforce the desired density.

\medskip

\noindent \textbf{Step 1. Constant density.}
We have assumed $\rho\leq (1-\eps)^dr_0^{-d}\rho_c(d)$. Let $\eta>0$ be a fixed small number which will later be chosen in terms of $\eps$. From the definition of $\rho_c(d)$ we can find a large cube $C_\ell=(-\ell/2,\ell/2)^d$ and $n=(1+2\eta)^{-d}r_0^{-d}\rho_c(d)\ell^d\in\N$ points $x_1^0,...,x_n^0\in C_\ell$ satisfying $|x^0_j-x^0_k|\geq r_0(1+\eta)$ for all $j\neq k$. We can also assume that no point is at a distance less than $r_0$ to the boundary of $C_\ell$. We are using here that the tight packing density for $r_0(1+\eta)$ is $(1+\eta)^{-d}r_0^{-d}\rho_c(d)>(1+2\eta)^{-d}r_0^{-d}\rho_c(d)$ and that the limit~\eqref{eq:sphere_packing} is the same for cubes and for balls.

Now, we replace each point $x_j^0$ by a smeared measure
$$\chi_j^0(x)=\frac{2^d}{(r_0\eta)^d}\chi \myp[\bigg]{ 2\frac{x-x_j^0}{r_0\eta} }$$
where $\chi=|B_1|^{-1}\1_{B_1}$. The smearing radius $\eta r_0/2$ has been chosen so that the supports of the $\chi^0_j$ remain at distance at least $r_0$.

Finally, we consider $(2K+1)^d$ copies of our system ($K\in\N$), repeated in a periodic fashion so as to form a very large cube $C_L=(-L/2,L/2)^d$ of side length $L=(2K+1)\ell$. In other words, we define the $N:=(2K+1)^dn$ points $x_j^k:=x_j^0+kL$ with $k\in\{-K,...,K\}^d$. The smeared measures $\chi_j^k$ are defined similarly. The state
$$\bP=\Pi_s\bigotimes_{\substack{j\in\{1,...,n\}\\ k\in\{-K,...K\}^d}} \chi_j^k$$
has the density $\rho=\sum_{j,k}\chi_j^k$ and the finite entropy
$$\int_{(\R^d)^N}\bP\log(N!\,\bP)=N\int_{\R^d}\chi\log\chi=N\log\left(\frac{2^d}{|B_1|r_0^d\eta^d}\right)$$
(recall $\Pi_s$ is the symmetrization operator  in~\eqref{eq:def_Pi_s}). Finally, we average over translations of the big cube and define the trial state
$$\tilde \bP=\frac1{\ell^d}\int_{C_\ell}\bP(\cdot+\tau)\,\rd\tau,$$
which has the density
$$\tilde\rho=\frac1{\ell^d}\sum_{j}\chi_j^0\ast\1_{C_L}.$$
The latter is constant, equal to $n/\ell^d=(1+2\eta)^{-d}r_0^{-d}\rho_c(d)$ well inside the large cube. Note that, by concavity, the entropy of $\tilde\bP$ can be estimated by that of $\bP$.

\medskip

\noindent \textbf{Step 2. Geometric localization.}
We assume for the rest of the proof that $\rho$ has a compact support and we choose $K$ large enough so that $\tilde\rho$ is constant on the support of $\rho$. Our estimates will not depend on $K$. One can then deduce the bound for general densities by adapting the proof of Theorem~\ref{thm:non_compact}, or by passing to the limit $K\to\ii$ in the formulas~\eqref{eq:trial_state_hard_core}--\eqref{eq:geom_localization_tensor} of the trial state.

We pick $\eta$ so that $(1-\eps)^d=(1+2\eta)^{-d}$, that is,
$$\eta=\frac\eps{2(1-\eps)}.$$
Then we have $\rho\leq \tilde \rho$ a.e. This enables us to consider the localization function
$$\theta:=\frac{\rho}{\tilde\rho}=\frac{\rho}{(1+2\eta)^{-d}r_0^{-d}\rho_c(d)}\leq1$$
and the $\theta$--localized state $\tilde\bP_{|\theta}$, which has the desired density $\theta\rho_{\tilde\bP}=\rho$.

We recall that the $\theta$--localization $\bQ_{|\theta}$ of a state $\bQ$ (with $0\leq\theta\leq1$) is the unique state which has the correlation functions $\rho^{(k)}=\rho^{(k)}_{\bQ}\theta^{\otimes k}$ for all $k$, see~\cite{HaiLewSol_2-09,Lewin-11,FouLewSol-18}. In our case we only need the definition for a tensor product since we  have by linearity
\begin{equation}
\tilde\bP_{|\theta}=\frac1{\ell^d}\int_{C_\ell}\bP(\cdot+\tau)_{|\theta}\,\rd\tau.
\label{eq:trial_state_hard_core}
\end{equation}
For a symmetric tensor product $\bQ=\Pi_s(q_1\otimes\cdots\otimes q_N)$ with probabilities $q_j$ of disjoint support, the $\theta$-localized state can be expressed as
\begin{multline}
\bQ_{|\theta}=\bigoplus_{n=0}^N\binom{N}{n} \frac{1}{N!}\sum_{\sigma\in\gS_N} (\theta q_{\sigma(1)})\otimes\cdots\otimes (\theta q_{\sigma(n)})\times\\
\times \left(1-\int\theta q_{\sigma(n+1)}\right)\cdots \left(1-\int\theta q_{\sigma(N)}\right).
\label{eq:geom_localization_tensor}
\end{multline}
We will need the following.

\begin{lemma}[Entropy of localization of tensor products]
Let $\bQ=\Pi_s(q_1\otimes\cdots\otimes q_N)$ be a symmetric tensor product, with $q_1,...,q_N$ probability measures of disjoint supports. For any $0\leq \theta\leq1$, we have
\begin{equation}
\cS(\bQ_{|\theta})=-\sum_j \int_{\R^d} (\theta q_{j})\log(\theta q_{j})-\sum_j \left(1-\int_{\R^d}\theta q_{j}\right)\log\left(1-\int_{\R^d}\theta q_{j}\right).
\label{eq:entropy_loc}
\end{equation}
In particular, we deduce
\begin{equation}
-\cS(\bQ_{|\theta})\leq\sum_j \int_{\R^d} (\theta q_{j})\log(\theta q_{j}).
\label{eq:entropy_loc_bd}
\end{equation}
\end{lemma}

As a side remark we also note also that~\eqref{eq:entropy_loc} provides
\begin{multline*}
\cS(\bQ_{|\theta})+\cS(\bQ_{|1-\theta})\\
=\cS(\bQ)-\sum_j \left[\left(1-\int \theta q_{j}\right)\log\left(1-\int \theta q_{j}\right)+\left(\int \theta q_{j}\right)\log\left(\int \theta q_{j}\right)\right]\\
-\int\rho \Big(\theta\log\theta +(1-\theta)\log(1-\theta)\Big).
\end{multline*}
The additional terms are positive and therefore we recover the subadditivity of the entropy $\cS(\bQ)\leq \cS(\bQ_{|\theta})+\cS(\bQ_{|1-\theta})$~\cite[Appendix A]{HaiLewSol_2-09}.

\begin{proof}
Each tensor product $(\theta q_{\sigma(1)})\otimes\cdots\otimes (\theta q_{\sigma(n)})$ appears exactly $(N-n)!$ times with the same weight in~\eqref{eq:geom_localization_tensor}. We can thus write it in the better form
$$\bQ_{|\theta}=\bigoplus_{n=0}^N\frac1{n!}\sum_{j_1\neq\cdots \neq j_n}(\theta q_{j_1})\otimes\cdots\otimes (\theta q_{j_n})\prod_{k\notin\{j_1,...,j_n\}}\left(1-\int\theta q_{k}\right)$$
where now the terms all have disjoint supports. We obtain that the entropy equals
\begin{multline*}
\cS(\bQ_{|\theta})=-\sum_{n=0}^N\frac1{n!}\int_{\R^{dn}}\sum_{j_1\neq\cdots \neq j_n}(\theta q_{j_1})\otimes\cdots\otimes (\theta q_{j_n})\prod_{k\notin\{j_1,...,j_n\}}\left(1-\int\theta q_{k}\right)\times\\
\times \log \myt[\bigg]{ (\theta q_{j_1})\otimes\cdots\otimes (\theta q_{j_n})\prod_{k\notin\{i_1,...,i_n\}}\left(1-\int\theta q_{k}\right) }.
\end{multline*}
Note that the $n!$ in the logarithm simplifies with the $1/n!$. Expanding the logarithm and collecting the terms we obtain the claimed formula.
\end{proof}

In our case, we deduce by concavity that
\begin{align*}
-\cS(\tilde\bP_{|\theta})&\leq-\frac1{\ell^d}\int_{C_\ell}\cS\big(\bP(\cdot-\tau)_{|\theta}\big)\,\rd\tau\\
&\leq\frac1{\ell^d}\sum_{j,k}\int_{C_\ell}\int_{\R^d}\theta(x)\chi_j^k(x-\tau)\log\big(\theta(x)\chi_j^k(x-\tau) \big)\,\rd\tau\,\rd x\\
&=\frac1{\ell^d}\sum_{j,k}\int_{C_\ell}\int_{\R^d}\theta(x)\chi_j^k(x-\tau)\log\frac{\rho(x)\chi_j^k(x-\tau)}{(1+2\eta)^{-d}r_0^{-d}\rho_c(d)}\,\rd\tau\,\rd x.
\end{align*}
We estimate $\chi_j^k$ in the logarithm by its supremum $\|\chi_j^k\|_\ii=\frac{2^d}{(r_0\eta)^d|B_1|}$ and use that
$$\frac{\theta(x)}{\ell^d}\sum_{j,k}\int_{C_\ell}\chi_j^k(x-\tau)\,\rd\tau=\theta(x)\tilde \rho(x)=\rho(x).$$
We obtain
\begin{align*}
-\cS(\tilde\bP_{|\theta})&\leq\log\left(\frac{(1+2\eta)^{d}}{\eta^dv_c(d)}\right)\int\rho+\int\rho\log\rho\\
&=\log\left(\frac{2^d}{\eps^dv_c(d)}\right)\int\rho+\int\rho\log\rho.
\end{align*}
On the other hand, the energy bound~\eqref{eq:simple_hard_core_T0} applies since we still have $|x_j-x_k|\geq r_0$ on the support of the localized state $\tilde\bP_{|\theta}$. This concludes the proof of Theorem~\ref{thm:hard-core2}.\qed

\bigskip

\paragraph*{\textbf{Data availability.}} Data sharing not applicable to this article as no datasets were generated or analysed during the current study.

\medskip

\paragraph*{\textbf{Conflict of interest.}} The authors declare that they have no conflict of interest.

\newcommand{\etalchar}[1]{$^{#1}$}


\begin{thebibliography}{GGRSB99}

\bibitem[{Bau}87]{Baus-87}
{\sc M.~{Baus}}, {\em {Statistical mechanical theories of freezing: An
  overview}}, Journal of Statistical Physics, 48 (1987), pp.~1129--1146.

\bibitem[BCDP18]{ButChaPas-18}
{\sc G.~Buttazzo, T.~Champion, and L.~De~Pascale}, {\em Continuity and
  estimates for multimarginal optimal transportation problems with singular
  costs}, Appl. Math. Optim., 78 (2018), pp.~185--200.

\bibitem[BH77]{BraHep-77}
{\sc W.~Braun and K.~Hepp}, {\em The {V}lasov dynamics and its fluctuations in
  the {$1/N$} limit of interacting classical particles}, Comm. Math. Phys., 56
  (1977), pp.~101--113.

\bibitem[BL91]{BauLut-91}
{\sc M.~Baus and J.~F. Lutsko}, {\em Statistical mechanical theories of
  freezing: Where do we stand?}, Physica A: Statistical Mechanics and its
  Applications, 176 (1991), pp.~28 -- 36.

\bibitem[BL05]{ButLeb-05}
{\sc P.~Butt{\`a} and J.~L. Lebowitz}, {\em Local mean field models of uniform
  to nonuniform density fluid--crystal transitions}, The Journal of Physical
  Chemistry B, 109 (2005), pp.~6849--6854.

\bibitem[CC84]{ChaCha-84}
{\sc J.~T. Chayes and L.~Chayes}, {\em On the validity of the inverse
  conjecture in classical density functional theory}, J. Statist. Phys., 36
  (1984), pp.~471--488.

\bibitem[CCL84]{ChaChaLie-84}
{\sc J.~T. Chayes, L.~Chayes, and E.~H. Lieb}, {\em The inverse problem in
  classical statistical mechanics}, Comm. Math. Phys., 93 (1984), pp.~57--121.

\bibitem[CDD15]{ColPasMar-15}
{\sc M.~Colombo, L.~{De Pascale}, and S.~{Di Marino}}, {\em Multimarginal
  optimal transport maps for one-dimensional repulsive costs}, Canad. J. Math.,
  67 (2015), pp.~350--368.

\bibitem[CDMS19]{ColMarStr-19}
{\sc M.~Colombo, S.~Di~Marino, and F.~Stra}, {\em Continuity of multimarginal
  optimal transport with repulsive cost}, SIAM J. Math. Anal., 51 (2019),
  pp.~2903--2926.

\bibitem[CDPS17]{CarDuvPeySch-17}
{\sc G.~Carlier, V.~Duval, G.~Peyr\'{e}, and B.~Schmitzer}, {\em Convergence of
  entropic schemes for optimal transport and gradient flows}, SIAM J. Math.
  Anal., 49 (2017), pp.~1385--1418.

\bibitem[CFK13]{CotFriKlu-13}
{\sc C.~Cotar, G.~Friesecke, and C.~Kl{\"u}ppelberg}, {\em Density functional
  theory and optimal transportation with {C}oulomb cost}, Comm. Pure Appl.
  Math., 66 (2013), pp.~548--599.

\bibitem[CFP15]{CotFriPas-15}
{\sc C.~Cotar, G.~Friesecke, and B.~Pass}, {\em Infinite-body optimal transport
  with {C}oulomb cost}, Calc. Var. Partial Differ. Equ., 54 (2015),
  pp.~717--742.

\bibitem[CLLF]{CanLinLiuFri-22_ppt}
{\sc {\'E}.~Canc{\`e}s, L.~Lin, J.~Liu, and G.~Friesecke}, eds., {\em Density
  Functional Theory}, vol.~1 of Springer series on Molecular Modeling and
  Simulation, Springer.
\newblock In preparation.

\bibitem[CLMP92]{CagLioMarPul-92}
{\sc E.~Caglioti, P.-L. Lions, C.~Marchioro, and M.~Pulvirenti}, {\em A special
  class of stationary flows for two-dimensional {E}uler equations: a
  statistical mechanics description}, Comm. Math. Phys., 143 (1992),
  pp.~501--525.

\bibitem[CLMP95]{CagLioMarPul-95}
\leavevmode\vrule height 2pt depth -1.6pt width 23pt, {\em A special class of
  stationary flows for two-dimensional {E}uler equations: a statistical
  mechanics description. {II}}, Comm. Math. Phys., 174 (1995), pp.~229--260.

\bibitem[Coh17]{Cohn-17}
{\sc H.~Cohn}, {\em A conceptual breakthrough in sphere packing}, Notices Amer.
  Math. Soc., 64 (2017), pp.~102--115.

\bibitem[CP19a]{CotPet-19b}
{\sc C.~{Cotar} and M.~{Petrache}}, {\em Equality of the jellium and uniform
  electron gas next-order asymptotic terms for {C}oulomb and {R}iesz
  potentials}, ArXiv e-prints 1707.07664 (version 5),  (2019).

\bibitem[CP19b]{CotPet-19}
{\sc C.~Cotar and M.~Petrache}, {\em Next-order asymptotic expansion for
  {$N$}-marginal optimal transport with {C}oulomb and {R}iesz costs}, Adv.
  Math., 344 (2019), pp.~137--233.

\bibitem[DDM64]{DomMar-64}
{\sc C.~De~Dominicis and P.~C. Martin}, {\em Stationary entropy principle and
  renormalization in normal and superfluid systems. {I}. {A}lgebraic
  formulation}, J. Mathematical Phys., 5 (1964), pp.~14--30.

\bibitem[{De }62]{DeDominicis-62}
{\sc C.~{De Dominicis}}, {\em Variational formulations of equilibrium
  statistical mechanics}, J. Math. Phys., 3 (1962), pp.~983--1002.

\bibitem[dG75]{Guzman}
{\sc M.~de~Guzm\'{a}n}, {\em Differentiation of integrals in {$R^{n}$}},
  Lecture Notes in Mathematics, Vol. 481, Springer-Verlag, Berlin-New York,
  1975.
\newblock With appendices by Antonio C\'{o}rdoba, and Robert Fefferman, and two
  by Roberto Moriy\'{o}n.

\bibitem[DG90]{DreGro-90}
{\sc R.~Dreizler and E.~Gross}, {\em Density functional theory}, Springer,
  Berlin, 1990.

\bibitem[DGN17]{MarGerNen-17}
{\sc S.~{Di Marino}, A.~Gerolin, and L.~Nenna}, {\em Optimal Transportation
  Theory with Repulsive Costs}, vol.~``Topological Optimization and Optimal
  Transport in the Applied Sciences'' of Radon Series on Computational and
  Applied Mathematics, De Gruyter, June 2017, ch.~9, pp.~204--256.

\bibitem[DLN22]{MarLewNen-22_ppt}
{\sc S.~{Di Marino}, M.~Lewin, and L.~Nenna}, {\em Grand-canonical optimal
  transport}, ArXiv e-prints,  (2022).

\bibitem[DM67]{DobMin-67}
{\sc R.~L. Dobru\v{s}in and R.~A. Minlos}, {\em Existence and continuity of
  pressure in classical statistical physics}, Teor. Verojatnost. i Primenen.,
  12 (1967), pp.~595--618.

\bibitem[Dob64]{Dobrushin-64}
{\sc R.~L. Dobru\v{s}in}, {\em Investigation of conditions for the asymptotic
  existence of the configuration integral of {G}ibbs' distribution}, Theory of
  Probability \& Its Applications, 9 (1964), pp.~566--581.

\bibitem[DP15]{Pascale-15}
{\sc L.~De~Pascale}, {\em Optimal transport with {C}oulomb cost.
  {A}pproximation and duality}, ESAIM Math. Model. Numer. Anal., 49 (2015),
  pp.~1643--1657.

\bibitem[ED11]{EngDre-11}
{\sc E.~Engel and R.~Dreizler}, {\em Density Functional Theory: An Advanced
  Course}, Theoretical and Mathematical Physics, Springer, 2011.

\bibitem[EORK16]{EvaOetRotKah-16}
{\sc R.~Evans, M.~Oettel, R.~Roth, and G.~Kahl}, {\em New developments in
  classical density functional theory}, Journal of Physics: Condensed Matter,
  28 (2016), p.~240401.

\bibitem[EP79]{EbnPun-76}
{\sc C.~Ebner and C.~Punyanitya}, {\em Density-functional theory of simple
  classical fluids. ii. localized excess electron states}, Phys. Rev. A, 19
  (1979), pp.~856--865.

\bibitem[ESS76]{EbnSaaStr-76}
{\sc C.~Ebner, W.~F. Saam, and D.~Stroud}, {\em Density-functional theory of
  simple classical fluids. i. surfaces}, Phys. Rev. A, 14 (1976),
  pp.~2264--2273.

\bibitem[Eva79]{Evans-79}
{\sc R.~Evans}, {\em The nature of the liquid-vapour interface and other topics
  in the statistical mechanics of non-uniform, classical fluids}, Advances in
  Physics, 28 (1979), pp.~143--200.

\bibitem[Eva92]{Evans-92}
{\sc R.~Evans}, {\em Density Functionals in the Theory of Nonuniform Fluids},
  Marcel Dekker, Inc., 1992, pp.~85--176.

\bibitem[FLS18]{FouLewSol-18}
{\sc S.~{Fournais}, M.~{Lewin}, and J.~P. {Solovej}}, {\em The semi-classical
  limit of large fermionic systems}, Calc. Var. Partial Differ. Equ.,  (2018),
  pp.~57--105.

\bibitem[FLW22]{FraLapWei-LT}
{\sc R.~Frank, A.~Laptev, and T.~Weidl}, {\em {Schr\"odinger operators:
  Eigenvalues and Lieb-Thirring inequalities}}, Cambridge Studies in Advanced
  Mathematics, Cambridge University Press, Cambridge, 2022.

\bibitem[Gat72]{Gates-72}
{\sc D.~Gates}, {\em Rigorous results in the mean-field theory of freezing},
  Annals of Physics, 71 (1972), pp.~395 -- 420.

\bibitem[Geo11]{Georgii-11}
{\sc H.-O. Georgii}, {\em Gibbs measures and phase transitions}, vol.~9 of De
  Gruyter Studies in Mathematics, Walter de Gruyter \& Co., Berlin, second~ed.,
  2011.

\bibitem[GGRSB99]{GelGubRadSli-99}
{\sc L.~D. Gelb, K.~E. Gubbins, R.~Radhakrishnan, and M.~Sliwinska-Bartkowiak},
  {\em Phase separation in confined systems}, Rep. Prog. Phys., 62 (1999),
  pp.~1573--1659.

\bibitem[GK76]{GreKle-76}
{\sc N.~Grewe and W.~Klein}, {\em Rigorous derivation of the
  {K}irkwood-{M}onroe equation for small activity}, J. Mathematical Phys., 17
  (1976), pp.~699--703.

\bibitem[GL17]{GigLeo-17}
{\sc G.~Gigante and P.~Leopardi}, {\em Diameter bounded equal measure
  partitions of {A}hlfors regular metric measure spaces}, Discrete Comput.
  Geom., 57 (2017), pp.~419--430.

\bibitem[GP69]{GatPen-69}
{\sc D.~J. Gates and O.~Penrose}, {\em The van der {W}aals limit for classical
  systems. {I}. {A} variational principle}, Comm. Math. Phys., 15 (1969),
  pp.~255--276.

\bibitem[HK64]{HohKoh-64}
{\sc P.~Hohenberg and W.~Kohn}, {\em Inhomogeneous electron gas}, Phys. Rev.,
  136 (1964), pp.~B864--B871.

\bibitem[HLS09]{HaiLewSol_2-09}
{\sc C.~Hainzl, M.~Lewin, and J.~P. Solovej}, {\em The thermodynamic limit of
  quantum {C}oulomb systems. {P}art {II}. {A}pplications}, Advances in Math.,
  221 (2009), pp.~488--546.

\bibitem[HM90]{HanMcDon-90}
{\sc J.-P. Hansen and I.~R. McDonald}, {\em Theory of simple liquids},
  Elsevier, 1990.

\bibitem[HM13]{HanMcD-13}
{\sc J.-P. Hansen and I.~McDonald}, {\em Theory of Simple Liquids (With
  Applications to Soft Matter)}, Academic Press, 2013.

\bibitem[HO81]{HayOxt-81}
{\sc A.~Haymet and D.~W. Oxtoby}, {\em A molecular theory for the solid--liquid
  interface}, The Journal of Chemical Physics, 74 (1981), pp.~2559--2565.

\bibitem[JKT22]{JanKunTsa-22}
{\sc S.~Jansen, T.~Kuna, and D.~Tsagkarogiannis}, {\em Virial inversion and
  density functionals}, J. Funct. Anal.,  (2022).
\newblock online first.

\bibitem[JLM23]{JexLewMad-23b_ppt}
{\sc M.~Jex, M.~Lewin, and P.~Madsen}, {\em Classical {D}ensity {F}unctional
  {T}heory: {T}he {L}ocal {D}ensity {A}pproximation}, in preparation,  (2023).

\bibitem[Kel84]{Kellerer-84}
{\sc H.~G. Kellerer}, {\em Duality theorems for marginal problems}, Z. Wahrsch.
  Verw. Gebiete, 67 (1984), pp.~399--432.

\bibitem[Kie89]{Kiessling-89}
{\sc M.~K.~H. Kiessling}, {\em On the equilibrium statistical mechanics of
  isothermal classical self-gravitating matter}, J. Statist. Phys., 55 (1989),
  pp.~203--257.

\bibitem[Kie93]{Kiessling-93}
{\sc M.~K.-H. Kiessling}, {\em Statistical mechanics of classical particles
  with logarithmic interactions}, Comm. Pure. Appl. Math., 46 (1993),
  pp.~27--56.

\bibitem[Kie09]{Kiessling-09b}
\leavevmode\vrule height 2pt depth -1.6pt width 23pt, {\em The {V}lasov
  continuum limit for the classical microcanonical ensemble}, Rev. Math. Phys.,
  21 (2009), pp.~1145--1195.

\bibitem[KM41]{KirMon-41}
{\sc J.~G. Kirkwood and E.~Monroe}, {\em Statistical mechanics of fusion}, J.
  Chem. Phys., 9 (1941), pp.~514--526.

\bibitem[KP95]{KiePer-95}
{\sc M.~K.-H. Kiessling and J.~K. Percus}, {\em Nonuniform van der {W}aals
  theory}, J. Statist. Phys., 78 (1995), pp.~1337--1376.

\bibitem[KS65]{KohSha-65}
{\sc W.~Kohn and L.~J. Sham}, {\em Self-consistent equations including exchange
  and correlation effects}, Phys. Rev. (2), 140 (1965), pp.~A1133--A1138.

\bibitem[Lew11]{Lewin-11}
{\sc M.~Lewin}, {\em Geometric methods for nonlinear many-body quantum
  systems}, J. Funct. Anal., 260 (2011), pp.~3535--3595.

\bibitem[Lew22]{Lewin-22}
\leavevmode\vrule height 2pt depth -1.6pt width 23pt, {\em Coulomb and {R}iesz
  gases: {T}he known and the unknown}, J. Math. Phys., 63 (2022), p.~061101.
\newblock Special collection in honor of Freeman Dyson.

\bibitem[LGN13]{LanGorNei-13}
{\sc J.~Landers, G.~Y. Gor, and A.~V. Neimark}, {\em Density functional theory
  methods for characterization of porous materials}, Colloids and Surfaces A:
  Physicochemical and Engineering Aspects, 437 (2013), pp.~3--32.
\newblock Characterization of Porous Materials: From Angstroms to Millimeters A
  Collection of Selected Papers Presented at the 6th International Workshop,
  CPM-6 April 30 – May 2nd, 2012, Delray Beach, FL, USA Co-sponsored by
  Quantachrome Instruments.

\bibitem[Lie83]{Lieb-83b}
{\sc E.~H. Lieb}, {\em Density functionals for {C}oulomb systems}, Int. J.
  Quantum Chem., 24 (1983), pp.~243--277.

\bibitem[LLS18]{LewLieSei-18}
{\sc M.~Lewin, E.~H. Lieb, and R.~Seiringer}, {\em Statistical mechanics of the
  {U}niform {E}lectron {G}as}, J. {\'E}c. polytech. Math., 5 (2018),
  pp.~79--116.

\bibitem[LLS19a]{LewLieSei-19b}
\leavevmode\vrule height 2pt depth -1.6pt width 23pt, {\em Floating {W}igner
  crystal with no boundary charge fluctuations}, Phys. Rev. B, 100 (2019),
  p.~035127.

\bibitem[LLS19b]{LewLieSei-19_ppt}
\leavevmode\vrule height 2pt depth -1.6pt width 23pt, {\em {Universal
  Functionals in Density Functional Theory}}, ArXiv e-prints,  (2019).
\newblock Chapter in the book ``Density Functional Theory --- Modeling,
  Mathematical Analysis, Computational Methods, and Applications" edited by
  \'Eric Canc\`es and Gero Friesecke, Springer.

\bibitem[LLS20]{LewLieSei-20}
\leavevmode\vrule height 2pt depth -1.6pt width 23pt, {\em The {L}ocal
  {D}ensity {A}pproximation in {D}ensity {F}unctional {T}heory}, Pure Appl.
  Anal., 2 (2020), pp.~35--73.

\bibitem[LNR16]{LewNamRou-16c}
{\sc M.~Lewin, P.~T. Nam, and N.~Rougerie}, {\em The mean-field approximation
  and the non-linear {S}chr{\"o}dinger functional for trapped {B}ose gases},
  Trans. Amer. Math. Soc, 368 (2016), pp.~6131--6157.

\bibitem[LO80]{LieOxf-80}
{\sc E.~H. Lieb and S.~Oxford}, {\em Improved lower bound on the indirect
  {C}oulomb energy}, Int. J. Quantum Chem., 19 (1980), pp.~427--439.

\bibitem[L{\"o}w02]{Lowen-02}
{\sc H.~L{\"o}wen}, {\em Density functional theory of inhomogeneous classical
  fluids: recent developments and new perspectives}, Journal of Physics:
  Condensed Matter, 14 (2002), pp.~11897--11905.

\bibitem[LP63]{LebPer-63}
{\sc J.~L. Lebowitz and J.~K. Percus}, {\em Statistical thermodynamics of
  nonuniform fluids}, J. Math. Phys., 4 (1963), pp.~116--123.

\bibitem[LP66]{LebPen-66}
{\sc J.~L. Lebowitz and O.~Penrose}, {\em Rigorous treatment of the van der
  {W}aals-{M}axwell theory of the liquid-vapor transition}, J. Mathematical
  Phys., 7 (1966), pp.~98--113.

\bibitem[LP93]{LevPer-93}
{\sc M.~Levy and J.~P. Perdew}, {\em Tight bound and convexity constraint on
  the exchange-correlation-energy functional in the low-density limit, and
  other formal tests of generalized-gradient approximations}, Phys. Rev. B, 48
  (1993), pp.~11638--11645.

\bibitem[MH61]{MorHir-61}
{\sc T.~Morita and K.~Hiroike}, {\em {A New Approach to the Theory of Classical
  Fluids. III: General Treatment of Classical Systems}}, Progress of
  Theoretical Physics, 25 (1961), pp.~537--578.

\bibitem[Mie20]{Mietzsch-20}
{\sc N.~Mietzsch}, {\em The validity of the local density approximation for
  smooth short range interaction potentials}, J. Math. Phys., 61 (2020),
  p.~113503.

\bibitem[MS82]{MesSpo-82}
{\sc J.~Messer and H.~Spohn}, {\em Statistical mechanics of the isothermal
  {L}ane-{E}mden equation}, J. Statist. Phys., 29 (1982), pp.~561--578.

\bibitem[Pas15]{Pass-15}
{\sc B.~Pass}, {\em Multi-marginal optimal transport: theory and applications},
  ESAIM Math. Model. Numer. Anal., 49 (2015), pp.~1771--1790.

\bibitem[PB06]{PliBer-06}
{\sc M.~Plischke and B.~Bergersen}, {\em Equilibrium statistical physics},
  World Scientific Publishing Company, 2006.

\bibitem[PBE96]{PerBurErn-96}
{\sc J.~P. Perdew, K.~Burke, and M.~Ernzerhof}, {\em Generalized gradient
  approximation made simple}, Phys. Rev. Lett., 77 (1996), pp.~3865--3868.

\bibitem[Per64]{Percus-64}
{\sc J.~Percus}, Frontiers in physics, W.A. Benjamin, INC, 1964, p.~II.33.

\bibitem[Per76]{Percus-76}
{\sc J.~K. Percus}, {\em Equilibrium state of a classical fluid of hard rods in
  an external field}, J. Statist. Phys., 15 (1976), pp.~505--511.

\bibitem[Per82a]{Percus-82b}
\leavevmode\vrule height 2pt depth -1.6pt width 23pt, {\em Nonuniform fluids in
  the grand canonical ensemble}, International Journal of Quantum Chemistry, 22
  (1982), pp.~33--48.

\bibitem[Per82b]{Percus-82}
\leavevmode\vrule height 2pt depth -1.6pt width 23pt, {\em One-dimensional
  classical fluid with nearest-neighbor interaction in arbitrary external
  field}, J. Statist. Phys., 28 (1982), pp.~67--81.

\bibitem[Per91]{Perdew-91}
{\sc J.~P. Perdew}, {\em Unified {T}heory of {E}xchange and {C}orrelation
  {B}eyond the {L}ocal {D}ensity {A}pproximation}, in Electronic Structure of
  Solids '91, P.~Ziesche and H.~Eschrig, eds., Akademie Verlag, Berlin, 1991,
  pp.~11--20.

\bibitem[Per97]{Percus-97}
{\sc J.~K. Percus}, {\em Nonuniform classical fluid mixture in one-dimensional
  space with next neighbor interactions}, vol.~89, 1997, pp.~249--272.
\newblock Dedicated to Bernard Jancovici.

\bibitem[PR07]{PetReb-07}
{\sc S.~N. Petrenko and A.~L. Rebenko}, {\em Superstable criterion and
  superstable bounds for infinite range interaction. {I}. {T}wo-body
  potentials}, Methods Funct. Anal. Topology, 13 (2007), pp.~50--61.

\bibitem[PS22]{PerSun-22}
{\sc J.~Perdew and J.~Sun}, {\em The {L}ieb-{O}xford Lower Bounds on the
  {C}oulomb Energy, Their Importance to Electron Density Functional Theory, and
  a Conjectured Tight Bound on Exchange}, EMS Press, 2022, ch.~36,
  pp.~165--178.

\bibitem[PY94]{ParYan-94}
{\sc R.~Parr and W.~Yang}, {\em Density-Functional Theory of Atoms and
  Molecules}, International Series of Monographs on Chemistry, Oxford
  University Press, USA, 1994.

\bibitem[Reb98]{Rebenko-98}
{\sc A.~L. Rebenko}, {\em A new proof of {R}uelle's superstability bounds}, J.
  Statist. Phys., 91 (1998), pp.~815--826.

\bibitem[Rot10]{Roth-10}
{\sc R.~Roth}, {\em Fundamental measure theory for hard-sphere mixtures: a
  review}, Journal of Physics: Condensed Matter, 22 (2010), p.~063102.

\bibitem[{Rou}15]{Rougerie-LMU}
{\sc N.~{Rougerie}}, {\em De {F}inetti theorems, mean-field limits and
  {B}ose-{E}instein condensation}, ArXiv e-prints,  (2015).

\bibitem[Roz71]{Rozenbljum-71}
{\sc G.~V. Rozenbljum}, {\em The distribution of the eigenvalues of the first
  boundary value problem in unbounded domains}, Dokl. Akad. Nauk SSSR, 200
  (1971), pp.~1034--1036.

\bibitem[Roz72]{Rozenbljum-72}
{\sc G.~V. Rozenblum}, {\em Distribution of the discrete spectrum of singular
  differential operators}, Dokl. Akad. Nauk SSSR, 202 (1972), pp.~1012--1015.

\bibitem[Rue70]{Ruelle-70}
{\sc D.~Ruelle}, {\em Superstable interactions in classical statistical
  mechanics}, Comm. Math. Phys., 18 (1970), pp.~127--159.

\bibitem[Rue99]{Ruelle}
\leavevmode\vrule height 2pt depth -1.6pt width 23pt, {\em Statistical
  mechanics. Rigorous results}, {Singapore: World Scientific. London: Imperial
  College Press}, 1999.

\bibitem[RV81]{RobVar-81}
{\sc A.~Robledo and C.~Varea}, {\em On the relationship between the density
  functional formalism and the potential distribution theory for nonuniform
  fluids}, J. Stat. Phys., 26 (1981), pp.~513--525.

\bibitem[RY77]{RamYus-77}
{\sc T.~Ramakrishnan and M.~Yussouff}, {\em Theory of the liquid-solid
  transition}, Solid State Communications, 21 (1977), pp.~389 -- 392.

\bibitem[RY79]{YamYus-79}
{\sc T.~V. Ramakrishnan and M.~Yussouff}, {\em First-principles order-parameter
  theory of freezing}, Phys. Rev. B, 19 (1979), pp.~2775--2794.

\bibitem[SB62]{StiBuf-62}
{\sc F.~H. Stillinger and F.~P. Buff}, {\em Equilibrium statistical mechanics
  of inhomogeneous fluids}, J. Chem. Phys., 37 (1962), pp.~1--12.

\bibitem[SDG{\etalchar{+}}17]{SeiMarGerNenGieGor-17}
{\sc M.~{Seidl}, S.~{Di Marino}, A.~{Gerolin}, L.~{Nenna}, K.~J.~H.
  {Giesbertz}, and P.~{Gori-Giorgi}}, {\em The strictly-correlated electron
  functional for spherically symmetric systems revisited}, ArXiv e-prints,
  (2017).

\bibitem[SE77]{SaaEbn-77}
{\sc W.~F. Saam and C.~Ebner}, {\em Density-functional theory of classical
  systems}, Phys. Rev. A, 15 (1977), pp.~2566--2568.

\bibitem[Sin91]{Singh-91}
{\sc Y.~Singh}, {\em Density-functional theory of freezing and properties of
  the ordered phase}, Physics Reports, 207 (1991), pp.~351 -- 444.

\bibitem[Spo81]{Spohn-81}
{\sc H.~Spohn}, {\em On the {V}lasov hierarchy}, Math. Methods Appl. Sci., 3
  (1981), pp.~445--455.

\bibitem[SPR15]{SunPerRuz-15}
{\sc J.~{Sun}, J.~P. {Perdew}, and A.~{Ruzsinszky}}, {\em Semilocal density
  functional obeying a strongly tightened bound for exchange}, Proc. Nat. Acad.
  Sci. U.S.A., 112 (2015), pp.~685--689.

\bibitem[SRP15]{SunRuzPer-15}
{\sc J.~Sun, A.~Ruzsinszky, and J.~P. Perdew}, {\em Strongly constrained and
  appropriately normed semilocal density functional}, Phys. Rev. Lett., 115
  (2015), p.~036402.

\bibitem[SRZ{\etalchar{+}}16]{Perdew_etal-16}
{\sc J.~Sun, R.~C. Remsing, Y.~Zhang, Z.~Sun, A.~Ruzsinszky, H.~Peng, Z.~Yang,
  A.~Paul, U.~Waghmare, X.~Wu, M.~L. Klein, and J.~P. Perdew}, {\em Accurate
  first-principles structures and energies of diversely bonded systems from an
  efficient density functional}, Nature Chemistry, 8 (2016), pp.~831--836.

\bibitem[Ste64]{Stell-64}
{\sc G.~Stell}, Frontiers in physics, W.A. Benjamin, INC, 1964, p.~II.171.

\bibitem[TPSS03]{TaoPerStaScu-03}
{\sc J.~Tao, J.~P. Perdew, V.~N. Staroverov, and G.~E. Scuseria}, {\em Climbing
  the density functional ladder: Nonempirical meta--generalized gradient
  approximation designed for molecules and solids}, Phys. Rev. Lett., 91
  (2003), p.~146401.

\bibitem[TS06]{TorSti-06}
{\sc S.~Torquato and F.~H. Stillinger}, {\em New conjectural lower bounds on
  the optimal density of sphere packings}, Experimental Mathematics, 15 (2006),
  pp.~307--331.

\bibitem[VDP89]{VanDavPer-89}
{\sc T.~K. Vanderlick, H.~T. Davis, and J.~K. Percus}, {\em The statistical
  mechanics of inhomogeneous hard rod mixtures}, The Journal of Chemical
  Physics, 91 (1989), pp.~7136--7145.

\bibitem[Via21]{Viazovska-21}
{\sc M.~Viazovska}, {\em Almost impossible ${E}_{8}$ and {L}eech lattices}, EMS
  Magazine, 121 (2021).

\bibitem[Wei96]{Weidl-96}
{\sc T.~Weidl}, {\em On the {L}ieb-{T}hirring constants {$L_{\gamma,1}$} for
  {$\gamma\geq 1/2$}}, Comm. Math. Phys., 178 (1996), pp.~135--146.

\bibitem[Wu06]{Wu-06}
{\sc J.~Wu}, {\em Density functional theory for chemical engineering: From
  capillarity to soft materials}, AIChE Journal, 52 (2006), pp.~1169--1193.

\bibitem[YFG76]{YanFleGib-76}
{\sc A.~J.~M. Yang, P.~D. Fleming, and J.~H. Gibbs}, {\em Molecular theory of
  surface tension}, J. Chem. Phys., 64 (1976), pp.~3732--3747.

\end{thebibliography}

\end{document}